\documentclass[12pt,twoside,fleqn]{article}
  \newdimen\paravsp  \paravsp=1.3ex
\usepackage{latexsym,amsmath,amssymb}
\usepackage{ntheorem}
\usepackage{tikz}
\def\beq{\begin{equation}}    \def\eeq{\end{equation}}
\def\beqn{\begin{displaymath}}\def\eeqn{\end{displaymath}}
\newtheorem{theorem}{Theorem}
\newtheorem{corollary}[theorem]{Corollary}

\newtheorem{definition}[theorem]{Definition}
\newtheorem{example}[theorem]{Example}
\newtheorem{proposition}[theorem]{Proposition}

\newtheorem{axiom}[theorem]{Axiom}
\newtheorem{rule1}[theorem]{Rule}
\def\paradot#1{\paragraph{#1.}}
\newenvironment{keyword}{\centerline{\bf\small
Keywords}\begin{quote}\small}{\par\end{quote}\vskip 1ex}
\newenvironment{proof}{{\vskip 1ex \noindent\bf Proof.}}{\vskip 1ex}

\def\email#1{\rm\texttt{#1}}
\def\req#1{(\ref{#1})}
\def\nq{\hspace{-1em}}
\def\qed{\hspace*{\fill}\rule{1.4ex}{1.4ex}$\quad$\\}
\def\eoe{\hspace*{\fill} $\diamondsuit\quad$\\} 
\def\fr#1#2{{\textstyle{#1\over#2}}}
\def\SetR{\mathbb{R}}
\def\SetN{\mathbb{N}}
\def\a{\alpha}                  
\def\eps{\varepsilon}           
\def\model{\text{\it mod}}      
\def\sepmodel{\widehat{\text\it mo\smash d}} 
\def\tfn#1{\underline{#1}}      
\def\mexp#1{\exp\bigg\{\!#1\bigg\}} 
\def\ind#1{[\![ #1 ]\!]}        
\def\ph{\varphi}                
\def\va{{\scriptstyle\cal V}}   
\def\B{{\cal B}}                
\def\I{{\cal I}}                
\def\S{{\cal S}}                
\def\E{{\mathbb E}}             
\def\Sat{\text{\it Sat}}        
\def\Nat{\text{\it Nat}}        
\def\tB{o}                      
\def\KL{\text{KL}}              
\def\prior{\xi}                 
\def\D{\mathcal{D}}             
\def\V{{\cal V}}                

\begin{document}

\title{
\vskip 2mm\bf\Large\hrule height5pt \vskip 4mm
Probabilities on Sentences in an Expressive Logic
\vskip 4mm \hrule height2pt}
\author{
\begin{minipage}{0.49\textwidth}\centering
{\bf Marcus Hutter}\\[2mm]
\normalsize Research School of Computer Science\\
The Australian National University\\
\email{marcus.hutter@anu.edu.au}
\end{minipage}\hfill
\begin{minipage}{0.49\textwidth}\centering
{\bf John W.~Lloyd}\\[2mm]
\normalsize Research School of Computer Science \\
The Australian National University \\
\email{john.lloyd@anu.edu.au}
\end{minipage}
\\[6ex]
\begin{minipage}{0.49\textwidth}\centering
{\bf Kee Siong Ng}\\[2mm]
\normalsize EMC Greenplum and\\
The Australian National University \\
\email{keesiong.ng@emc.com\\}
\end{minipage}
\begin{minipage}{0.49\textwidth}\centering
{\bf William T.\;B.~Uther}\\[2mm]
\normalsize National ICT Australia and \\
University of New South Wales \\
\email{william.uther@nicta.com.au}
\end{minipage}\hfill
\\[4ex]
}

\date{12 September 2012}

\maketitle

\begin{abstract}
\footnote{Presented at the Fifth Workshop on Combining
Probability and Logic (Progic 2011) in New York.}
%
Automated reasoning about uncertain knowledge has many applications.
One difficulty when developing such systems is the lack
of a completely satisfactory integration of logic and probability.
We address this problem directly.
%
Expressive languages like higher-order logic are ideally suited
for representing and reasoning about structured knowledge.
Uncertain knowledge can be modeled by using graded probabilities
rather than binary truth-values.
%
The main technical problem studied in this paper is the following:
Given a set of sentences, each having some probability of being true,
what probability should be ascribed to other (query) sentences?
%
A natural wish-list, among others, is that the probability distribution
(i) is consistent with the knowledge base, %
(ii) allows for a consistent inference procedure and in particular %
(iii) reduces to deductive logic in the limit of probabilities being 0 and 1, %
(iv) allows (Bayesian) inductive reasoning and %
(v) learning in the limit and in particular %
(vi) allows confirmation of universally quantified hypotheses/sentences.
%
We translate this wish-list into technical requirements for a prior probability
and show that probabilities satisfying all our criteria exist.
We also give explicit constructions and several general
characterizations of probabilities that satisfy some or all of
the criteria and various (counter) examples.
%
We also derive necessary and sufficient conditions for
extending beliefs about finitely many sentences to suitable
probabilities over all sentences,
and in particular least dogmatic or least biased ones.
%
We conclude with a brief outlook on how the developed theory might
be used and approximated in autonomous reasoning agents.
%
Our theory is a step towards a globally consistent and
empirically satisfactory unification of probability and logic.
\newpage
\def\contentsname{\centering\normalsize Contents}
{\parskip=-2.7ex\tableofcontents}
\end{abstract}


\begin{quote}\it
``The study of probability functions defined over the sentences
of a rich enough formal language yields interesting insights in
more than one direction.'' \par \hfill --- {\sl Haim Gaifman (1982)}
\end{quote}

\section{Introduction}

\paradot{Motivation}
Sophisticated computer applications generally require
expressive languages for knowledge representation and
reasoning. In particular, such languages need to be able to
represent both structured knowledge and uncertainty
\cite{nilsson86,halpern-uncertainty,muggleton95stochastic,DeRaedt-Kersting,richardson06,hajek,williamson02}.
%
A suitable language for this purpose is higher-order logic
\cite{church,henkin,andrews2,LogicforLearning,vanbenthem-doets,leivant,shapiroHOL},
which admits higher-order functions that can take
functions as arguments and/or return functions as results.
This facility is convenient for probabilistic modeling since
it means that theories can contain probability
densities \cite{farmer,pfeffer07,Church-language}.
In particular, many forms of probabilistic reasoning can be done
in higher-order logic using the traditional axiomatic method: a
theory can be written down which has the intended interpretation as a model
and then conventional proof and computation techniques can be
used to answer queries \cite{ng-lloyd-JAL,ng-lloyd-uther-AMAI}.
%
While such a computational approach is effective, it is sometimes
more natural to pose a problem as one where
the probability of some sentences in the
theory being true
may be strictly less than one and/or the query sentence (and its
negation) may not be a logical consequence of the theory. In
such cases, deductive reasoning
does not suffice for answering queries and it becomes necessary
to use probabilistic methods
\cite{paris-book,kersting07blp,richardson06,muggleton95stochastic,milch07}.

\paradot{Main aim}
These considerations lead to the main technical issue studied
in this paper:
\begin{quote}
  Given a set of sentences, each having some probability of being true,\\
  what probability should be ascribed to other (query) sentences?
\end{quote}

\noindent We build on the work of Gaifman
\cite{gaifman64} whose paper with Snir \cite{gaifman-snir}
develops a quite comprehensive theory of probabilities on sentences
in first-order Peano arithmetic.
We take up these ideas, using
non-dogmatic priors \cite{gaifman-snir} and additionally the
minimum relative entropy principle as in \cite{Williamson:08},
but for general theories and in a higher-order setting.
%
We concentrate on developing probabilities on sentences in a
higher-order logic. This sets the stage for
combining it with the probabilities inside sentences approach
\cite{ng-lloyd-JAL,ng-lloyd-uther-AMAI}.

\paradot{Summary of key concepts}
Section~\ref{Logic} introduces higher-order logic and its relevant properties.
We use the higher-order logic
(Definitions~\ref{def:type}, \ref{def:term}, and \ref{def:interpretation}) based on Church's simple theory of types
\cite{church,henkin,andrews2}.
We employ the Henkin semantics and make use of a particular class of
interpretations, called separating interpretations (Definition~\ref{def:separating}).

Section~\ref{Probabilities on Sentences} gives the definition
of probabilities on sentences in higher-order logic
(Definition~\ref{def:prob_phi}), introduces the Gaifman
condition, and develops some basic properties of such
probabilities.
Section~\ref{Probabilities on Interpretations} then introduces
probabilities on interpretations and shows their close
connection with probabilities on sentences.
Gaifman \cite{gaifman64} (generalized in
Definition~\ref{def:Gaifman} and Propositions~\ref{prop:Gaifman},~\ref{prop:LCA},~\ref{prop:GCA}) introduced a condition,
called Gaifman in \cite{scott-krauss}, that connects
probabilities of quantified sentences to limits of
probabilities of finite conjunctions.
In our case, it effectively restricts probabilities to
separating interpretations while maintaining countable
additivity.

While generally accepted in probability theory
(Definition~\ref{def:probmeasure}), some circles argue that countable additivity (CA)
does not have a good philosophical justification, and/or that
it is not needed since real experience is always finite, hence
only non-asymptotic statements are of practical relevance, for
which CA is not needed. On the other hand, it is usually much
easier to first obtain asymptotic statements which requires CA,
and then improve upon them. Furthermore we will show that CA
can guide us in the right direction to find good finitary prior
probabilities.

Another principle which has received much less attention than CA
but is equally if not more important is that of Cournot
\cite{Cournot:1843,Shafer:06}: An event of probability (close
to) zero singled out in advance is physically impossible; or
conversely, an event of probability 1 will physically happen for sure.
In short: zero probability means impossibility. The history of
the semantics of probability is stony \cite{Fine:73}. Cournot's
``forgotten'' principle is one way of giving meaning to
probabilistic statements like, ``the relative frequency of
heads of a fair coin converges to 1/2 {\em with probability
1}''. The contraposition of Cournot is that one must assign
non-zero probability to possible events. If ``events'' are
described by sentences and ``possible'' means it is possible to
satisfy these sentences, i.e. they possess a model, then
we arrive at the strong Cournot principle that satisfiable sentences should be
assigned non-zero probability (Definitions~\ref{def:stronglyCournot}
and~\ref{def:stronglyCournotmustar}).
This condition has been
appropriately called `non-dogmatic' in \cite{gaifman-snir}. As
long as something is not proven false, there is a (small)
chance it is valid in the intended interpretation. This
non-dogmatism is crucial in Bayesian inductive reasoning, since
no evidence (however strong) can increase a zero prior belief
to a non-zero posterior belief \cite{Hutter:11uiphil}.
The Gaifman condition is inconsistent with the strong Cournot
principle (Example~\ref{ex:StronglyCournotNotGaifman}),
but consistent with a weaker version (Definition~\ref{def:Cournot}).
Probabilities that are Gaifman and (plain, not strong)
Cournot allow learning in the limit (Theorem~\ref{thm:CUH} and
Corollary~\ref{cor:LIL}). 

A standard way to construct (general / Cournot / Gaifman)
probabilities on sentences is to construct (general /
non-dogmatic / separating) probabilities on
interpretations, and then transfer them to sentences
(Propositions~\ref{prop:mustartomu}, \ref{prop:mustartomuseparating},
and~\ref{prop:CournotMSeqM}).
At the same time we give model-theoretic characterizations
of the Gaifman condition (Corollary~\ref{cor:mueqmuseparating})
and the Cournot condition (Definition~\ref{def:Cournotmustar}).
In Section~\ref{Existence of Probabilities}, we give a
particularly simple construction of a probability that is
Cournot and Gaifman (Theorem~\ref{thm:CandG}) and a complete
characterization of general/Cournot/Gaifman probabilities
(Theorems~\ref{thm:charGandC} and \ref{thm:char} and
Corollary~\ref{cor:GandC}).
We also give various examples of (strong) (non)Cournot and/or
Gaifman probabilities and (non)separating interpretations for
countable domains (Examples~\ref{ex:SINat}, \ref{ex:NSINat},
and \ref{ex:iota}) and finite domains
(Examples~\ref{ex:NotGaifman},
\ref{ex:StronglyCournotNotGaifman}, \ref{ex:GaifmanNotCournot},
\ref{ex:CournotNotStronglyCournot}).

Section~\ref{sec:Extension} considers the
important practical situation of whether a real-valued function
on a set of sentences can be extended to a probability on all
sentences; a necessary and sufficient condition is given for
this, as is a method for determining such probabilities using
minimum relative entropy introduced in
Section~\ref{sec:entropy}.
Prior knowledge and data constrain our (belief) probabilities
in various ways, which  we need to take into account when
constructing probabilities. Prior knowledge is usually given in
the form of probabilities on sentences like ``the coin has head
probability 1/2'', or facts like ``all electrons have the same
charge'', or non-logical axioms like ``there are infinitely
many natural numbers''. They correspond to requiring their
probability to be 1/2, extremely close to 1, and 1,
respectively. It is therefore necessary to be able to go
from probabilities on sentences to probability on
interpretations (Proposition~\ref{prop:mutomustar}). This
allows us to prove various necessary and sufficient conditions
under which such partial probability specifications can be
completed and what properties they have
(Propositions~\ref{prop:ExtProb} and~\ref{prop:SAandE}). In
particular we show that hierarchical probabilistic knowledge
(Definitions~\ref{def:HierarchicalS}) is always
probabilistically consistent (Proposition~\ref{prop:ERC}).
%
Further, seldom does knowledge constrain the probability on all
sentences to be uniquely determined. In this case it is natural
to choose a probability that is least dogmatic or biased
\cite{nilsson86,Williamson:08}. The minimum relative entropy
(Definition~\ref{def:relentropy}) principle can be used to
construct such a unique minimally more
informative probability that is consistent with our prior
knowledge (Definition~\ref{def:mmimu} and
Propositions~\ref{prop:mmimu} and~\ref{prop:mmiGmu}).

Section~\ref{sec:Examples} is a brief outlook on how the
developed theory might be used and approximated in autonomous
reasoning agents. In particular, certain knowledge, learning in
the limit (\ref{cor:LIL}), the infamous black raven paradox,
and the Monty Hall problem are discussed, but only briefly.
The paper ends with a more detailed discussion in
Section~\ref{Discussion} of the broader context and motivation
of this work, as well as related results in the literature,
the outline of a framework for probabilistic reasoning and
modeling in higher-order logic, and future research
directions.

While some of the results presented in this paper are known in
the first-order case and their extension to the higher-order
case is straightforward, it nevertheless seems useful to
provide a survey of this material (with proofs included). Also,
many beautiful ideas in the long and technical paper by Gaifman
\cite{gaifman-snir} deserve wider attention than they have
received. We hope our exposition helps to rectify this
situation.

\section{Logic}\label{Logic}

We review here a standard formulation of higher-order logic \cite{andrews2}
that is based on Church's simple theory of types \cite{church}.
Other references on higher-order logic include
\cite{LogicforLearning,farmer,vanbenthem-doets,leivant,shapiroHOL}.
Some discussion of the interesting history of the simple theory of types
is given in \cite{andrews2,farmer}.

The best way to think about higher-order logic is that it is
the formalization of everyday informal mathematics: whatever
mathematical description one might give of some situation, the
formalization of that situation in higher-order logic is likely
to be a straightforward translation of the informal description.
In particular, higher-order logic provides a
suitable foundation for mathematics itself which has several
advantages over more traditional approaches that are based on
axiomatizing sets in first-order logic.
Furthermore, higher-order logic is the logical formalism of choice for much
of theoretical computer science and also applications areas
such as software and hardware verification.
For a convincing account of the advantages of higher-order over
first-order logic in computer science, see \cite{farmer}.

The logic presented here differs in a minor way from that in \cite{andrews2}
in that we omit the description operator $\iota$, for reasons that are discussed later.
All the results from \cite{andrews2} that are used here also hold for the logic
with $\iota$ omitted, by obvious changes to their proofs.
In addition the notation for the logic used here differs somewhat from that in \cite{andrews2},
but the correspondences will always be clear.
There are also a few differences in terminology here compared to \cite{andrews2}
that are noted along the way.

We begin with the definition of a type.

\begin{definition}[type $\a$]\label{def:type}
A {\em type} is defined inductively as follows.
\begin{enumerate}\itemsep1mm\parskip0mm
\item $\tB$ is a type.
\item $\imath$ is a type.
\item If $\alpha$ and $\beta$ are types, then $\alpha \to \beta$ is a type.
\end{enumerate}
\end{definition}

In this definition, $\tB$ is the type of the truth values, $\imath$ is the type of individuals,
and $\alpha \rightarrow \beta$ is the type of functions from elements of type $\alpha$ to elements of type $\beta$.
We use the convention that $\to$ is right associative. So, for
example, when we write $\a\to\beta\to\gamma\to\kappa$
we mean $\a\to(\beta\to(\gamma\to\kappa))$.
A {\em function type} is a type of the form $\alpha\to\beta$, for some $\alpha$ and $\beta$.

There is a denumerable list of variables of each type.
The logical constants are
$=_{\alpha \rightarrow \alpha \rightarrow\tB}$, for each type $\alpha$.
The denotation of equality $=_{\alpha \rightarrow \alpha \rightarrow\tB}$ is the identity relation between individuals
of type $\alpha$.
In addition, there may be other non-logical constants of various types.
The {\em alphabet} is the set of all variables and constants.

Next comes the definition of a term.

\begin{definition}[term $t$]\label{def:term}
A {\em term}, together with its type, is defined inductively as follows.
\begin{enumerate}\itemsep1mm\parskip0mm
\item A variable of type $\a$ is a term of type $\a$.
\item A constant of type $\a$ is a term of type $\a$.
\item If $t_\beta$ is a term of type $\beta$ and $x_\alpha$ a variable
    of type $\a$, then $\lambda x_\alpha.t_\beta$ is a term of type $\a\to\beta$.
\item If $s_{\a\to\beta}$ is a term of type $\a\to\beta$ and $t_\alpha$ a
    term of type $\a$, then $(s_{\a\to\beta} \; t_\alpha)$ is a term of type $\beta$.
\end{enumerate}
A {\em formula} is a term of type $\tB$.
A {\em closed} term is a term with no free variables.
A {\em sentence} is a closed formula.
A {\em theory} is a set of formulas.
\end{definition}

If the set of non-logical constants is countable, then the set of terms is denumerable.
As shown in \cite[p.212]{andrews2}, 
using equality, it is easy to define
$\top_\tB$ (truth), $\bot_\tB$ (falsity), $\land_{\tB \rightarrow\tB \rightarrow\tB}$ (conjunction),
$\lor_{\tB \rightarrow\tB \rightarrow\tB}$ (disjunction),
$\lnot_{\tB \rightarrow\tB}$ (negation), $\forall x_\alpha.t_\tB$ (universal quantification),
and $\exists x_\alpha.t_\tB$ (existential quantification).
The axioms for the logic are as follows \cite[p.213]{andrews2}: 

\begin{axiom}[logical axioms]\hfill\par
\begin{enumerate}
\item Truth values: $(g_{\tB \rightarrow \tB} \; \top_\tB) \wedge (g_{\tB \rightarrow \tB} \; \bot_\tB) = \forall x_\tB. (g_{\tB \rightarrow \tB} \; x_\tB)$ 
\item Leibniz' law: $(x_\alpha = y_\alpha) \rightarrow ((h_{\alpha \rightarrow \tB} \; x_\alpha) = (h_{\alpha \rightarrow \tB} \; y_\alpha))$ 
\item Extensionality: $(f_{\alpha \rightarrow \beta} = g_{\alpha \rightarrow \beta}) =
              \forall x_\alpha. ((f_{\alpha \rightarrow \beta} \; x_\alpha) = (g_{\alpha \rightarrow \beta} \; x_\alpha))$
\item $\beta$-reduction: $(\lambda {\mathbf x_\alpha}. {\mathbf t_\beta} \; {\mathbf s_\alpha)} = {\mathbf t_\beta} \{ {\mathbf x_\alpha}/{\mathbf s_\alpha} \}
\;\;\; \text{(provided that $\mathbf{s_\alpha}$ is free for $\mathbf{x_\alpha}$ in $\mathbf{t_\beta}$)}$
\end{enumerate}
\end{axiom}
In the above, $g_{\tB \rightarrow \tB}$, \ldots are variables of the indicated type,
$\mathbf{x_\alpha}$ is a syntactical variable for variables of type $\alpha$, and
${\mathbf t_\beta}$, \ldots are syntactical variables for terms of the indicated type.
Also ${\mathbf t_\beta} \{ {\mathbf x_\alpha}/{\mathbf s_\alpha} \}$ is the result of simultaneously substituting ${\mathbf s_\alpha}$
for all free occurrences of ${\mathbf x_\alpha}$ in ${\mathbf t_\beta}$.

Axiom (1) expresses the idea the truth and falsity are the only truth values;
Axioms (2) (for each type $\alpha$) express a basic property of equality;
Axioms (3) (for each type $\alpha \rightarrow \beta$) are the axioms of extensionality; and
Axiom schemata (4) is the axiom for $\beta$-reduction.

Here is the single rule of inference \cite[p.213]{andrews2}: 
\begin{rule1}[rule of inference; equality substitution]
From ${\mathbf t_\tB}$ and ${\mathbf s_\alpha} = {\mathbf r_\alpha}$, infer the result of replacing one occurrence
of ${\mathbf s_\alpha}$ in ${\mathbf t_\tB}$  by an occurrence of ${\mathbf r_\alpha}$, provided that the occurrence of
${\mathbf s_\alpha}$ in ${\mathbf t_\tB}$ is not (an occurrence of a variable) immediately preceded by a $\lambda$.
\end{rule1}

The logic also has an equational reasoning system that has been used as the computational basis
for a functional logic programming language \cite{LogicforLearning,ng-lloyd-JAL,ng-lloyd-uther-AMAI,lloyd-ng-JAAMAS}.

In the following, to simplify the notation, we usually omit the type subscripts on terms;
the type of a term will always either be unimportant or clear from the context.
We use $\ph,\chi,\psi$ for sentences and sometimes for formulas, and $t,r,s$ for terms.
With this notation,
$\forall x.\ph \equiv [\lambda x.\ph=\lambda x.\top]$ and
$\exists x.\ph \equiv [\lambda x.\ph\neq\lambda x.\bot]$.

The logic includes Church's $\lambda$-calculus:
a term of the form $\lambda x.t$ is an abstraction and a term of the form $(s \; t)$ is an application.

The logic is given a conventional Henkin semantics \cite{henkin}.

\begin{definition}[frame $\{ {\cal D}_\a \}_\a$]\label{def:frame}
A {\em frame} is a collection $\{ {\cal D}_\a \}_\a$ of non-empty sets, one
for each type $\a$, satisfying the following conditions.
\begin{enumerate}\itemsep1mm\parskip0mm
 \item ${\cal D}_\tB = \{ \mathsf{T}, \mathsf{F} \}$.
 \item ${\cal D}_{\beta\to\gamma}$ is some collection of functions from ${\cal D}_\beta$ to ${\cal D}_\gamma$.

\end{enumerate}
For each type $\a$, ${\cal D}_\a$ is a called a {\em domain}.
\end{definition}

The members of ${\cal D}_\tB$ are called the {\em truth values} and the members of ${\cal D}_\imath$ are called {\em individuals}.

\begin{definition}[valuation $V$]\label{def:valuation}
Given a frame $\{ {\cal D}_\alpha \}_\alpha$, a {\em valuation} $V$ is a function that maps each constant
having type $\a$ to an element of ${\cal D}_\a$ such that
$V(=_{\a \to \a \to\tB})$ is the function from ${\cal D}_\a$ into ${\cal D}_{\a\to\tB}$ defined by
        \begin{equation*}
        V(=_{\a \to \a \to\tB})\,x\,y = \begin{cases} \mathsf{T} & \text{if } x = y \\ \mathsf{F} & \text{otherwise,} \end{cases}
        \end{equation*}
        for $x,y \in {\cal D}_\a$.
\end{definition}

\begin{definition}[variable assignment $\va$]\label{def:var assignment}
A {\em variable assignment} $\va$ with respect to a frame $\{ {\cal D}_\a \}_\a$ is a function that maps each variable of type
$\a$ to an element of ${\cal D}_\a$.
\end{definition}

An interpretation can now be defined.

\begin{definition}[interpretation $\langle \{ {\cal D}_\alpha \}_\alpha, V \rangle$]\label{def:interpretation}
A pair $I \equiv \langle \{ {\cal D}_\alpha \}_\alpha, V \rangle$ is an {\em interpretation} if there is a function
$\V$ such that, for each variable assignment $\nu$ and for each term $t$ of type $\alpha$,
$\V(t,I,\nu) \in \D_\alpha$ and the following conditions are satisfied.
\begin{enumerate}\itemsep1mm\parskip0mm
 \item $\V(x,I,\va) = \va(x)$, where $x$ is a variable.
 \item $\V(C,I,\va) = V(C)$, where $C$ is a constant.
 \item $\V(\lambda x.s,I,\va) = $ the function whose value for each $d \in {\cal D}_\beta$ is
       $\V(s,I,\va')$, where $\lambda x.s$ has type $\beta\to\gamma$ and $\va'$ is $\va$
       except $\va'(x) = d$.
 \item $\V((r\;s),I,\va) = \V(r,I,\va)(\V(s,I,\va))$.
\end{enumerate}
\end{definition}

If $\langle \{ {\cal D}_\alpha \}_\alpha, V \rangle$ is an interpretation, then the function $\V$ is uniquely defined.
$\V(t,I,\va)$ is called the {\em denotation} of $t$ with respect to $I$ and $\nu$.
If $t$ is a closed term, then $\V(t,I,\va)$ is independent of $\va$ and we write it as $\V(t,I)$.
Not every pair $\langle \{ {\cal D}_\alpha \}_\alpha, V \rangle$ is an interpretation;
to be an interpretation, every term must have a denotation with respect to each variable assignment.

What is called an interpretation here is called a {\em general model} in \cite{andrews2}, following Henkin.
In \cite{andrews2}, a general model is called a {\em standard model} if, for each $\alpha$ and $\beta$, $\D_{\alpha \to \beta}$ is the
set of {\em all} functions from  $\D_\alpha$ to $\D_\beta$.
Moving from standard models to general models was the crucial step that allowed Henkin to prove
the completeness of the logic \cite{henkin}.

\begin{definition}[satisfiable]\label{def:satifiable valid}
Let $t$ be a formula, $I\equiv \langle \{ {\cal D}_\alpha \}_\alpha, V \rangle$ an interpretation,
and $\va$ a variable assignment with respect to $\{ {\cal D}_\a \}_\a$.
\begin{enumerate}\itemsep1mm\parskip0mm
 \item {\em $\va$ satisfies $t$ in $I$} if $\V(t,I,\va) = \mathsf{T}$.
 \item $t$ is {\em satisfiable in $I$} if there is a variable assignment which satisfies $t$ in $I$.
 \item $t$ is {\em valid in $I$} if every variable assignment satisfies $t$ in $I$.
 \item $t$ is {\em valid} if $t$ is valid in every interpretation.
 \item A {\em model} for a theory is an interpretation in which each formula in the theory is valid.
\end{enumerate}
\end{definition}

\begin{definition}[consistency]\label{def:consistent}
A theory is {\em consistent} if $\bot$ cannot be derived from the theory.
\end{definition}

\begin{definition}[logical consequence]\label{def:logical-consequence}
A formula $t$ is a {\em logical consequence} of a theory if
$t$ is valid in every model of the theory.
\end{definition}

We will have need for a particular class of interpretations, defined as follows.

\begin{definition}[separating interpretation/model]\label{def:separating}
An interpretation $I$ for an alphabet is {\em
separating}\index{separating
interpretation}\index{interpretation|separating} if, for every pair $r$, $s$ of closed terms
of the same function type,
say, $\alpha \rightarrow \beta$, such that $\V(r,I) \neq \V(s,I)$,
there exists a closed term $t$ of type $\alpha$ such that
$\V((r \; t), I) \neq \V((s \; t), I)$.

A {\em separating  model}\index{separating model} is a separating
interpretation that is a model (for some set of formulas).
\end{definition}

We emphasize that, in the definition of a separating
interpretation, the closed term $t$ is formed only from symbols {\em in the given alphabet}.
Intuitively, an interpretation is separating if, for every pair $r$, $s$ of closed terms of the same type
$\alpha \rightarrow \beta$, whose respective denotations in the interpretation are different,
there exists a closed term $t$ of type $\alpha$ for which the
respective denotations in the interpretation of $(r \; t)$ and $(s \; t)$ are different.
Thus, in a separating interpretation, closed terms that have distinct functions as denotations
must be distinct on an argument in the domain that is the denotation of some closed term
using the given alphabet and thus is `accessible' or `nameable' via that term.

The concept of a separating interpretation is closely related to the concept of an extensionally
complete theory that plays a crucial part in the proof of completeness \cite[p.248]{andrews2}.

\begin{definition}[extensionally complete]\label{defn:extensionally_complete}
A set $S$ of sentences is {\em extensionally complete} if, for every pair $r$, $s$ of closed terms
of the same function type, say, $\alpha \rightarrow \beta$, there exists a closed term $t$ of type $\alpha$
such that $r \neq s \rightarrow (r \; t) \neq (s \; t)$ is derivable from $S$.
\end{definition}

A connection with separating interpretations is provided by the following result.

\begin{proposition}[extensionally complete $\Rightarrow$ separating]\label{extensionally-complete-separating}
Every model of an extensionally complete set of sentences is separating.
\end{proposition}

\begin{proof}
Let $S$ be a set of sentences that is extensionally complete and $I$ be a model for $S$.
Suppose that $r$, $s$ is a pair of closed terms of the same function type,
say, $\alpha \rightarrow \beta$, such that $\V(r,I) \neq \V(s,I)$.
By extensional completeness, there exists a closed term $t$ such that
$r \neq s \rightarrow (r \; t) \neq (s \; t)$ is derivable from $S$.
Since $I$ is a model for $S$ and the proof system is sound,
it follows that $\V((r \; t), I) \neq \V((s \; t), I)$.
Hence $I$ is separating.
\qed
\end{proof}

Now we show that, if we are willing to expand the alphabet,
any set of sentences having a model also has a separating model in an expanded alphabet.

\begin{proposition}[existence of separating models]\label{prop:existseparating}
If a set $S$ of sentences has a model, then there exists an alphabet that includes the original alphabet
and an interpretation based on the expanded alphabet which is a separating model for $S$.
\end{proposition}

\begin{proof}
Since $S$ has a model, $S$ is consistent.
By \cite[Theorem 5500]{andrews2}, there is an expansion of the original alphabet and a set $T$ of sentences
such that $S \subseteq T$, $T$ is consistent, and $T$ is extensionally complete in the expanded alphabet.
Since $T$ is consistent, by Henkin's Theorem \cite[Theorem 5501]{andrews2}, it has a model (based on the expanded alphabet).
By Proposition~\ref{extensionally-complete-separating}, this model must be a separating one,
and it is also a model for $S$.
\qed
\end{proof}

The most important property of the logic that we will need is compactness \cite[Theorem 5503]{andrews2}.

\begin{theorem}[compactness]\label{prop:compactness}
If every finite subset of a set $S$ of sentences has a model, then $S$ has a model.
\end{theorem}

In fact, most of the development in the paper can be carried out in any logic that has the compactness property.

While the version of higher-order logic introduced in this section generally provides much more direct and succinct
formalisations than first-order logic, for practical applications a number of extensions are highly desirable.
Some of these extensions are nothing more than abbreviations, such as those used to introduce the connectives
and quantifiers, and some are deeper.
These extensions include many-sortedness, which allows more than one domain of individuals;
tuples and product types; and type constructors and polymorphism.
The logic of \cite{LogicforLearning}, which is also used in \cite{ng-lloyd-JAL,ng-lloyd-uther-AMAI},
includes all these extensions.
These and other extensions are discussed in \cite{farmer}.

\section{Probabilities on Sentences}\label{Probabilities on Sentences}

We now define probabilities on sentences.
They are not probabilities in the conventional sense of probability theory (on $\sigma$-algebras);
however, a connection between probabilities on sentences and (conventional) probabilities on a
$\sigma$-algebra on the set of interpretations will be made below.

\begin{definition}[probability on sentences]\label{def:prob_phi}
Let $\S$ be the set of all sentences (for some alphabet).
A {\em probability (on sentences)}\index{probability!on sentences} is a non-negative function
$\mu:\S\to\SetR$ satisfying the following conditions:
\begin{enumerate}
\item If $\ph$ is valid, then $\mu(\ph) = 1$.
\item If $\neg(\ph \wedge \psi)$ is valid, then $\mu(\ph \vee \psi) = \mu(\ph) + \mu(\psi)$.
\end{enumerate}
\end{definition}

For a sentence $\psi$, where $\mu(\psi) > 0$, one can define
the conditional probability\index{conditional probability}
$\mu(\cdot | \psi)$ by
\beqn
  \mu(\ph | \psi) \;=\; \frac{\mu(\ph \wedge \psi)}{\mu(\psi)},
\eeqn
for each sentence $\ph$.

A probability $\mu:\S\to\SetR$ on sets of
sentences has the following intended meaning:
\begin{quote}
For a sentence $\ph$, $\mu(\ph)$ is the degree of belief that $\ph$ is true.
\end{quote}

\begin{definition}[pairwise disjoint sentences]\label{def:disjoint}
The sentences $\ph_1, ..., \ph_n$ are {\em pairwise
disjoint}\index{pairwise disjoint sentences} if, for each $i,j
= 1, ..., n$ such that $i \neq j$, $\neg (\ph_i \wedge
\ph_j)$ is valid.
\end{definition}

\begin{proposition}[properties of probability on sentences]\label{prop:PPS}
Let $\mu:\S\to\SetR$ be a probability
on sentences. Then the following hold:
\begin{enumerate} \itemsep1mm\parskip0mm
\item $\mu(\neg \ph) = 1 - \mu(\ph)$, for each $\ph \in \S$.
\item $\mu(\ph) \leq 1$, for each $\ph \in \S$.
\item If $\ph$ is unsatisfiable, then $\mu(\ph) = 0$.
\item If $\ph \rightarrow \psi$ is valid, then $\mu(\ph) \leq \mu(\psi)$.
\item If $\ph = \psi$ is valid, then $\mu(\ph) = \mu(\psi)$.
\item If $\{\ph_i\}_{i=1}^n$ is a finite subset of pairwise disjoint \\ sentences in $\S$,
      then $\mu(\bigvee_{i=1}^n \ph_i) = \sum_{i=1}^n \mu(\ph_i)$.
\item If $\{\ph_i\}_{i=1}^n$ is a finite subset of $\S$,
      then $\mu(\bigvee_{i=1}^n \ph_i) \leq \sum_{i=1}^n \mu(\ph_i)$.
\item The following are equivalent: \\
      (a) For each $\ph \in \S$, $\mu(\ph) = 1$ implies $\ph$ is valid. \\
      (b) For each $\ph \in \S$, $\mu(\ph) = 0$ implies $\ph$ is unsatisfiable.
\item If $\mu(\psi) > 0$, then $\mu(\cdot | \psi)$ is a probability.
\item $\mu(\ph\vee\psi)+\mu(\ph\wedge\psi) = \mu(\ph)+\mu(\psi)$.
\end{enumerate}
\end{proposition}

\begin{proof}
The proof is elementary and standard, and only included for completeness.

1. Since $\neg(\ph \wedge \neg \ph)$ is valid,
   $\mu(\ph \vee \neg \ph) = \mu(\ph) + \mu(\neg \ph)$.
   Also, since $\ph \vee \neg \ph$ is valid, $\mu(\ph \vee \neg \ph) = 1$.
   Thus $\mu(\neg \ph) = 1 - \mu(\ph)$.

2. Since $1 - \mu(\ph) = \mu(\neg \ph) \geq 0$, we have that $\mu(\ph) \leq 1$.

3. Note that $\ph$ is unsatisfiable iff $\neg \ph$ is valid.
   Thus $\mu(\neg \ph) = 1 - \mu(\ph) = 1$, so that $\mu(\ph) = 0$.

4. Note first that $\ph \rightarrow \psi$ is valid iff $\neg(\ph \wedge \neg \psi)$ is valid.
   Thus $\mu(\ph \vee \neg \psi) = \linebreak \mu(\ph) + \mu(\neg \psi) = \mu(\ph) + 1 - \mu(\psi)$.
   Hence $\mu(\ph) = \mu(\psi) + \mu(\ph \vee \neg \psi) -1 \leq \mu(\psi)$.

5. This follows immediately from Part 4.

6. The proof is by induction on $n$.
   When $n = 1$ the result is obvious.
   Assume now the result is true for $n-1$.
   Note that $\bigwedge_{i=2}^n \neg (\ph_1 \wedge \ph_i)$ is valid and so
   $\neg(\ph_1 \wedge \bigvee_{i=2}^n \ph_i)$ is valid.
   Then
   \begin{align*}
     & \; \textstyle\mu(\bigvee_{i=1}^n \ph_i) \\
   = & \; \textstyle\mu(\ph_1 \vee \bigvee_{i=2}^n \ph_i) \\
   = & \; \textstyle\mu(\ph_1) + \mu(\bigvee_{i=2}^n \ph_i) & \text{[$\textstyle\neg(\ph_1 \wedge \bigvee_{i=2}^n \ph_i)$ is valid]} \\
   = & \; \textstyle\mu(\ph_1) +  \sum_{i=2}^n \mu(\ph_i) & \text{[induction hypothesis]} \\
   = & \; \textstyle\sum_{i=1}^n \mu(\ph_i).
   \end{align*}

7. The proof is by induction on $n$.
   When $n = 1$ the result is obvious.
   Assume now the result is true for $n-1$.
   Then
   \begin{align*}
        & \; \textstyle\mu(\bigvee_{i=1}^n \ph_i) \\
      = & \; \textstyle\mu((\ph_1 \wedge \neg \bigvee_{i=2}^n \ph_i) \vee \bigvee_{i=2}^n \ph_i) \\
      = & \; \textstyle\mu(\ph_1 \wedge \neg \bigvee_{i=2}^n \ph_i) + \mu(\bigvee_{i=2}^n \ph_i) \\
   \leq & \; \textstyle\mu(\ph_1) +  \sum_{i=2}^n \mu(\ph_i) & \text{[Part 4 and induction hypothesis]} \\
      = & \; \textstyle\sum_{i=1}^n \mu(\ph_i).
   \end{align*}

8. Suppose that, for each $\ph \in \S$, $\mu(\ph) = 1$
   implies $\ph$ is valid. Now let $\psi \in \S$
   satisfy $\mu(\psi) = 0$. By Part 1, $\mu(\neg \psi) = 1$.
   Thus $\neg \psi$ is valid and so $\psi$ is unsatisfiable.

   Conversely, suppose that, for each $\ph \in \S$,
   $\mu(\ph) = 0$ implies $\ph$ is unsatisfiable. Now let $\psi
   \in \S$ satisfy $\mu(\psi) = 1$. By Part 1,
   $\mu(\neg \psi) = 0$. Thus $\neg \psi$ is unsatisfiable and
   so $\psi$ is valid.

9. Suppose that $\ph$ is valid.
   Then $\mu(\ph | \psi) = \frac{\mu(\ph \wedge \psi)}{\mu(\psi)} = \frac{\mu(\psi)}{\mu(\psi)} = 1$.

   Suppose that $\neg(\ph \wedge \chi)$ is valid.
   Then
   \begin{align*}
     & \; \mu(\ph \vee \chi | \psi) \\
   = & \; \mu((\ph \vee \chi) \wedge \psi)~/~\mu(\psi) \\
   = & \; \mu((\ph \wedge \psi) \vee (\chi \wedge \psi))~/~\mu(\psi) \\
   = & \; [\mu(\ph \wedge \psi) + \mu(\chi \wedge \psi)]~/~\mu(\psi)
                                    & \text{[$\neg((\ph \wedge \psi) \wedge (\chi \wedge \psi))$ is valid]} \\
   = & \; \mu(\ph | \psi) + \mu(\chi | \psi).
   \end{align*}
   Thus $\mu(\cdot | \psi)$ is a probability.

10. Let $\chi:=\neg\ph\wedge\psi$. Then
   \begin{align*}
        & \; \mu(\ph\vee\psi) + \mu(\ph\wedge\psi) \\
      = & \; \mu(\ph\vee\chi) + \mu(\ph\wedge\psi) & \text{[elementary logic]} \\
      = & \; \mu(\ph)+\mu(\chi) + \mu(\ph\wedge\psi) & \text{[$\neg(\ph\wedge\chi)$ is valid and Def.~\ref{def:prob_phi}.2]} \\
      = & \; \mu(\ph)+\mu(\chi\vee(\ph\wedge\psi)) & \text{[$\neg(\chi\wedge(\ph\wedge\psi))$ is valid and Def.~\ref{def:prob_phi}.2]} \\
      = & \; \mu(\ph)+\mu(\psi) & \text{[elementary logic]}
   \end{align*}
\qed
\end{proof}

Next we introduce Gaifman probabilities.

\begin{definition}[Gaifman probability]\label{def:Gaifman}
Let $\mu:\S\to\SetR$ be a probability on sentences.
Then $\mu$ is {\em Gaifman} if
\beqn
  \mu(r = s) = \inf_{\{t_1, ..., t_n\}} \mu(\bigwedge_{i=1}^n ((r \; t_i) = (s \; t_i))),
\eeqn
for every pair $r$ and $s$ of closed terms having the same function type, say, $\alpha \rightarrow \beta$,
and where $\{t_1, ..., t_n\}$ ranges over all finite sets of closed terms of type $\a$.
\end{definition}

\begin{proposition}[Gaifman probability]\label{prop:Gaifman}
Let $\mu:\S\to\SetR$ be a probability on sentences.
Then the following are equivalent.
\begin{enumerate}

\item $\mu$ is Gaifman.

\item $\displaystyle\mu(r \neq s) = \sup_{\{t_1, ..., t_n\}} \mu(\bigvee_{i=1}^n ((r \; t_i) \neq (s \; t_i)))$, \\[0.75em]
      for every pair $r$ and $s$ of closed terms having the same function type, say, $\alpha \rightarrow \beta$,
      and where $\{t_1, ..., t_n\}$ ranges over all finite sets of closed terms of type $\a$.

\item $\displaystyle\mu(\exists x.\ph) = \sup_{\{t_1, ..., t_n\}} \mu(\bigvee_{i=1}^n\ph\{x/t_i\})$, \\[0.75em]
      for every formula $\ph$ having a single free variable $x$ of type $\a$, say,
      and where $\{t_1, ..., t_n\}$ ranges over all finite sets of closed terms of type $\a$.

\item $\displaystyle\mu(\forall x.\ph) = \inf_{\{t_1, ..., t_n\}} \mu(\bigwedge_{i=1}^n\ph\{x/t_i\})$,  \\[0.75em]
      for every formula $\ph$ having a single free variable $x$ of type $\a$, say,
      and where $\{t_1, ..., t_n\}$ ranges over all finite sets of closed terms of type $\a$.

\end{enumerate}

\end{proposition}

\begin{proof}
1. implies 2.
Suppose that the probability $\mu$ is Gaifman.
Then
\begin{align*}
  & \; \mu(r \neq s) \\
= & \; 1 - \mu(r = s) \\
= & \; 1 - \textstyle\inf_{\{t_1, ..., t_n\}} \mu(\bigwedge_{i=1}^n ((r \; t_i) = (s \; t_i))) \\
= & \; 1 - \textstyle\inf_{\{t_1, ..., t_n\}} \mu(\neg \bigvee_{i=1}^n ((r \; t_i) \neq (s \; t_i))) \\
= & \; 1 - \textstyle\inf_{\{t_1, ..., t_n\}} (1 - \mu(\bigvee_{i=1}^n ((r \; t_i) \neq (s \; t_i))) \\
= & \; \textstyle\sup_{\{t_1, ..., t_n\}} \mu(\bigvee_{i=1}^n ((r \; t_i) \neq (s \; t_i))).
\end{align*}
Hence 2. holds.

2. implies 3.
Suppose that 2. holds.
Then
\begin{align*}
  & \; \mu(\exists x.\ph) \\
= & \; \mu(\lambda x.\ph \neq \lambda x.F) \\
= & \; \textstyle\sup_{\{t_1, ..., t_n\}} \mu(\bigvee_{i=1}^n ((\lambda x.\ph \; t_i) \neq (\lambda x.F \; t_i))) \\
= & \; \textstyle\sup_{\{t_1, ..., t_n\}} \mu(\bigvee_{i=1}^n \ph\{x/t_i\}).
\end{align*}
Hence 3. holds.

3. implies 4.
Suppose that 3. holds.
Then
\begin{align*}
  & \; \mu(\forall x.\ph) \\
= & \; \mu(\neg \exists x.\neg \ph) \\
= & \; 1 - \mu(\exists x.\neg \ph) \\
= & \; 1 - \textstyle\sup_{\{t_1, ..., t_n\}} \mu(\bigvee_{i=1}^n \neg\ph\{x/t_i\}) \\
= & \; 1 - \textstyle\sup_{\{t_1, ..., t_n\}} \mu(\neg \bigwedge_{i=1}^n \ph\{x/t_i\}) \\
= & \; 1 - \textstyle\sup_{\{t_1, ..., t_n\}} (1 - \mu(\bigwedge_{i=1}^n \ph\{x/t_i\})) \\
= & \; \textstyle\inf_{\{t_1, ..., t_n\}} \mu(\bigwedge_{i=1}^n \ph\{x/t_i\}).
\end{align*}
Hence 4. holds.

4. implies 1.
Suppose that 4. holds.
Then
\begin{align*}
  & \; \mu(r = s) \\
= & \; \mu(\forall x. ((r \; x) = (s \; x))) & \text{[Axioms of Extensionality]} \\
= & \; \textstyle\inf_{\{t_1, ..., t_n\}} \mu(\bigwedge_{i=1}^n ((r \; x) = (s \; x))\{x/t_i\}) \\
= & \; \textstyle\inf_{\{t_1, ..., t_n\}} \mu(\bigwedge_{i=1}^n ((r \; t_i) = (s \; t_i))).
\end{align*}
Hence 1. holds.
\qed
\end{proof}

\begin{proposition}[limits for countable alphabet]\label{prop:LCA}
Let the alphabet be countable, $\mu:\S\to\SetR$ a probability on sentences, and
$\ph$ a formula having a single free variable $x$ of type $\a$.

\begin{enumerate}
\item $\displaystyle\sup_{\{t_1, ..., t_n\}} \mu(\bigvee_{i=1}^n\ph\{x/t_i\}) =
                        \displaystyle\lim_{n\to\infty} \mu(\bigvee_{i=1}^n \ph\{x/t_i\})$
\item $\displaystyle\inf_{\{t_1, ..., t_n\}} \mu(\bigwedge_{i=1}^n\ph\{x/t_i\}) =
             \displaystyle\lim_{n\to\infty} \mu(\bigwedge_{i=1}^n \ph\{x/t_i\})$,
\end{enumerate}
where, on the LHS, $\{t_1, ..., t_n\}$ ranges over all finite sets of closed terms of type $\a$
and, on the RHS, $t_1, t_2, ...$ is an enumeration of all closed terms of type $\a$.
\end{proposition}

\begin{proof}
Since the alphabet is countable, the set of all closed terms of
type $\a$ is countable and hence can be enumerated.

1. Let $\{t'_1, ..., t'_m\}$ be a subset of closed terms of type $\a$.
Let $n$ be sufficiently large so that each $t'_j$, for $j = 1, ..., m$, appears in the enumeration
$t_1, ..., t_n$ of the first $n$ terms of an enumeration of all closed terms of type $\a$.

Then $\bigvee_{j=1}^m\ph\{x/t'_j\} \rightarrow \bigvee_{i=1}^n \ph\{x/t_i\}$ is valid,
so that
\beqn
  \textstyle\mu(\bigvee_{j=1}^m\ph\{x/t'_j\}) \;\leq\; \mu(\bigvee_{i=1}^n \ph\{x/t_i\}),
\eeqn
by Proposition~\ref{prop:PPS}.4.
By first taking the supremum on the RHS and then the supremum on the LHS we get
\beqn
  \textstyle\sup_{\{t'_1, ..., t'_m\}} \mu(\bigvee_{j=1}^m\ph\{x/t'_j\})
  \;\leq\;  \sup_{n} \mu(\bigvee_{i=1}^n\ph\{x/t_i\}).
\eeqn
Conversely we have
\beqn
  \textstyle\sup_{\{t'_1, ..., t'_m\}} \mu(\bigvee_{j=1}^m\ph\{x/t'_j\})
  \;\geq\;  \mu(\bigvee_{i=1}^n\ph\{x/t_i\}).
\eeqn
since the sup on the LHS includes $\{t_1,...,t_n\}$.
Now taking the limit $n\to\infty$ and combining both inequalities gives equality.
Proposition~\ref{prop:PPS}.4 gives that $\mu(\ph \vee \psi)\geq\mu(\ph)$;
hence $\mu(\bigvee_{i=1}^n\ph\{x/t_i\})$ is monotone non-decreasing in $n$,
which allows the replacement of $\sup_n$ by $\lim_{n\to\infty}$.

2.  The proof is similar.
\qed
\end{proof}

We can reduce the class of terms that is necessary to ``browse'' through even further,
by considering only one term from each equivalence class,
where two terms $t$ and $t'$ are equivalent iff $t=t'$ is valid.

\begin{proposition}[Gaifman for countable alphabet]\label{prop:GCA}
Let the alphabet be countable and $\mu:\S\to\SetR$ a probability on sentences.
Then the following are equivalent.
\begin{enumerate}

\item $\mu$ is Gaifman.

\item $\displaystyle\mu(r = s) = \lim_{n \rightarrow \infty} \mu(\bigwedge_{i=1}^n ((r \; t_i) = (s \; t_i)))$, \\[0.75em]
      for every pair $r$ and $s$ of closed terms having the same function type, say, $\alpha \rightarrow \beta$,
      and where $t_1, t_2, \ldots$ is an enumeration of all closed terms of type $\a$.

\item $\displaystyle\mu(r \neq s) = \lim_{n \rightarrow \infty} \mu(\bigvee_{i=1}^n ((r \; t_i) \neq (s \; t_i)))$, \\[0.75em]
      for every pair $r$ and $s$ of closed terms having the same function type, say, $\alpha \rightarrow \beta$,
      and where $t_1, t_2, \ldots$ is an enumeration of all closed terms of type $\a$.

\item $\displaystyle\mu(\exists x.\ph) = \lim_{n \rightarrow \infty} \mu(\bigvee_{i=1}^n\ph\{x/t_i\})$, \\[0.75em]
      for every formula $\ph$ having a single free variable $x$ of type $\a$, say,
      and where $t_1, t_2, \ldots$ is an enumeration of all closed terms of type $\a$.

\item $\displaystyle\mu(\forall x.\ph) = \lim_{n \rightarrow \infty} \mu(\bigwedge_{i=1}^n\ph\{x/t_i\})$,  \\[0.75em]
      for every formula $\ph$ having a single free variable $x$ of type $\a$, say,
      and where $t_1, t_2, \ldots$ is an enumeration of all closed terms of type $\a$.

\end{enumerate}
In each case, the enumeration $t_1, t_2, \ldots$ of closed terms of type $\a$ can be reduced
to one where a single representative is chosen from each equivalence class under the equivalence relation
$t$ and $t'$ are equivalent if $t = t'$ is valid.

\end{proposition}

\begin{proof}
Two terms $t$ and $t'$ are said to be equivalent iff $t=t'$ is
valid, which implies $\ph\{x/t\} = \ph\{x/t'\}$ is valid.
This allows us to relax in the proof of Proposition~\ref{prop:LCA} `appears' by
`is equivalent to some term in' and `includes' by `includes a term equivalent to some term in'.
Finally combine this with Proposition~\ref{prop:Gaifman} and Definition~\ref{def:Gaifman}.
\qed
\end{proof}

While these forms of the Gaifman condition closely resemble the
continuity condition (countable additivity (CA) axiom) in measure
theory, we will see that CA over (general) interpretations
is derived from the compactness theorem and not from the Gaifman
condition (see Definition~\ref{def:probmeasure} and
Proposition~\ref{prop:FAvsCA} in the next section).
But the Gaifman condition confines probabilities to separating interpretations
while preserving CA (Propositions~\ref{prop:mustartomu} and \ref{prop:mutomustar}).

\begin{example}[natural numbers $\Nat$]\label{ex:Nat}
Consider the standard type $\Nat$ of natural numbers, as the type of individuals,
and the usual Peano axioms. Let $\tfn{0}$ be the constant of
type $\Nat$ whose denotation is the natural number 0, and $\tfn{n}\equiv
S^n(\tfn{0})=(S\;(S\;(S \cdots(S\;\tfn{0}))))$ be the term of type $\Nat$
whose denotation is the natural number $n$, where $S$ is a
constant of type $\Nat\to\Nat$ whose denotation is the
successor function.
In practice one usually defines denumerably many constants
$\tfn{1},\tfn{2},\tfn{3},...$, one for
each natural number, directly.
Further, let $+,\times:\Nat\to\Nat\to\Nat$ be functions with
their usual axioms and meaning.
Now there are many closed terms that represent the same
natural number. For instance $\tfn{8}$, $(\lambda x.x~\tfn{8})$,
$(\tfn{3}+\tfn{5})$, $(\tfn{2}\times\tfn{4})$ are different terms, all
having the number $8$ as denotation.
For type $\Nat$, it is sufficient to choose $t_n=\tfn{n}$ in
Proposition~\ref{prop:GCA}.4, and so the
condition in Definition~\ref{def:Gaifman} (indeed) reduces to
the one used by Gaifman \cite{gaifman-snir}.
\eoe\end{example}

Of particular interest are probabilities that are strictly
positive on satisfiable sentences since this is a desirable
property of a prior.
This suggests the following definition.

\begin{definition}[strongly Cournot probability]\label{def:stronglyCournot}
A probability $\mu:\S\to\SetR$ is
{\em strongly Cournot} if, for each $\ph \in \S$, $\ph$ is
satisfiable implies $\mu(\ph) > 0$.
\end{definition}

By Part 8 of Proposition~\ref{prop:PPS}, a probability is
strongly Cournot iff, for each $\ph \in \S$, $\ph$ is not valid
implies $\mu(\ph)<1$, or, by contraposition, $\mu(\ph)=1$
implies $\ph$ is valid. This is akin to Cournot's principle as
discussed in the introduction that an event of probability 1
singled out in advance will happen for sure in the real world.
We will see this general idea plays an important role for inductive
inference.

However, the following weaker form of the Cournot principle will turn out to be more useful.

\begin{definition}[Cournot probability]\label{def:Cournot}
A probability $\mu:\S\to\SetR$ is
{\em Cournot} if, for each $\ph \in \S$, $\ph$ has a
separating model implies $\mu(\ph) > 0$.
\end{definition}

Clearly a strongly Cournot probability is Cournot.
It will be the Cournot probabilities (not the strongly Cournot ones) that will
be of most interest in the subsequent development.
The major reasons for this are as follows.
First, Theorem~\ref{thm:CandG} below shows that, if the alphabet is countable,
there exists a probability on sentences that is Cournot and Gaifman.
Such a probability makes a good prior.
Second, the Cournot and Gaifman conditions are necessary and sufficient to do learning
in the limit of universal hypotheses as the following theorem
shows and as discussed in more detail in
Section~\ref{sec:Examples}.

\begin{theorem}[confirming universal hypotheses]\label{thm:CUH}
Let the alphabet be countable, $\mu:\S\to\SetR$ a probability on sentences,
$\ph$ a formula having a single free variable $x$ of some type $\a$,
$t_1,t_2,...$ an enumeration of (representatives of) all closed terms of type $\a$. Then
\beqn
  \mu(\forall x.\ph \,|\, \bigwedge_{i=1}^n\ph\{x/t_i\})\stackrel{n\to\infty}\longrightarrow 1
  \quad\Leftrightarrow\quad
  \mu(\bigwedge_{i=1}^n\ph\{x/t_i\}) \stackrel{n\to\infty}\longrightarrow \mu(\forall x.\ph)>0
\eeqn
If the left hand side (hence also the r.h.s.) holds, we say that $\mu$ can confirm
universal hypothesis $\forall x.\ph$.
It also holds that
\beqn
  {\text{$\mu $ can confirm all universal hypotheses}\atop\text{ that have a separating model}}
  \quad\Leftrightarrow\quad \mu \text{ is Gaifman and Cournot}
\eeqn
\end{theorem}

\begin{proof}
{\boldmath{$(top\Leftarrow)$}}\vspace{-4ex}
\begin{align*}
  & \qquad\qquad\textstyle\lim_{n\to\infty} \mu(\forall x.\ph\,|\,\bigwedge_{i=1}^n \ph\{x/t_i\}) \\
= & \; \frac{\mu(\forall x.\ph)}{\lim_{n\to\infty}\mu(\bigwedge_{i=1}^n \ph\{x/t_i\})}
  & \textstyle \text{[$\forall x.\ph\to\bigwedge_{i=1}^n \ph\{x/t_i\}$]} \\
= & \; \frac{\mu(\forall x.\ph)}{\mu(\forall x.\ph)}
  & \textstyle\text{[$\bigwedge_{i=1}^n\ph\{x/t_i\})\stackrel{n\to\infty}\longrightarrow \mu(\forall x.\ph)$]} \\
= & \; 1 & \text{[$\mu(\forall x.\ph)>0$]}
\end{align*}
{\boldmath{$(top\Rightarrow)$}} As can be seen from the
$\Leftarrow$ proof, if one or both of the conditions fail, then
$\mu(\forall x.\ph\,|\,\bigwedge_{i=1}^n \ph\{x/t_i\})$ does not converge to 1.

For the bottom$\Leftrightarrow$ we abbreviate the statements
\begin{align*}
  L(\ph) \;&:=\; \textstyle[\mu(\forall x.\ph \,|\, \bigwedge_{i=1}^n\ph\{x/t_i\})\stackrel{n\to\infty}\longrightarrow 1] \\
  G(\ph) \;&:=\; \textstyle[\mu(\bigwedge_{i=1}^n\ph\{x/t_i\}) \stackrel{n\to\infty}\longrightarrow \mu(\forall x.\ph)] \\
  S(\ph) \;&:=\; [\forall x.\ph \text{ has a separating model}] \\
  A(\ph) \;&:=\; [\mu(\forall x.\ph)>0]
\end{align*}
In this notation, the top$\Leftrightarrow$ reads
$L(\ph)$ {\em iff} $G(\ph)$ and $A(\ph)$.

{\boldmath{$(bottom\Leftarrow)$}} Assume $\mu$ is Gaifman and Cournot and $S(\ph)$.
This implies $G(\ph)$ and $A(\ph)$. By $top\Leftarrow$ we get $L(\ph)$.
We have shown that for any $\ph$, if $\mu$ is Gaifman and Cournot, then $S(\ph)$ implies $L(\ph)$.

{\boldmath{$(bottom\Rightarrow)$}} Case 1 [$S(\ph)$ is true]
Then by assumption, $L(\ph)$. Then by $top\Rightarrow$ we get
$G(\ph)$ and $A(\ph)$.
Note that every sentence $\psi$ can be written as $\psi=\forall
x.\ph$ with $\ph:=[\psi\wedge(x=x)]$ being a formula having a
single free variable $x$. Therefore, $\mu(\psi)=\mu(\forall
x.\ph)>0$ for all $\psi$ that have a separating model. Hence
$\mu$ is Cournot.
\\
Case 2 [$S(\ph)$ is false] That is, $\forall x.\ph$ has no
separating model, therefore $\neg\forall x.\ph$ must have (at
least one) separating model, say $\widehat I$. Since $\widehat
I$ is a {\em separating} model of $\exists x.\neg\ph$,
Definition~\ref{def:separating} implies that there exists a
closed term $t$ such that $\widehat I$ is also a separating
model of $\chi:=\neg\ph\{x/t\}$. Now
\begin{align*}
  & \; \mu(\forall x.\ph)+\mu(\chi) \\
= & \; \mu(\forall x.\ph \vee \chi)
     & \text{[$\forall x.\ph$ and $\chi$ are disjoint]} \\
= & \; \mu(\forall x.(\ph \vee \chi))
     & \text{[$x$ is not free in $\chi$]} \\
= & \; \textstyle\lim_n\mu(\bigwedge_{i=1}^n(\ph \vee \chi)\{x/t_i\})
     & \text{[since $S(\ph \vee \chi)$, Case 1 implies $G(\ph \vee \chi)$]} \\
= & \; \textstyle\lim_n\mu(\bigwedge_{i=1}^n\ph\{x/t_i\} \vee \chi))
     & \text{[$x$ is not free in $\chi$]} \\
= & \; \textstyle\lim_n\mu(\bigwedge_{i=1}^n\ph\{x/t_i\})+\mu(\chi)
     & \text{[$t=t_i$ for some $i$, and $\ph\{x/t\}\wedge\chi$ false]}
\end{align*}
This proves $G(\ph)$ for $S(\ph)$ false.

Case 1 and 2 together prove $G(\ph)$ for all $\ph$, hence $\mu$ is Gaifman.
\qed\end{proof}

\section{Probabilities on Interpretations}\label{Probabilities on Interpretations}

We now study probabilities defined on sets of interpretations.

Consider the set $\I$ of interpretations (for the alphabet).
A Borel $\sigma$-algebra can be defined on $\I$.
For that, a topology needs to be defined first.
Given some alphabet, let $\S$ denote the set of sentences based on the alphabet.
For each sentence $\ph$, let $\model(\ph)$ denote the set
\beqn
  \{ I \in \I \; | \; \ph \text{ is valid in } I \}.
\eeqn
Consider the set $\B_\S = \{\model(\ph) \; | \; \ph\in\S\}$.
Since $\B_\S$ is closed under finite intersections, it is a
basis for a topology $\mathcal{T}$ on $\I$. $\B_\S$ is also an
algebra, since it is closed under complementation and finite
unions, and $\I \in \B_\S$. Let $\B$ be the Borel
$\sigma$-algebra formed from the topology $\mathcal{T}$ on
$\I$. In the following, probabilities on $\B$ will be
considered.

Suppose that the alphabet is countable (equivalently, the set of constants is countable).
Then the set of terms and, in particular, the set $\S$ is countable.
In this case, $\B_\S$ is countable and hence the $\sigma$-algebra generated by $\B_\S$
is the same as the Borel $\sigma$-algebra $\B$ generated by $\mathcal{T}$.

\begin{definition}[probability on interpretations]\label{def:probmeasure}
A function $\mu^*:\B\to \SetR$ is a
{\em finitely additive probability} on algebra $\B$
if $\mu^*(\emptyset)=0$ and $\mu^*(\I)=1$
and
$\mu^*(A\cap C)+\mu^*(A\cup C) = \mu^*(A)+\mu^*(C)$
for all $A,C\in\B$.
It is called a {\em Countably Additive (CA) probability} or simply a {\em probability}
if additionally for all countable collections $\{A_i\}_{i\in I}\subset\B$
of pairwise disjoint sets with $\bigcup_{i\in I}A_i\in\B$ it holds
that $\mu^*(\bigcup_{i\in I}A_i) =
\sum_{i\in I}\mu^*(A_i)$.
\end{definition}
For CA-probabilities, $\B$ is usually assumed to be a Borel $\sigma$-algebra,
i.e.\ $\bigcup_{i\in I}A_i\in\B$ always holds.
Countable additivity is equivalent to finite additivity and {\em continuity}:
\beqn
  \quad \lim_{n\to\infty}\textstyle\mu^*(\bigcap_{i=1}^n A_i) =  \mu^*(\lim_{n\to\infty}\bigcap_{i=1}^n A_i)
  \quad\text{for all } A_i\in\B.
\eeqn

First we show that a probability on the algebra gives a probability on sentences.

\begin{proposition}[$\mu^*\Rightarrow\mu$]\label{prop:mustartomu}
Let $\S$ be the set of sentences,
$\I$ the set of interpretations,
$\B_\S = \{ \model(\ph) \; | \; \ph \in \S \}$ the algebra on $\I$,
and $\mu^* : \B_\S \rightarrow \SetR$ a finitely additive probability on  $\B_\S$.
Define $\mu : \S \rightarrow \SetR$ by
\beqn
  \mu(\ph) = \mu^*(\model(\ph)),
\eeqn
for each $\ph \in \S$.
Then $\mu$ is a probability on $\S$.
\end{proposition}

\begin{proof}
The two conditions of Definition~\ref{def:prob_phi} have to be established.
Note that $\mu$ is non-negative because $\mu^*$ is.

Suppose that $\ph$ is valid.
   Then $\model(\ph) = \I$, so that
   $\mu(\ph) = \mu^*(\model(\ph)) = \mu^*(\I) = 1$.

Suppose that $\neg(\ph \wedge \psi)$ is valid.
   Hence $\model(\ph) \cap \model(\psi) = \emptyset$.
   Thus
   \begin{align*}
     & \; \mu(\ph \vee \psi) \\
   = & \; \mu^*(\model(\ph \vee \psi)) \\
   = & \; \mu^*(\model(\ph) \cup \model(\psi)) \\
   = & \; \mu^*(\model(\ph)) + \mu^*(\model(\psi))
                             \;\;\; \text{[$\mu^*$ is finitely additive]} \\
   = & \; \mu(\ph) + \mu(\psi).
   \end{align*}
Hence $\mu$ is a probability.
\qed
\end{proof}

Note that only the finite additivity of $\mu^*$ is needed in Proposition~\ref{prop:mustartomu}.

Next we show that a probability on sentences gives a probability on interpretations.
For this, a useful property of probabilities on $\B_\S$ is needed.

\begin{proposition}[finite $\Leftrightarrow$ countable additivity]\label{prop:FAvsCA}
Let $\S$ be the set of sentences, $\I$ the set of interpretations,
and $\B_\S = \{ \model(\ph) \; | \; \ph \in \S \}$ the algebra on $\I$.
Then every finitely additive probability on $\B_\S$ is countably additive on $\B_\S$.
\end{proposition}

\begin{proof}
Let $\mu^*$ be a finitely additive probability on $\B_\S$.
Suppose that $\{ \ph_n \}_{n=1}^\infty$ is a sequence of sentences such that
$\model(\ph_n) \supseteq \model(\ph_{n+1})$,
for $n = 1, 2, \ldots$ , and $\bigcap_{n=1}^\infty \model(\ph_n) = \emptyset$.
Clearly $\ph_{n+1} \longrightarrow \ph_n$ is valid, for $n = 1, 2, \ldots $.
Next we claim that $\ph_{n_0}$ is unsatisfiable, for some $n_0$.
To prove this, suppose on the contrary that $\ph_n$ is satisfiable, for $n = 1, 2, \ldots$.
Since $\ph_{n+1} \longrightarrow \ph_n$ is valid, for $n = 1, 2, \ldots $,
it follows that $\{ \ph_1, \ldots, \ph_n \}$ is satisfiable, for $n = 1, 2, \ldots $.
By the compactness theorem, $\{ \ph_n \}_{n =1}^\infty$ is satisfiable,
which contradicts the assumption that $\bigcap_{n=1}^\infty \model(\ph_n) = \emptyset$.
Thus the claim that $\ph_{n_0}$ is unsatisfiable, for some $n_0$, is proved.
Since the $\model(\ph_n)$ are decreasing, we have that
$\model(\ph_n) = \emptyset$, for $n \geq n_0$.
It thus follows that
$\lim_{n \rightarrow \infty} \mu^*(\model(\ph_n)) = \mu^*(\emptyset) = 0$.
Hence, by \cite[Theorem 3.1.1]{dudley}, $\mu^*$ is countably additive on $\B_\S$.
\qed
\end{proof}

\begin{proposition}[$\mu\Rightarrow\mu^*$]\label{prop:mutomustar}
Let the alphabet be countable, $\S$ the set of sentences, $\I$ the set of interpretations,
and $\B$ the Borel $\sigma$-algebra on $\I$.
Let $\mu : \S \rightarrow \SetR$ be a probability on sentences.
Then there exists a unique probability $\mu^* : \B \rightarrow \SetR$ such that
\beqn
  \mu^*(\model(\ph)) = \mu(\ph),
\eeqn
for each $\ph \in \S$.
\end{proposition}

\begin{proof}
Consider the algebra $\B_\S = \{ \model(\ph) \; | \; \ph \in \S \}$.
Define $\mu^* : \B_\S \rightarrow \SetR$ by
\beqn
  \mu^*(\model(\ph)) = \mu(\ph),
\eeqn
for each $\ph \in \S$.
Suppose that $\ph$ and $\psi$ are sentences such that $\model(\ph) = \model(\psi)$.
Then $\ph = \psi$ is valid, and so $\mu(\ph) = \mu(\psi)$.
This shows that $\mu^*$ is well-defined on basic sets.

Clearly $\mu^*(\I) = \mu^*(\model(T)) = \mu(T) = 1$.

Next it is shown that $\mu^*$ is finitely additive on the algebra $\B_\S$.
Let $\{ \model(\ph_i) \}_{i=1}^n$ be a finite collection of pairwise disjoint sets in $\B_\S$.
Suppose that, for some $i$ and $j$, $\neg (\ph_i \wedge \ph_j)$ is not valid.
Hence $\ph_i \wedge \ph_j$ has a  model,
and so $\model(\ph_i) \cap \model(\ph_j) \neq \emptyset$.
Thus $\model(\ph_i) \cap \model(\ph_j) = \emptyset$ implies
$\neg (\ph_i \wedge \ph_j)$ is valid.
Then
\beqn
  \mu^*(\bigcup_{i=1}^n \model(\ph_i))
\;=\;\mu^*(\model(\bigvee_{i=1}^n \ph_i))
\;=\; \mu(\bigvee_{i=1}^n \ph_i)
\;=\; \sum_{i=1}^n \mu(\ph_i)
\;=\; \sum_{i=1}^n \mu^*(\model(\ph_i)),
\eeqn
where the second last equality follows from
Part 6 of Proposition~\ref{prop:PPS}.
Thus $\mu^*$ is finitely additive on $\B_\S$.

Now, by Proposition~\ref{prop:FAvsCA}, $\mu^*$ is countably additive on $\B_\S$.
Since the alphabet is countable, $\B_\S$
is countable, and so the Borel $\sigma$-algebra $\B$ generated by the topology on $\I$
is the same as the $\sigma$-algebra generated by $\B_\S$.
By Caratheodory's theorem \cite[Theorem 3.1.4]{dudley}, there is a unique extension of $\mu^*$ to
the Borel $\sigma$-algebra $\B$ on $\I$.
\qed
\end{proof}

A probability $\mu^* : \B \rightarrow \SetR$ on sets of interpretations
has the following intended meaning:
\begin{quote}
For a Borel set $B\in\B$, $\mu^*(B)$ is the degree of belief
that the intended interpretation is a member of $B$.
\end{quote}

We now consider probabilities defined on sets of {\em separating} interpretations.
Let $\widehat\I$ be the set of separating interpretations (for the alphabet).
A Borel $\sigma$-algebra can be defined on $\widehat\I$.
For that, a topology needs to be defined first.
For each sentence $\ph$, let $\sepmodel(\ph)$ denote the set
\beqn
  \{ I \in \widehat\I \; | \; \ph \text{ is valid in } I \}.
\eeqn
Consider the set $\widehat\B_\S = \{ \sepmodel(\ph) \; |\; \ph \in \S \}$.
Since $\widehat\B_\S$ is
closed under finite intersections, it is a basis for a topology
$\widehat{\mathcal{T}}$ on $\widehat\I$. $\widehat\B_\S$ is
also an algebra, since it is closed under complementation and
finite unions, and $\widehat\I \in \widehat\B_\S$.
Let $\widehat\B$ be the Borel $\sigma$-algebra formed from the
topology $\widehat{\mathcal{T}}$ on $\widehat\I$. In the following,
probabilities on $\widehat\B$ will be considered.
The Gaifman condition is crucial for them to be CA, since
$\widehat\B$ is not compact unlike $\B$.

Suppose that the alphabet is countable. Then the set of terms
and, in particular, the set $\S$ is countable. In this
case, $\widehat\B_\S$ is countable and hence the
$\sigma$-algebra generated by $\widehat\B_\S$ is the
same as the Borel $\sigma$-algebra $\widehat\B$ generated by
$\widehat{\mathcal{T}}$.

Note that there is a one-to-one correspondence between the set of probabilities on $\widehat\B$
and the set of probabilities on $\B$ which give measure 0 to the set of non-separating interpretations.
(The set of non-separating interpretations, and hence the set of separating interpretations,
are shown to be $\B$-measurable in the proof of Proposition~\ref{prop:mutomustarseparating} below.)
A probability $\widehat\mu^*:\widehat\B\to\SetR$ can be extended to a probability $\mu^*:\B\to\SetR$
defined by $\mu^*(B) = \widehat\mu^*(B\cap\widehat\I)$, for each $B\in\B$.
Note that $\mu^*(\I\setminus\widehat\I) = 0$.
Conversely, a probability $\mu^*:\B\to\SetR$ having the property that
$\mu^*(\I \setminus \widehat\I) = 0$
can be restricted to a probability $\mu^*|_{\widehat\B} :\widehat\B\to\SetR$ defined by
$\mu^*|_{\widehat\B}(B) = \mu^*(B)$, for each $B \in \widehat\B$.

The next result shows that a probability on the set of
separating interpretations gives a Gaifman probability on sentences.

\begin{proposition}[separating $\mu^*\Rightarrow\mu$ Gaifman]\label{prop:mustartomuseparating}
Let the alphabet be countable, $\S$ the set of sentences,
$\widehat\I$ the set of separating  interpretations,
and $\mu^* : \widehat\B \to \SetR$ a probability
on the Borel $\sigma$-algebra $\widehat\B$ on $\widehat\I$.
Define $\mu: \S \to \SetR$ by
\beqn
  \mu(\ph) = \mu^*(\sepmodel(\ph)),
\eeqn
for each $\ph \in \S$.
Then $\mu$ is a Gaifman probability on $\S$.
\end{proposition}

\begin{proof}
First, the two conditions of Definition~\ref{def:prob_phi} have to be established.
Note that $\mu$ is non-negative because $\mu^*$ is.

1. Suppose that $\ph$ is valid.
   Then $\sepmodel(\ph) = \widehat\I$, so that
   $\mu(\ph) = \mu^*(\sepmodel(\ph)) = \mu^*(\widehat\I) = 1$.

2. Suppose that $\neg(\ph \wedge \psi)$ is valid.
   Hence $\sepmodel(\ph) \cap \sepmodel(\psi) = \emptyset$.
   Thus
   \begin{align*}
     & \; \mu(\ph \vee \psi) \\
   = & \; \mu^*(\sepmodel(\ph \vee \psi)) \\
   = & \; \mu^*(\sepmodel(\ph) \cup \sepmodel(\psi)) \\
   = & \; \mu^*(\sepmodel(\ph)) + \mu^*(\sepmodel(\psi))
                             & \text{[$\mu^*$ is finitely additive]} \\
   = & \; \mu(\ph) + \mu(\psi).
   \end{align*}
Hence $\mu$ is a probability.

Let $r$ and $s$ be closed terms of type $\alpha \rightarrow \beta$ and  $t_1, t_2, ... $ an
enumeration of all closed terms of type $\a$.
Then
\beqn
  \sepmodel(r = s) \;=\; \bigcap_{i=1}^\infty \sepmodel((r \; t_i) = (s \; t_i)).
\eeqn
To see this, suppose first that $I \in \sepmodel(r = s)$.
Then clearly $I \in \sepmodel((r \; t_i) = (s \; t_i))$, for each $t_i$.
Conversely, suppose that $I$ is a separating interpretation such that $I \notin \sepmodel(r = s)$.
Since $I$ is separating, there exists a closed term $t_j$ such that
$I \notin \sepmodel((r \; t_j) = (s \; t_j))$, for some $j$.
Hence $I \notin \bigcap_{i=1}^\infty \sepmodel((r \; t_i) = (s \; t_i))$.
[Note, by the way, that $\model(r = s) \neq \bigcap_{i=1}^\infty \model((r \; t_i) = (s \; t_i))$.]

Since $\forall x.\ph$ is logically equivalent to
$\lambda x. \ph = \lambda x. T$,
it follows immediately from the remark of the preceding paragraph that
\beqn
  \sepmodel(\forall x. \ph) \;=\; \bigcap_{i=1}^\infty \sepmodel(\ph\{x/t_i\}).
\eeqn
Thus\vspace{-5ex}
\begin{align*}
\mu(\forall x.\ph) = & \; \mu^*(\sepmodel(\forall x. \ph)) \\
= & \; \textstyle\mu^*(\bigcap_{i=1}^\infty \sepmodel(\ph\{x/t_i\})) \\
= & \; \textstyle\lim_{n\to\infty} \mu^*(\bigcap_{i=1}^n \sepmodel(\ph\{x/t_i\}))
               & \text{[$\mu^*$ is countably additive]} \\
= & \; \textstyle\lim_{n\to\infty} \mu^*(\sepmodel(\bigwedge_{i=1}^n \ph\{x/t_i\})) \\
= & \; \textstyle\lim_{n\to\infty} \mu(\bigwedge_{i=1}^n\ph\{x/t_i\}),
\end{align*}
and so $\mu$ is Gaifman, by Proposition~\ref{prop:GCA}.
\qed
\end{proof}

A probability $\mu^* : \widehat\B \to \SetR$ on sets of separating interpretations
has the following intended meaning:
\begin{quote}
For a Borel set $B\in\widehat\B$, $\mu^*(B)$ is the degree of
belief that the intended (separating) interpretation is a
member of $B$.
\end{quote}

Next we show that a Gaifman probability on sentences gives a
probability on separating interpretations.

\begin{proposition}[Gaifman $\mu\Rightarrow\mu^*$ separating]\label{prop:mutomustarseparating}
Let the alphabet be countable, $\S$ the set of sentences,
$\widehat\I$ the set of separating interpretations,
and $\widehat\B$ the Borel $\sigma$-algebra on $\widehat\I$.
Let $\mu : \S \rightarrow \SetR$ be a Gaifman probability on sentences.
Then there exists a unique probability
$\widehat\mu^* : \widehat\B \rightarrow \SetR$ such that
\beqn
  \widehat\mu^*(\sepmodel(\ph)) = \mu(\ph),
\eeqn
for each $\ph \in \S$.
\end{proposition}

\begin{proof}
Let $\I$ be the set of interpretations.
Then $\I \setminus \widehat\I$ is the set of all non-separating interpretations.
First we show that $\I \setminus \widehat\I$ is $\B$-measurable.
Let $r$ and $s$ be closed terms of the same function type, say,  $\alpha \rightarrow \beta$, and
$t_1, t_2, ... $ an enumeration of all closed terms of type $\alpha$.
Then
\beqn
\model(r \neq s) \cap \bigcap_{i=1}^\infty \model((r \; t_i) = (s \; t_i))
\eeqn
is a measurable set of non-separating interpretations.
Since there are countably many such pairs $r$ and $s$, and since
\beqn
\I \setminus \widehat\I =
    \bigcup_{r,s} \left(\model(r \neq s) \cap
                \bigcap_{i=1}^\infty \model((r \; t_i) = (s \; t_i))\right),
\eeqn
it follows immediately that $\I \setminus \widehat\I$ is measurable.

According to Proposition~\ref{prop:mutomustar}, there is a unique probability
$\mu^* : \B \rightarrow \SetR$ such that
\beqn
  \mu^*(\model(\ph)) = \mu(\ph),
\eeqn
for each $\ph \in \S$.
We now show that $\mu^*(\I \setminus \widehat\I) = 0$:
\begin{align*}
  & \; \textstyle\mu^*(\model(r \neq s) \cap \bigcap_{i=1}^\infty \model((r \; t_i) = (s \; t_i))) \\
= & \; \textstyle\mu^*(\bigcap_{i=1}^\infty \model((r \; t_i) = (s \; t_i))) - \mu^*(\model(r = s)) \\
= & \; \textstyle\lim_{n \rightarrow \infty} \mu^*(\bigcap_{i=1}^n \model((r \; t_i) = (s \; t_i)))
                   - \mu(r = s) & \text{[$\mu^*$ is countably additive]} \\
= & \; \textstyle \lim_{n \rightarrow \infty} \mu^*(\model(\bigwedge_{i=1}^n((r \; t_i) = (s \; t_i)))) - \mu(r = s) \\
= & \; \textstyle \lim_{n \rightarrow \infty} \mu(\bigwedge_{i=1}^n((r \; t_i) = (s \; t_i))) - \mu(r = s) \\
= & \; \textstyle \mu(r = s) - \mu(r = s) & \text{[$\mu$ is Gaifman]} \\
= & \;  0.
\end{align*}
Hence $\mu^*(\I \setminus \widehat\I) = 0$.

Note that $\widehat\B \subseteq \B$, since $\widehat\I$ is measurable.
Define $\widehat\mu^* : \widehat\B \rightarrow \SetR$ to be
the restriction of $\mu^*$ to $\widehat\B$.
Then, for each $\ph \in \S$,
\begin{align*}
  & \; \widehat\mu^*(\sepmodel(\ph)) \\
= & \; \mu^*(\model(\ph) \cap \widehat\I) \\
= & \; \mu^*(\model(\ph)) - \mu^*(\model(\ph) \cap (\I \setminus \widehat\I)) \\
= & \;  \mu^*(\model(\ph))  & \text{[$\mu^*(\I \setminus \widehat\I) = 0$]} \\
= & \;  \mu(\ph).
\end{align*}
Also $\widehat\mu^*(\widehat\I) = \mu^*(\widehat\I) = \mu^*(\I) - \mu^*(\I \setminus \widehat\I) = \mu^*(\I) = 1$,
so that $\widehat\mu^*$ is a probability.
\qed
\end{proof}

Propositions~\ref{prop:mustartomuseparating} and \ref{prop:mutomustarseparating} and imply

\begin{corollary}[$\mu^*(\I\setminus\widehat\I)=0\Leftrightarrow\mu$ Gaifman]\label{cor:mueqmuseparating}
For countable alphabet and any probability $\mu:\S\to\SetR$ on
sentences and probability $\mu^*:\B\to\SetR$ on
interpretations (one-to-one) related by $\mu^*(\model(\ph))=\mu(\ph)$ it
holds that: $\mu^*(\I\setminus\widehat\I)=0\Leftrightarrow\mu$ Gaifman.
\end{corollary}

There is a concept of being strongly Cournot for probabilities on sets
of interpretations that corresponds to that of being
strongly Cournot for probabilities on sentences.

\begin{definition}[strongly Cournot $\mu^*$]\label{def:stronglyCournotmustar}
A probability $\mu^*:\B\to\SetR$ is
{\em strongly Cournot} if, for each $\ph \in \S$, $\ph$ is
satisfiable implies $\mu^*(\model(\ph)) > 0$.
\end{definition}

\begin{proposition}[strongly Cournot $\mu^*\Leftrightarrow\mu$]\label{prop:stronglyCournotMSeqM}
Let $\S$ be the set of sentences and $\I$ the set of interpretations.
Suppose that $\mu^*: \B \to \SetR$, a probability on the Borel
$\sigma$-algebra on $\I$, and $\mu:\S\to\SetR$, a probability on sentences, are
related by
\beqn
  \mu(\ph) = \mu^*(\model(\ph)),
\eeqn
for each $\ph \in \S$.
Then $\mu$ is a strongly Cournot probability on sentences
iff $\mu^*$ is a strongly Cournot probability on sets of interpretations.
\end{proposition}

\begin{proof}
Suppose that $\mu$ is a strongly Cournot probability on sentences.
Let $\ph$ be a satisfiable sentence.
Then $\mu^*(\model(\ph)) = \mu(\ph) > 0$,
and so $\mu^*$ is a strongly Cournot probability.

Conversely, suppose that $\mu^*$ is a strongly Cournot probability on sets of interpretations.
Let $\ph$ be a satisfiable sentence.
Then $\mu(\ph) = \mu^*(\model(\ph)) > 0$,
and so $\mu$ is a strongly Cournot probability.
\qed
\end{proof}

As with probabilities on sentences, we can also define a Cournot condition for probabilities on sets
of separated interpretations.

\begin{definition}[Cournot $\mu^*$]\label{def:Cournotmustar}
A probability $\mu^*: \B\to\SetR$ is
{\em Cournot} if, for each $\ph \in \S$, $\ph$
has a separating model implies $\mu^*(\model(\ph)) > 0$.
\end{definition}

Clearly every strongly Cournot probability is Cournot.

\begin{proposition}[Cournot $\mu^*\Leftrightarrow\mu$]\label{prop:CournotMSeqM}
Let $\S$ be the set of sentences and $\I$ the set of interpretations.
Suppose that $\mu^*: \B \to \SetR$, a probability on the Borel
$\sigma$-algebra $\B$ on $\I$, and $\mu:\S\to\SetR$, a probability on sentences, are
related by
\beqn
  \mu(\ph) = \mu^*(\model(\ph)),
\eeqn
for each $\ph \in \S$.
Then $\mu$ is a Cournot probability on sentences
iff $\mu^*$ is a Cournot probability on sets of interpretations.
\end{proposition}

\begin{proof}
Suppose that $\mu$ is a Cournot probability on sentences.
Let $\ph$ be a sentence having a separating model.
Then $\mu^*(\model(\ph)) = \mu(\ph) > 0$,
and so $\mu^*$ is a Cournot probability.

Conversely, suppose that $\mu^*$ is a Cournot probability on sets of interpretations.
Let $\ph$ be a sentence having a separating model.
Then $\mu(\ph) = \mu^*(\model(\ph)) > 0$,
and so $\mu$ is a Cournot probability.
\qed
\end{proof}

\section{Existence of Probabilities}\label{Existence of Probabilities}

Now we turn to the issue of the existence of probabilities.

\begin{definition}[discrete $\mu^*$]
A probability $\mu^* : \B \rightarrow \SetR$ is
{\em discrete} if there exists a countable set of
interpretations $\{ I_i \}_{i=1}^\infty$ and a set of
non-negative real numbers $\{ m_i \}_{i=1}^\infty$ such that
$\sum_{i=1}^\infty m_i = 1$ and, for each Borel set $B$,
$\mu^*(B) = \sum_{i: I_i \in B} m_i$.
\end{definition}

Each $m_i$ is called a {\em mass}\index{mass}.
Clearly, a discrete probability {\em is} a probability on the Borel $\sigma$-algebra $\B$.
The set $\{ I_i \}_{i=1}^\infty$ is called the {\em support}\index{support} of the probability.

\begin{theorem}[Cournot and Gaifman probability]\label{thm:CandG}
If the alphabet is countable, there exists a probability on sentences that is Cournot and Gaifman.
\end{theorem}

\begin{proof}
Consider an enumeration $\chi_1, \chi_2, ... $ of the countable
set of sentences which have a separating model.
Choose a separating
interpretation $I_i$ in $\model(\chi_i)$ and assign the mass
$m_i=\frac{1}{i(i+1)}$ to $I_i$, for $i=1,2,...$ .

Define $\mu^* : \B \rightarrow \SetR$ to be the discrete probability
defined by the masses assigned to this
countable set of interpretations. That is, for a Borel set
$B\in\B$, $\mu^*(B)=\sum_{i:I_i\in B}{1\over i(i+1)}$ is
the sum of the masses of the subset of separating interpretations
in $\{I_i\}_{i=1}^\infty$ that are members of $B$. It is
possible that the same interpretation is chosen for
more than one $\model(\chi_i)$; in this case, the masses
corresponding to each choice of that interpretation
are added together.
$\mu^*$ is a probability, since it is a countable sum
of point masses, and $\mu^*(\I)=\sum_{i=1}^\infty{1\over
i(i+1)}=1$.
Since, for all $i$, $\mu^*(\model(\chi_i)) \geq \frac{1}{i(i+1)}>0$, $\mu^*$ is Cournot.

Now define $\mu : \S \rightarrow \SetR$ by $\mu(\ph) = \mu^*(\model(\ph))$, for $\ph \in \S$.
By Proposition~\ref{prop:mustartomu}, $\mu$ is a probability on sentences.
Also, by Proposition~\ref{prop:CournotMSeqM}, $\mu$ is Cournot.
Finally, note that, if $\I$ is the set of interpretations and $\widehat\I$ the set of separating
interpretations, then $\mu^*(\I \setminus \widehat\I) = 0$.
Consequently, the restriction of $\mu^*$ to $\widehat\B$ is a probability on $\widehat\B$ and
$\mu(\ph) = \mu^*(\sepmodel(\ph))$, for $\ph \in \S$.
Thus, by Proposition~\ref{prop:mustartomuseparating}, $\mu$ is Gaifman.
\qed
\end{proof}

Note that the support of the discrete probability $\mu^*$
constructed in Theorem~\ref{thm:CandG} is a dense subset
of $\widehat\I$, since there is a point from the support of
the probability in each set in a basis for its topology.
Every class of separating models that can be characterized by a
finite number of axioms can also be characterized by a single
sentence, hence is assigned a non-zero probability.

\begin{proposition}[strongly Cournot probability]\label{thm:SC}
If the alphabet is countable, there exists a probability on sentences that is strongly Cournot.
\end{proposition}

\begin{proof}
Consider an enumeration $\chi_1, \chi_2, ... $ of the countable
set of sentences which have a model.
Choose an interpretation $I_i$ in $\model(\chi_i)$ and assign the mass
$\frac{1}{i(i+1)}$ to $I_i$, for $i=1,2,...$ .

Define $\mu^* : \B \rightarrow \SetR$ to be the discrete probability
defined by $\mu^*(B)=\sum_{i:I_i \in B}{1\over i(i+1)}$ for $B \in \B$.
$\mu^*$ is a probability, since it is a countable sum
of point masses, and $\mu^*(\I)=\sum_{i=1}^\infty{1\over
i(i+1)}=1$.
Since, for all $i$, $\mu^*(\model(\chi_i)) \geq \frac{1}{i(i+1)}>0$, $\mu^*$ is strongly Cournot.

Now define $\mu : \S \rightarrow \SetR$ by $\mu(\ph) = \mu^*(\model(\ph))$, for $\ph \in \S$.
By Proposition~\ref{prop:mustartomu}, $\mu$ is a probability on sentences.
Also, by Proposition~\ref{prop:stronglyCournotMSeqM}, $\mu$ is strongly Cournot.
\qed
\end{proof}

Now we give some illustrative examples concerning the various
classes of probabilities that have been introduced.

\begin{example}[a probability which is not Gaifman]\label{ex:NotGaifman}
Choose an alphabet for which there exists a non-separating
interpretation. Construct $\mu^*$ by putting unit mass on some
non-separating interpretation. The probability on sentences
corresponding to $\mu^*$ is not Gaifman by
Corollary~\ref{cor:mueqmuseparating}.

Here is such an alphabet and interpretation.
Let there be no non-logical constants in the alphabet.
Let the interpretation $I$ be the standard model defined as follows.
The domain $\D_\imath = \{d\}$.
Each $\D_{\alpha \rightarrow \beta}$ consists of {\em all} functions from $\D_{\alpha}$ to
$\D_{\beta}$.
Note that $d$ is not the denotation of any closed term of type $\imath$.
Now consider $\lambda x. \top$ and $\lambda x. \bot$, each having type $\imath \rightarrow \tB$.
Clearly $\V(\lambda x. \top, I) \neq \V(\lambda x. \bot, I)$.
However, there does not exist a closed term $t$ of type $\imath$ such that
$\V((\lambda x. \top \; t), I) \neq \V((\lambda x. \bot \; t), I)$.
Hence $I$ is not a separating interpretation.
\eoe\end{example}

Theorem~\ref{thm:CandG} shows that, for any countable alphabet, there is always a probability which is Cournot and Gaifman.
The next example shows that it is not guaranteed that there is a probability which is strongly Cournot and Gaifman,
because these two concepts may conflict on non-separating interpretations.

\begin{example}[a probability which is strongly Cournot but not Gaifman]\label{ex:StronglyCournotNotGaifman}
Choose an alphabet for which there exists a non-separating interpretation.
Construct $\mu^*$ by forming an enumeration $\ph_1, \ph_2, \ldots $ of all satisfiable sentences,
and putting mass $\frac{1}{2}$ on some non-separating interpretation
and for each $i$ mass $\frac{1}{(i+1)(i+2)}$ on an interpretation in $\model(\ph_i)$.
The probability on sentences corresponding to $\mu^*$ is strongly Cournot, but not Gaifman.
\eoe\end{example}

\begin{example}[a probability which is Gaifman but not Cournot]\label{ex:GaifmanNotCournot}
Choose an alphabet for which there exist two disjoint sentences each having a separating model.
Construct $\mu^*$ by putting unit mass on a separating model of one of the sentences.
The probability on sentences corresponding to $\mu^*$ is Gaifman but not Cournot.

Here is such alphabet and pair of sentences.
Let $d$ be any element, $\D_\imath = \{d\}$, and, for definiteness, each domain $\D_{\alpha \rightarrow \beta}$
the set of all functions from $\D_\alpha$ to $\D_\beta$.
Each of the domains $\D_\alpha$ is finite.
Let there be a non-logical constant $a$ of type $\imath$ such that $V(a) = d$.
The domain $\D_{\imath \rightarrow\tB}$ consists of two functions, one that maps $d$ to $\mathsf{T}$
and is the denotation of $\lambda x. \top$, and one that maps $d$ to $\mathsf{F}$
and is the denotation of $\lambda x. \bot$.
For each element of each of the domains $\D_{\alpha \rightarrow \beta}$ (other than $\D_{\imath \rightarrow\tB}$)
introduce a non-logical constant of a suitable type into the alphabet in such a way that the denotation of the constant
is the corresponding function.
Note that every element of every domain is the denotation of a closed term.
Now introduce a non-logical constant $p$ of type $\imath \rightarrow\tB$.
For the interpretation $I_1$, take everything defined so far and give $p$ the denotation $d \mapsto \mathsf{T}$.
Then $I_1$ is a separating model of the sentence $(p \; a)$.
On the other hand, for the interpretation $I_2$ take everything defined so far except give $p$ the denotation $d \mapsto \mathsf{F}$.
Then $I_2$ is a separating model of the sentence $\neg (p \; a)$.
Finally, note that $(p \; a)$ and $\neg (p \; a)$ are disjoint.

Example~\ref{ex:SINat} below provides another such alphabet and sentence,
but with infinite domain $\D_\imath=\{0,1,2,...\}$:
There, $\forall x.(B~x)$ and
$\neg\forall x.(B~x)$ each have a separating model, say
$\widehat I$ and $\widehat I'$. Hence we can set
$\mu^*(\widehat I')=1$, which implies $\mu(\forall x.(B~x))=0$
and so $\mu$ cannot confirm $\forall x.(B~x)$. Note that $\mu$
is Gaifman by Corollary~\ref{cor:mueqmuseparating} but not
Cournot.
\eoe\end{example}

\begin{example}[a probability which is Cournot but not strongly Cournot]\label{ex:CournotNotStronglyCournot}
Choose an alphabet for which there is a sentence having a non-empty set of models all of which are non-separating.
Construct $\mu^*$ by forming an enumeration $\ph_1, \ph_2, \ldots $ of all sentences that have a separating model
and putting mass $\frac{1}{i(i+1)}$ on a separating interpretation in $\model(\ph_i)$, for each $i$.
The probability on sentences corresponding to $\mu^*$ is Cournot but not strongly Cournot.

Here is such an alphabet and sentence.
Let the alphabet contain the non-logical constants $a$ of type $\imath$ and $p$ of type $\imath \rightarrow\tB$.
Consider the sentence $\ph \equiv \exists x.(\neg (p \; x) \land (p \; a))$, which has a model.
Let $I$ be any model for $\ph$.
Then for $I$ the domain $\D_\imath$ must have at least two elements,
one of which is the denotation of $a$ and where none of the others is the denotation of a closed term of type $\imath$.
Clearly $\V(p, I) \neq \V(\lambda x. \top, I)$.
However, there does not exist a closed term $t$ of type $\imath$ such that
$\V((p \; t), I) \neq \V((\lambda x. \top \; t), I)$.
Hence $I$ is not a separating interpretation.
\eoe\end{example}

\begin{example}[standard interpretation of $\Nat$]\label{ex:SINat}
This continues Example~\ref{ex:Nat}. As non-logical constants
in our theory we consider $\tfn{0}:\Nat$ and $S:\Nat\to\Nat$,
and abbreviate $\tfn{n}\equiv
S^n(\tfn{0})=(S\;(S\;(S \cdots(S\;\tfn{0}))))$.
The standard interpretation $I$ is defined as follows:
The domain $\D_\Nat=\{0,1,2,...\}$,
and each domain $\D_{\alpha \rightarrow \beta}$
is the set of all functions from $\D_\alpha$ to $\D_\beta$.
We interpret $\V(\tfn{n},I)=n$ and
$V(S):\D_\Nat\to\D_\Nat$ is the successor function mapping $n$ to $n+1$.
This interpretation satisfies the Peano axioms
$\forall x.(S~x)\neq\tfn{0}$ and $\forall x.\forall y.((S~x)=(S~y))\to (x=y)$
and $\forall p.(((p~\tfn{0})\wedge\forall x.((p~x)\to(p~(S~x))))\to\forall x.(p~x))$.
We can add to our logic any number of constants of type
$\Nat\to\tB$. Let $\cal J$ be the set of interpretations
obtained by augmenting $I$ with any valuation of these new
constants. Every interpretation in $\cal J$ (still) satisfies
the Peano axioms. Here and in later examples we only add one
such predicate $B:\Nat\to\tB$, used for induction.
For any probability $\mu^*$ that concentrates on $\cal J$,
i.e.\ $\mu^*({\cal J})=1$, $\mu(\forall x.(\ph~x)) = \lim_{n
\rightarrow \infty}
\mu((\ph~\tfn{0})\wedge...\wedge(\ph~\tfn{n}))$ holds for every
closed term $\ph$ of type $\Nat\to\tB$, and in particular for
$B$.
\eoe\end{example}

\begin{example}[non-standard interpretation of $\Nat$]\label{ex:NSINat}
Consider Example~\ref{ex:SINat} and modify the interpretation $I$ to $I'$ as follows:
Expand $\D_\Nat$ to $\D_\Nat=\{0,1,2,...\}\cup\{...,-\tilde 2,-\tilde 1,\tilde 0,\tilde 1,\tilde 2,...\}$
and $V(S)$ mapping $\tilde n\mapsto\widetilde{\smash{n\!+\!1}}$ in addition to $n\mapsto n+1$.
We call $\tilde n\in\{...,-\tilde 2,-\tilde 1,\tilde 0,\tilde 1,\tilde 2,...\}$,
non-standard numbers.
%
As before, augment $I'$ by an interpretation of $B$. Here we
only consider valuations $V(B)$ that are true everywhere,
except on a single non-standard number, say $\tilde c$. This leads to
a non-separating interpretation $I'$, since $\exists x.\neg(B~x)$ is valid in $I'$
but there is no closed term $t$ for which $\neg(B~t)$ is. Note that every
closed term of type $\Nat$ has some standard number $n$ as denotation.
%
For a point probability $\mu^*$ that concentrates on $I'$ we
therefore have $\mu(\forall x.(B~x))=0$ but
$\mu((B~t))=1$ for all closed terms $t$ of type $\Nat$.
Hence $\mu$ is not Gaifman and cannot confirm $\forall
x.(B~x)$. Note that $I'$ even satisfies the ``Peano'' axioms if
either $\forall p$ is replaced by ``for all closed terms $p$ of type
$\Nat\to\tB$'' or a suitable subset of
$\{\mathsf{T},\mathsf{F}\}^{\D_\Nat}$ is chosen for $\D_{\Nat\to\tB}$.
(this is due to the absence of $+$ and $\times$).
\eoe\end{example}

\begin{example}[the description operator $\iota$]\label{ex:iota}
We can use the previous Example~\ref{ex:NSINat} to illustrate
the complications a description operator $\iota$ causes. Let
constant $\iota_{(\Nat\rightarrow\tB)\rightarrow\Nat}$ denote a
function that selects the unique member of a singleton set
($(\iota~(\lambda x.(y=x)))=y$). Since $\V((\iota~\neg
B),I')=\tilde c$, $(\iota~\neg B)=\tfn{n}$ is not valid in $I'$
for any standard number, and $\mu(B~(\iota~\neg B))=0$. Indeed,
$\iota$ makes accessible all non-standard numbers via
$\widetilde{\smash{c\!+\!k}}=S^k(\iota~\neg B)$ and
$\widetilde{\smash{c\!-\!k}}=(\iota~\lambda x.(\neg B~S^k(x)))$.
Hence $I'$ is now separating for type $\Nat$ and all
non-standard numbers must be included in the enumeration of
terms in the Gaifman condition, even if we only care about the
standard interpretation. We do not know how to avoid this
problem, e.g.\ adding additional axioms that constrain $\iota$.
On the other hand, $\iota$ can easily be eliminated from the
logic (the basic idea is that formulas like $(p~(\iota~B))$ can
be replaced by something like
$(\exists!x.(B~x)\wedge(p~x))\vee(\neg\exists!x.(B~x)\wedge(p~\tfn{0}))$).
\eoe\end{example}

At least asymptotically, the Cournot and Gaifman
probabilities constructed in the proof of
Theorem~\ref{thm:CandG} are good priors for sentences,
since they are non-dogmatic \cite{gaifman-snir}. We will use them
in Sections~\ref{sec:entropy} and \ref{sec:Extension}, called
$\prior$ there, to construct minimally more informative
distributions given some background knowledge like
non-logical axioms.

After having seen various examples of (non)Cournot and
(non)Gaifman probabilities, we now give a general
characterization of Gaifman and Cournot probabilities.

\begin{definition}[rigid mixture representation]\label{def:mixture}
Let $\chi_1,\chi_2,...$ be an enumeration of all sentences that have a separating model.
We say that a probability $\mu:\S\to\SetR$ on sentences has a
mixture representation {\em iff}
$\mu(\ph)=\sum_{i=1}^\infty m_i\mu_i(\ph)$ for some $\{m_i>0\}$
and $\sum_i m_i=1$ and probabilities $\mu_i$ satisfying
$\mu_i(\chi_i)=1$ (hence $\mu_i(\neg\chi_i)=0$).
\end{definition}

\begin{theorem}[probability characterization - Gaifman and Cournot]\label{thm:charGandC}\hfill\\
Let $\mu$ be a probability on sentences. Then\vspace{-0.5ex}
\beqn
  \displaystyle{\text{$\mu$ is Cournot}\atop\text{(and Gaifman)}}
  \quad\Leftrightarrow\quad
  \displaystyle{\text{$\mu$ has a rigid mixture representation}\atop
  \text{(and all $\mu_i$ in Definition~\ref{def:mixture} are Gaifman)}}
\eeqn
\end{theorem}

This result eases the construction of Cournot $\mu$, in that it
reduces the problem of finding a single $\mu$ that simultaneously
satisfies the infinitely many conditions $\mu(\chi_i)>0$ $\forall\chi_i$ to the
problem of finding infinitely many probabilities $\mu_i$ with
each only satisfying one constraint $\mu_i(\chi_i)>0$.

For instance, as in the proof of Theorem~\ref{thm:CandG}, for
any $I_i\in\sepmodel(\chi_i)$,
$\mu_i(\ph):=\ind{I_i\in\sepmodel(\ph)}$ satisfies
$\mu_i(\chi_i)=1$. This also shows that some Cournot (and
Gaifman) $\mu$ can be built purely from deterministic measures
$\mu_i\in\{0,1\}$, i.e.\ sets of models.
%
Corollary~\ref{cor:GandC} below illustrates more generally how
Theorem~\ref{thm:charGandC} can help.

\begin{proof} With the notation of Definition~\ref{def:mixture} we have:\\
({\it Cournot$\Leftarrow$}) Assume $\ph$ has a separating model.\\
Then $\ph=\chi_i$ for some $i$, and hence $\mu(\ph)=\mu(\chi_i)\geq m_i\mu_i(\chi_i)>0$.\\
({\it\&Gaifman$\Leftarrow$}) A linear combination $\mu$ of Gaifman $\mu_i$ is itself Gaifman.\\
({\it Cournot$\Rightarrow$})
Consider $\SetN$-partition \\
${\cal T}:=\{i\in\SetN:\mu(\chi_i)=1\}$,\\
${\cal E}:=\{i\not\in{\cal T}: \chi_i$ starts with an even (incl.\ zero) number of negations $\neg\}$,\\
${\cal O}:=\{i\not\in{\cal T}: \chi_i$ starts with an odd number of negations $\neg\}$.\\
and let $c:{\cal E}\to{\cal O}$ biject $\chi_{c(i)}=\neg\chi_i$.
Let $\ph$ be an arbitrary sentence.
\beqn
  \mbox{For } i\in{\cal E}:~~~\mu(\ph) \;=\;
    \underbrace{\mu(\ph|\chi_i)}_{=:\mu_i(\ph)}
    \underbrace{\mu(\chi_i)}_{=:p_i>0} +
    \underbrace{\mu(\ph|\neg\chi_i)}_{=:\mu_{c(i)}(\ph)}
    \underbrace{\mu(\neg\chi_i)}_{=1-p_i>0}
\eeqn
Let $\sum_{i\in\cal E}r_i=1$ and $r_i>0$ and $m_i=\fr12r_i p_i>0$ and $m_{c(i)}=\fr12 r_i(1-p_i)>0$ for $i\in\cal E$.
Then
\beqn
  \mu(\ph) \;=\; \sum_{i\in\cal E}r_i\mu(\ph)
     \;=\; \sum_{i\in\cal E} r_i[p_i\mu_i(\ph)+(1-p_i)\mu_{c(i)}(\ph)]
     \;=\; \sum_{i\in{\cal E}\dot\cup{\cal O}} 2m_i\mu_i(\ph)
\eeqn
For $i\in\cal T$ define $\mu_i(\ph):=\mu(\ph)\equiv\mu(\ph|\chi_i)$ and $\sum_{i\in\cal T}m_i=\fr12$ with $m_i>0$.
Then $\mu(\ph)=\sum_{i\in\cal T}2m_i\mu_i(\ph)$.
Adding both representations gives
\beqn
   \mu \;=\; \fr12[\mu+\mu]
    \;=\; \fr12[\sum_{i\in{\cal E}\dot\cup{\cal O}}2m_i\mu_i(\ph)+\sum_{i\in\cal T}2m_i\mu_i(\ph)]
    \;=\; \sum_{i=1}^\infty m_i\mu_i
\eeqn
with $\sum_{i=1}^\infty m_i=1$, $m_i>0$, $\mu_i(\chi_i)=1$ as needed.\\
({\it\&Gaifman$\Rightarrow$})
$\mu$ Gaifman implies $\mu_i=\mu(\cdot|\chi_i)$ Gaifman.
\qed
\end{proof}

The next theorem is a complete characterization of general
and (strongly) Cournot or Gaifman probabilities on sentences.
It is based on a tree construction: Consider a sequence of
(some or all) sentences $\ph_1, \ph_2, \ph_3, ...$,
arranged in a finite or infinite complete binary tree with all
left (right) children at depth $n$ labeled by $\neg\ph_n$
($\ph_n$) as depicted below. Furthermore, each node stores
the $\mu$-probability of the conjunction $\psi_{n,S}$ of
sentences along the edges from the root to this node.

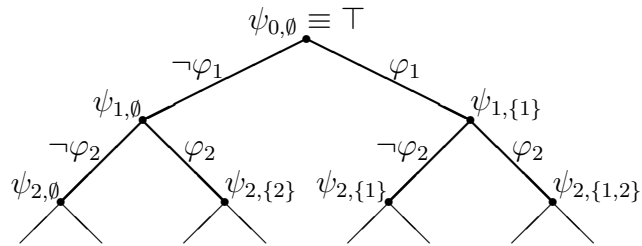
\begin{center}
\unitlength=3ex
\begin{picture}(14,6)(0,0)
\thicklines
\put(8,5){\makebox(0,0)[bc]{$\psi_{0,\emptyset} \equiv \top$}}
\put(8,5){\circle*{0.2}}
\put(8,5){\line(-2,-1){4.0}}
\put(6,4){\makebox(0,0)[br]{$\neg\ph_1$}}
\put(8,5){\line(2,-1){4.0}}
\put(10,4){\makebox(0,0)[bl]{$\ph_1$}}

\put(4,3){\makebox(0,0)[br]{$\psi_{1,\emptyset}$}}
\put(12,3){\makebox(0,0)[bl]{$\psi_{1,\{1\}}$}}
\multiput(0,0)(8,0){2}{
\put(4,3){\circle*{0.2}}
\put(4,3){\line(-1,-1){2.0}}
\put(4,3){\line(1,-1){2.0}}
\put(3,2){\makebox(0,0)[br]{$\neg\ph_2$}}
\put(5,2){\makebox(0,0)[bl]{$\ph_2$}}
}
\put(2,1){\makebox(0,0)[br]{$\psi_{2,\emptyset}$}}
\put(6,1){\makebox(0,0)[bl]{$\psi_{2,\{2\}}$}}
\put(10,1){\makebox(0,0)[br]{$\psi_{2,\{1\}}$}}
\put(14,1){\makebox(0,0)[bl]{$\psi_{2,\{1,2\}}$}}

\thinlines
\multiput(0,0)(4,0){4}{
\put(2,1){\circle*{0.2}}
\put(2,1){\line(-1,-1){1.0}}
\put(2,1){\line(1,-1){1.0}}
}
\end{picture}
\end{center}

\begin{proposition}[$\psi_S\ph$-tree]\label{prop:psiphi} 
For $i=1, ..., n$, let $\ph_i$ be a sentence. For each
$S\subseteq\{1\!:\!n\}\equiv\{1,...,n\}$, define the sentence
$\psi_{n,S}$ by
\beqn
  \psi_{n,S} \equiv (\bigwedge_{i \in S} \ph_i) \wedge (\bigwedge_{j \in \{1:n\} \setminus S} \neg \ph_j).
\eeqn
Then the following hold.
\begin{enumerate}
\item The $\psi_{n,S}$'s are pairwise disjoint.
\item $\bigvee_{S\subseteq\{1:n\}} \psi_{n,S}$ is valid.
\item For each $i=1, ..., n$,
      $\ph_i$ is logically equivalent to $\displaystyle\smash{\bigvee_{S\subseteq\{1:n\}: i \in S} \psi_{n,S}}$.
\end{enumerate}
\end{proposition}

\begin{proof}
Straightforward.
\qed
\end{proof}

The following is our main characterization theorem. It
states necessary and sufficient conditions on the labels
$\a_{n,S}:=\mu(\psi_{n,S})$, for general $\mu$, as well as
(strongly) Cournot $\mu$, and sufficient conditions for Gaifman
$\mu$. We do not yet have a complete tree characterization of
Gaifman probabilities, which is a major open problem. The
characterization can easily be converted to a procedure that
assigns probabilities to one sentence after the other, but it
is not an algorithm, since satisfiability is not decidable.

\begin{theorem}[tree characterization of general/Cournot/Gaifman probabilities]\label{thm:char}
Let the alphabet be countable and $\ph_1, \ph_2, \ph_3, ...$ an enumeration of all sentences.
For each $n \geq 1$ and each $S\subseteq\{1\!:\!n\}$, define the sentence $\psi_{n,S}$ by
\beqn
  \psi_{n,S} \equiv (\bigwedge_{i \in S} \ph_i) \wedge (\bigwedge_{j \in \{1:n\} \setminus S} \neg \ph_j).
\eeqn
\begin{enumerate}
\item Let $\mu$ be a probability on sentences.
      Then, for each $n \geq 1$,
      \beq
         \mu(\ph_n) = \sum_{S\subseteq\{1:n\}: n \in S} \mu(\psi_{n,S}). \label{psinS}
      \eeq
      Furthermore, $\mu$ is Cournot (resp., strongly Cournot) iff, for each $n \geq 1$ and $S\subseteq\{1\!:\!n\}$,
      $\psi_{n,S}$ has a separating model (resp., is satisfiable) implies $\mu(\psi_{n,S}) > 0$.
\item For each $n \geq 1$ and $S\subseteq\{1\!:\!n\}$, let $\a_{n,S} \in \SetR$ satisfy the following conditions.
      \begin{enumerate}
      \item $\a_{n,S} \geq 0$.
      \item If $\psi_{n,S}$ is unsatisfiable, then $\a_{n,S} = 0$.
      \item $\a_{n,S} = \a_{n+1,S} + \a_{n+1, S \cup \{n+1\}}$.
      \item $\sum_{S\subseteq\{1:n\}} \a_{n,S} = 1$.
      \end{enumerate}
      Then there exists a probability $\mu$ on sentences such that, for each $n \geq 1$ and each $S\subseteq\{1\!:\!n\}$,
      \beqn
        \mu(\psi_{n,S}) = \a_{n,S}.
      \eeqn
\item Suppose that, in addition to the conditions in Part 2, the following condition also holds:
      for each $n \geq 1$ and $S\subseteq\{1\!:\!n\}$, $\psi_{n,S}$ has a separating model (resp., is satisfiable) implies $\a_{n,S} > 0$.
      Then $\mu$ is Cournot (resp., strongly Cournot).
\item Suppose that, the conditions of Part 2 hold.
      Strengthen 2b by demanding that if $\psi_{n,S}$ has no separating model, then $\a_{n,S}=0$.
      Further, assume that enumeration $\ph_1,\ph_2,...$ is such that
      if $\ph_{n+1}=[r=s]$ for terms $r$ and $s$ having the same function type, then
      $\ph_{n+2}=\bigvee_{S\subseteq\{1:n\}}\psi_{n,S}\wedge\ph\{x/t_S\}$,
      where $\ph:=[(r\;x)=(s\;x)]$ and
      $t_S$ is such that $\psi_{n,S}\wedge\neg\ph\{x/t_S\}$ has a separating
      model (if no such $t_S$ exists, choose $t_S$ arbitrarily or drop this contribution from $\bigvee$).
      For $\ph_{n+1}=[r=s]$ also set $\a_{n+2,S}=\a_{n+1,S}$.
      Then $\mu$ is Gaifman.
\item For every probability $\mu$, $\a_{n,S}:=\mu(\psi_{n,S})$ satisfies 2(a)-(d).
\end{enumerate}
\end{theorem}

Items 1,2,3,and 5 are rather natural. The somewhat ugly item 4
requires explanation: First, the assumption on the enumeration
$\ph_i$ can easily be satisfied by inserting appropriate
$\ph_{n+2}$ at the required $n$. The intuition behind the
construction for $n=0$ is that if $I$ is a model of
$\neg\ph_1$, i.e.\ of $\exists x.\neg\ph$, Gaifman requires a
witness $t$, which exists by the extensionality axiom. We can
guarantee such a witness by putting $\ph_2=\ph\{x/t\}$ and
following exclusively the $\neg\ph_2$ branch by setting
$\alpha_{2,\{2\}}=0$. For general $n$, the witnesses $t$ and hence
$\ph_{n+2}=\ph\{x/t\}$ may depend on $S$; this would lead to a
branch-dependent enumeration of sentences. There is nothing
wrong with this, and is probably even the preferred solution.
In order to keep things simple, we kept the enumeration branch
independent by or-ing $\ph_{n+2}$ over all $2^n$ branches,
which makes it formally independent of the branch $S$.

\begin{proof}
{\bf 1.} The first part follows immediately from Parts 1 and 3
of Proposition~\ref{prop:psiphi} and
Proposition~\ref{prop:PPS}.6.

The second part for strongly Cournot follows immediately from the definition of a
strongly Cournot probability, Proposition~\ref{prop:psiphi}.3,
and Proposition~\ref{prop:PPS}.4:
That strongly Cournot implies $\mu(\psi_{n,S})>0$ for satisfiable
$\psi_{n,S}$ is trivial. For the other direction, $\ph_n$
is satisfiable implies that there exists an $S \ni n$ for which
$\psi_{n,S}$ is satisfiable.
Hence $\mu(\ph_n)\geq\mu(\psi_{n,S})>0$.
Thus $\mu(\ph)>0$, for all satisfiable $\ph$, and so $\mu$ is strongly Cournot.
The proof for the Cournot case is similar.

{\bf 2.} First define $\mu_0:\{\psi_{n,S}\}_{n\geq 1,S\subseteq\{1\!:\!n\}}\to\SetR$ by
\beqn
  \mu_0(\psi_{n,S}) \;=\; \a_{n,S},
\eeqn
for each $n \geq 1$ and $S\subseteq\{1\!:\!n\}$.
We prove by induction that, for $m \geq n$,
\beqn
  \mu_0(\psi_{n,S}) \;=\; \sum_{R: S \subseteq R \subseteq S \cup \{n+1, ..., m \}} \a_{m, R}.
\eeqn
The result is obvious when $m = n$.
Suppose now it holds for $m$.
Then
\begin{align*}
  & \; \mu_0(\psi_{n,S}) \\
= & \; \textstyle\sum_{R: S \subseteq R \subseteq S \cup \{n+1, ..., m \}} \a_{m, R} & \text{[Induction hypothesis]} \\
= & \; \textstyle\sum_{R: S \subseteq R \subseteq S \cup \{n+1, ..., m \}} (\a_{m+1, R} + \a_{m+1, R \cup \{m+1\}}) \\
= & \; \textstyle\sum_{R: S \subseteq R \subseteq S \cup \{n+1, ..., m+1 \}} \a_{m+1, R}.
\end{align*}
This completes the induction argument.

Now define $\mu:\S\to\SetR$ by
\beqn
  \mu(\ph_n) \;=\; \sum_{S\subseteq\{1:n\}: n \in S} \a_{n,S}.
\eeqn
for each $n \geq 1$.
We prove by induction that, for $m \geq n$,
\beqn
\mu(\ph_n) = \sum_{S\subseteq\{1:m\}: n \in S} \a_{m,S}.
\eeqn
The result is obvious when $m = n$.
Suppose now it holds for $m$.
Then
\begin{align*}
  & \; \mu(\ph_n) \\
= & \; \textstyle\sum_{S\subseteq\{1:m\}: n \in S} \a_{m,S} & \text{[Induction hypothesis]} \\
= & \; \textstyle\sum_{S\subseteq\{1:n\}: n \in S} (\a_{m+1,S} + \a_{m+1, S \cup \{m+1\}}) \\
= & \; \textstyle\sum_{S\subseteq\{1:m+1\}: n \in S} \a_{m+1,S}.
\end{align*}
This completes the induction argument.

We show that $\mu$ extends $\mu_0$.
Suppose that  $\psi_{n,S}$, for some $n \geq 1$ and $S\subseteq\{1\!:\!n\}$, is $\ph_k$, for some $k \geq 1$.
Let $m=\max \{k, n\}$ and
$\mathcal{A} = \{R: k\in R\subseteq\{1, ..., m\}\}$ %
and $\B = \{ R:S \subseteq R \subseteq S \cup \{n+1, ..., m \} \}$. %
Then $\bigvee_{R \in \mathcal{A}} \psi_{m, R}$ is logically
equivalent to $\ph_k$ which is equal to $\psi_{n,S}$ which
is logically equivalent to $\bigvee_{R \in \B}
\psi_{m, R}$. Also the $\psi_{m, R}$ are pairwise disjoint.
Hence $\psi_{m,R}$ is unsatisfiable (and so $\a_{m,R} = 0$)
for each $R \in (\mathcal{A} \setminus \B) \cup
(\B \setminus \mathcal{A})$. This implies
\beqn
  \mu(\psi_{n,S}) \;=\; \mu(\ph_k)
  \;=\; \sum_{R\in\mathcal{A}} \a_{m,R}
  \;=\; \sum_{R\in\B} \a_{m,R}
  \;=\; \mu_0(\psi_{n,S}).
\eeqn
In summary, $\mu:\S\to\SetR$ is well-defined and satisfies
\beqn
  \mu(\ph_n) = \sum_{S\subseteq\{1:n\}: n \in S} \mu(\psi_{n,S}),
\eeqn
for each $n \geq 1$.

We show that $\mu$ is a probability.
Clearly, $\mu$ is non-negative.
Now suppose that, for some $n \geq 1$, $\ph_n$ is valid.

Then, for $n \not\in S$, $\psi_{n,S}$ is a conjunction that
contains $\neg\ph_n$, hence is not satisfiable and
therefore $\a_{n,S}=0$ for $n\not\in S$. This implies
\beqn
  \mu(\ph_n) \;=\; \sum_{S\subseteq\{1:n\}: n \in S} \a_{n,S}
  \;=\; \sum_{S\subseteq\{1:n\}} \a_{n,S} \;=\; 1.
\eeqn

Finally, suppose that $\neg (\ph_n \wedge \ph_m)$ is valid.
There exists $k \geq 1$ such $\ph_k$ is $\ph_n \vee \ph_m$.
Choose any $p$ greater than $n$, $m$ and $k$.
Consider $\mathcal{A}:=\{S\subseteq\{1\!:\!p\}: k \in S\}$ and
$\B:=\{S\subseteq\{1\!:\!p\}: n \in S\}$ and
$\mathcal{C}:=\{S\subseteq\{1\!:\!p\}: m \in S\}$.

$\a_{p,S}=0$ for $S\in\B\cap\mathcal{C}$, since
$\psi_{n,S}$ is a conjunction containing $\ph_n\wedge\ph_m$.

$\a_{p,S}=0$ for $\mathcal{A}\setminus(\B\cup\mathcal{C})$, since
$\psi_{n,S}$ is a conjunction containing $\ph_k\wedge\neg\ph_n\wedge\neg\ph_m$.

$\a_{p,S}=0$ for $(\B\cup\mathcal{C})\setminus\mathcal{A}$, since
$\psi_{n,S}$ is a conjunction containing $\neg\ph_k\wedge\ph_n\wedge\ph_m$.

\noindent Together this implies
\beqn
  \mu(\ph_n \vee \ph_m)
  \;=\; \mu(\ph_k)
  \;=\; \sum_{S\in\mathcal{A}} \a_{p,S}
  \;=\; \sum_{S\in\B} \a_{p,S} + \sum_{S\in\mathcal{C}} \a_{p,S}
  \;=\; \mu(\ph_n) + \mu(\ph_m).
\eeqn
Thus $\mu$ is a probability on sentences.

{\bf 3.} For the strongly Cournot case, suppose that, for some $n \geq 1$, $\ph_n$ is satisfiable.
   Thus $\psi_{n, S'}$ is satisfiable for some $S'\subseteq\{1:n\}$ for which $n\in S'$.
   By the condition, $\mu(\psi_{n, S'}) > 0$.
   Hence $\mu(\ph_n) = \sum_{S\subseteq\{1:n\}: n \in S} \mu(\psi_{n,S}) > 0$.
   The Cournot case is similar.

{\bf 4.}
$\exists x.\ph$ has separating model (s.m.) {\em iff} there exists $t$
such that $\ph\{x/t\}$ has s.m.
The $\Rightarrow$ direction follows from Definition~\ref{def:separating} with
$r=\lambda x.\ph$ and $s=\lambda x.\top$.
The $\Leftarrow$ direction follows from $\ph\{x/t\}\to \exists x.\ph$.

We need to show the Gaifman condition in Definition~\ref{def:Gaifman}.
This is equivalent to:
For all terms $r$ and $s$ having the same function type,
$\mu(r=s) \;=\; \lim_{m\to\infty}\mu(\bigwedge_{i=1}^m ((r\;t_i)=(s\;t_i)))$.
Fix $r$ and $s$, define $\ph:=[(r\;x)=(s\;x)]$.
Using the extensionality axiom we hence have to show
\beqn
  \mu(\forall x.\ph) \;=\; \lim_{m\to\infty}\mu(\bigwedge_{i=1}^m \ph\{x/t_i\})
\eeqn
Consider $n$ such that $\ph_{n+1}=[r=s]\equiv\forall x.\ph$.
By assumption,
$\ph_{n+2}=\bigvee_{S\subseteq\{1:n\}}\psi_{n,S}\wedge\ph\{x/t_S\}$.

We first prove that setting $\a_{n+2,S}=\a_{n+1,S}$ is allowed:\\
Assume $\psi_{n,S}\wedge\neg\ph_{n+1}\equiv\exists x.(\psi_{n,S}\wedge\neg\ph)$ has s.m.\\
$\Rightarrow$ There exists $t_S$  s.th.\ $\psi_{n,S}\wedge\neg\ph\{x/t_S\}$ has s.m.\\
$\Rightarrow$ $\psi_{n,S}\wedge\neg\ph_{n+1}\wedge\neg\ph\{x/t_S\}$ has s.m., since $\neg\ph\{x/t_S\}$ implies $\neg\ph_{n+1}$.\\
The last expression is logically equivalent to $\psi_{n,S}\wedge\neg\ph_{n+1}\wedge\neg\ph_{n+2}$,
since for $\psi_{n,S}=\bot$, both expressions are false, and for $\psi_{n,S}=\top$,
$\psi_{n,S'}=\bot$ for all $S'\neq S$, hence $\bigvee_S$ in $\ph_{n+2}$ collapses to $\ph\{x/t_S\}$.
Since $\psi_{n,S}\wedge\neg\ph_{n+1}\wedge\neg\ph_{n+2}$ has s.m.,
$\a_{n+2,S}=\a_{n+1,S}$ is allowed.
Assume now that $\psi_{n,S}\wedge\neg\ph_{n+1}$ has no s.m.
Then $\psi_{n,S}\wedge\neg\ph_{n+1}\wedge\neg\ph_{n+2}$ has neither,
and $\a_{n+2,S}=\a_{n+1,S}=0$.
Hence (4.) is a consistent instantiation of (2.) and generates a probability on
sentences $\mu$ with $\mu(\psi_{n',S'})=\a_{n',S'}$ for all $n'$ and $S'$.
We now prove that it is Gaifman.

For $\mu(\forall x.\ph\wedge\psi_{n,S})>0$, trivially
\beqn
  \mu(\bigwedge_{i=1}^m \ph\{x/t_i\}\,|\,\forall x.\ph\wedge\psi_{n,S}) \;=\; 1
  \;=\; \mu(\forall x.\ph\,|\,\forall x.\ph\wedge\psi_{n,S})
\eeqn
For $\mu(\neg\forall x.\ph\wedge\psi_{n,S})>0$ and sufficiently large $m$,
\beqn
  \mu(\bigwedge_{i=1}^m \ph\{x/t_i\}\,|\,\neg\forall x.\ph\wedge\psi_{n,S}) \;=\; 0
  \;=\; \mu(\forall x.\ph\,|\,\neg\forall x.\ph\wedge\psi_{n,S})
\eeqn
since
$\mu(\neg\ph\{x/t_S\}|\neg\forall x.\ph\wedge\psi_{n,S})
=\mu(\neg\ph_{n+2}|\neg\forall x.\ph\wedge\psi_{n,S})=
\a_{n+2,S}/\a_{n+1,S}=1$,
and $\bigwedge_{i=1}^m \ph\{x/t_i\}$ will eventually contradict $\neg\ph\{x/t_S\}$.

Since both displayed equalities hold for all $S\subseteq\{1\!:\!n\}$,
for sufficiently large $m$ this implies
$\mu(\bigwedge_{i=1}^m \ph\{x/t_i\}) = \mu(\forall x.\ph)$.

{\bf 5.} Straightforward.
\qed
\end{proof}

Unfortunately items 3 and 4 in Theorem~\ref{thm:char} cannot be combined.
The $\mu$ in item 4.\ is not Cournot, since e.g.\
$\neg\ph_{n+1}\wedge\ph_{n+2}$ has a separating model if there
is more than one possible witness $t_S$, but is assigned zero
probability. We can do something else though.

The following corollary boosts Gaifman $\mu$ constructed in Theorem~\ref{thm:char}.4 with
the rigid mixture representation to a Gaifman and Cournot $\mu$,
and this without having to choose interpretations $I$ as required in Theorem~\ref{thm:charGandC}.

\begin{corollary}[Gaifman and Cournot probability]\label{cor:GandC}
Let $\chi_1,\chi_2,...$ be an enumeration of all sentences that have a separating model.
For each $i$, let $\ph_1:=\chi_i,\ph_2,\ph_3,...$ be different (in the first sentence) enumerations of all sentences,
and $\mu_i$ be a corresponding Gaifman probability constructed in Theorem~\ref{thm:char}.4,
choosing $\mu_i(\chi_i)\equiv\a_{1,\{1\}}:=1$. Then by Theorem~\ref{thm:charGandC},
the rigid mixture $\mu$ of Definition~\ref{def:mixture} is Gaifman and Cournot.
\end{corollary}

\section{Relative Entropy of Probabilities on Sentences}\label{sec:entropy}

Assume we ``know'' the probabilities $\mu_0(\ph_i)$ of
sentences $\ph_1,...,\ph_n$. Note that
$\mu_0:\{\ph_1,...,\ph_n\}\to[0,1]$ is {\em not} a probability
on all sentences, but only a partial specification.
In the next section (Proposition~\ref{prop:ExtProb}) we derive
conditions under which $\mu_0$ can be extended to a
probability over all sentences.

However, if there are any solutions at all,
then there are many. It then makes sense to ask whether some
distributions that meet our constraints are ``better'', in some
sense, than others.

A natural idea is to choose $\mu$ in such a way as to be ``as
uninformative as possible'', consistent with our constraints as
defined by $\mu_0$.  Unfortunately it is not possible to define
``as uninformative as possible'' in absolute terms, but we can
define it relative to a prior distribution, $\prior$. We will
formalise this using the concept of relative entropy, or
Kullback-Leibler divergence. We now show that this selection of
$\mu$ has exactly the form of a piecewise re-scaled $\prior$,
and show how to find the optimal rescaling constants, the
various $\a_S$ introduced in the next section, under this criterion.
Natural choices for the prior $\prior$ are the non-dogmatic
probabilities constructed in Theorem~\ref{thm:CandG}.

We start by introducing the relevant concepts on general
measure spaces before constructing the new distribution $\mu$
that meets our constraints while being uninformative relative
to our prior, $\prior$.

From \cite[p.21]{ihara93}\cite{csiszar75}:
Let $\mu^*$ and $\prior^*$ be probabilities on a measurable
space $(\mathbf{X}, \B(\mathbf{X}))$.  We say that
$\mu^*$ is {\em absolutely continuous} with respect to
$\prior^*$, $\mu^* \prec \prior^*$, if $\mu^*(A) = 0$ for every $A \in
\B(\mathbf{X})$ such that $\prior^*(A) = 0$.
By the Radon-Nikodym theorem \cite[Theorem 5.5.4]{dudley}, if $\mu^*$ is absolutely
continuous with respect to $\prior^*$, then there exists a
$\prior^*$-integrable function $\psi(x)$ such that
\beqn
  \mu^*(A) \;=\; \int_A \psi(x) d\prior^*(x), \;\;\; \forall A \in \B(\mathbf{X}).
\eeqn
The function $\psi(x)$ is called the Radon-Nikodym derivative
and is written in the form
\beqn
  \psi(x) \;=\; \frac{d\mu^*}{d\prior^*}(x).
\eeqn

For probabilities $\mu^*$ and $\prior^*$ on $(\mathbf{X},
\B(\mathbf{X}))$, the {\em relative entropy}
$\KL(\mu^*||\prior^*)$ of $\mu^*$ with respect to $\prior^*$ is defined by
\beqn
  \KL(\mu^*||\prior^*) \;:=\; \begin{cases} \int_{\mathbf{X}}
  \log \frac{d\mu^*}{d\prior^*}(x) d\mu^*(x) & \text{if} \; \mu^* \prec \prior^*, \\
  \infty & \text{otherwise.} \end{cases}
\eeqn
The measure $\prior^*$ is referred to as the {\em reference measure}.

By reference to this general definition for relative entropy,
we can define the relative entropy for probabilities on sentences
in two ways:

\begin{definition}[relative entropy on sentences]\label{def:relentropy}
For a countable alphabet and for probabilities $\mu$ and $\prior$
defined on some set of sentences $\S$, the {\em relative
entropy} $\KL(\mu||\prior)$ of $\mu$ with respect to $\prior$ is defined
by
\beqn
  \KL(\mu||\prior) \;:=\; \lim_{n\to\infty} \sum_{S\subseteq\{1:n\}} \mu(\psi_S)\log{\mu(\psi_S)\over\prior(\psi_S)}
  \;=\; \KL(\mu^*||\prior^*)
\eeqn
where $0\log{0\over\prior}:=0$ and $\mu\log{\mu\over 0}:=\infty$ if $\mu>0$.
The last equality holds true if
$\mu^*$ and $\prior^*$ are the probabilities in
Proposition~\ref{prop:mutomustar} on interpretations that
correspond to $\mu$ and $\prior$ respectively.
\end{definition}

The first definition is more general and useful and
conceptually easier. Since the relative entropy increases with
refinement, the limit always exists and is independent of the
order of enumeration of sentences. The second definition is the
``obvious'' choice for a definition, but is more restrictive
and based on much heavier machinery. Equivalence follows from
exchanging limits with integrals, which requires some justification.

\begin{proof} (sketch)
{\bf (i)} {\em Order independence:} Let $\Phi$ be a finite set of
sentences, and $\KL_{\Phi}(\mu||\prior)$ be the relative
entropy of the sentences in $\Phi$. Then by the
monotonicity of the relative entropy under refinement,
$\Phi\subseteq\Phi'$ implies $\KL_{\Phi}\leq\KL_{\Phi'}$.
It is now routine to establish independence of the limit on the
order of enumeration of the sentences.

{\bf (ii)} {\em Equivalence of both definitions:}
For $\mu^*\not\prec\prior^*$ one can show that the limit diverges,
which implies equality. We will only prove the interesting
case when $\mu^*\prec\prior^*$.
Let $\ph_1,\ph_2,...$ be an enumeration of all sentences. For
an interpretation $I\in\cal I$, let $S$ be such that
$\psi_{n,S}$ is valid in $I$, i.e.\ $S\equiv S(n,I) :=
\{i\in\{1,...,n\}:I\in\model(\ph_i)\}$.
Using $\mu(\psi_{n,S})=\mu^*(\model(\psi_{n,S}))=\int_{\model(\psi_{n,S})} d\mu^*$, let
\begin{align*}
  \KL_n(\mu||\prior) \; &:= \; \sum_{S\subseteq\{1:n\}} \mu(\psi_{n,S})\log{\mu(\psi_{n,S})\over\prior(\psi_{n,S})}
  \;=\; \int_{\cal I} \log{\mu(\psi_{n,S})\over\prior(\psi_{n,S})}\, d\mu^* \;\geq\; 0 \\
  \KL^*(\mu||\prior) \; &:= \; \int_{\cal I} \log{d\mu^*\over d\prior^*}\, d\mu^* \;\geq\; 0
\end{align*}
%
Elementary algebra (telescoping property of $\KL$)
allows us to split $\KL^*$ into a finitary and a tail part
\beqn
  \KL^*(\mu||\prior) \;=\; \KL_n(\mu||\prior) \;+
  \sum_{S\subseteq\{1:n\}} \mu(\psi_{n,S})\,\KL^*(\mu(\cdot|\psi_{n,S})\,||\,\prior(\cdot|\psi_{n,S}))
\eeqn
which shows that $\KL^*\geq\KL_n$.

For the other direction,
let ${\cal F}_n$ be the Borel $\sigma$-algebra generated by
$\{\model(\psi_{n,S}):S\subseteq\{1\!:\!n\}\}$.
%
Then ${\cal F}_1\subseteq{\cal F}_2\subseteq...$ is a
filtration with ${\cal F}_\infty=\B$ the Borel $\sigma$-algebra
generated by $\bigcup_{n=1}^\infty{\cal F}_n$. Define
\beqn
  Z_n(I)
  \;:=\; {\mu^*(\model(\psi_{n,S})) \over \prior^*(\model(\psi_{n,S}))}
  \;=\; {\mu(\psi_{n,S}) \over \prior(\psi_{n,S})}
\eeqn
$Z_n:{\cal I}\to\SetR$ is an ${\cal F}_n$ measurable function,
well-defined with $\prior$-probability 1 (w.$\prior$.p.1).
$Z_1,Z_2,...$ forms a $\prior$-martingale sequence, since
\begin{align*}
  \E_\prior[Z_{n+1}|{\cal F}_n] \; =& \; {\mu(\psi_{n+1,S}) \over \prior(\psi_{n+1,S})}\,\prior(\psi_{n+1,S}|\psi_{n,S})
  + {\mu(\psi_{n+1,S\cup\{n+1\}}) \over \prior(\psi_{n+1,S\cup\{n+1\}})}\,\prior(\psi_{n+1,S\cup\{n+1\}}|\psi_{n,S}) \\
  =& \; {\mu(\psi_{n,S}) \over \prior(\psi_{n,S})} \;=\; Z_n
\end{align*}
Since $\mu^*\prec\prior^*$, by \cite[VII\textsection 8]{Doob:53} the sequence
converges to the Radon-Nikodym derivative
\beqn
  \lim_{n\to\infty} Z_n \;=\; {d\mu^*\over d\prior^*} \qquad\text{w.$\prior$.p.1}
\eeqn
%
Now consider
\beqn
  \KL_n(\mu||\prior) \;= \sum_{S\subseteq\{1:n\}} {\mu(\psi_{n,S})\over\prior(\psi_{n,S})} \log{\mu(\psi_{n,S})\over\prior(\psi_{n,S})}\,\prior(\psi_{n,S})
  \;=\; \int_{\cal I} Z_n\log Z_n\,d\prior^*
\eeqn
By Fatou's lemma applied to $1+Z_n\log Z_n$, which is
non-negative, and the existence of the pointwise limit
$Z_n$ w.$\prior$.p.1, we get
\begin{align*}
  & \liminf_{n\to\infty}\KL_n(\mu||\prior)
  \; \geq\; \int_{\cal I} \liminf_{n\to\infty} Z_n\log Z_n\,d\prior^* \\
  \;&=\; \int_{\cal I} {d\mu^*\over d\prior^*}\log{d\mu^*\over d\prior^*}d\prior^*
  \;=\; \int_{\cal I} \log{d\mu^*\over d\prior^*}d\mu^*
  \;=\; \KL^*(\mu||\prior)
\end{align*}
%
Since $\KL_n$ is monotone increasing and together with
$\KL^*\geq\KL_n$, we have
$\lim_{n\to\infty}\KL_n=\KL^*$. This shows the
equivalence of both definitions in
Definition~\ref{def:relentropy}. \qed
\end{proof}

Given some base measure $\prior^*$, we are interested in finding a
measure $\widehat\mu^*$ that minimizes $\KL(\mu^*||\prior^*)$ under some
\beq\label{eq:contraint}
  \text{constraints}\quad \int_{\mathbf{X}} f_i(x)d\widehat\mu^*(x) \;=\; a_i, \quad i=1,...,n.
\eeq
We assume that these constraints are satisfiable for some $\widehat\mu^*\prec\prior^*$.

\cite{ihara93} defines the $\KL$-projection of a
probability under some constraints as the measure that
minimises the relative entropy subject to those constraints. In
practice, the $\KL$-projection is defined by giving a
Radon-Nikodym derivative that re-scales the original
probability to meet the constraints.  This is similar
to the rescaling used in the proof of
Proposition~\ref{prop:ExtProb} below.

\cite[pp.104-5]{ihara93} proves the following:
Define functions $\theta_i(\lambda)$, $i=1,...,n$, of
$\lambda = (\lambda_1, ..., \lambda_n) \in \SetR^n$ by
\begin{align*}
  \theta_i(\lambda) = & \frac{1}{\Phi(\lambda)} \int_{\mathbf{X}} f_i(x)
  \mexp{ \sum_{j=1}^n \lambda_j f_j(x)} \; d\prior^*(x), \qquad i=1,...,n,
\\
  \text{where} \quad & \Phi(\lambda) = \int_{\mathbf{X}} \mexp{ \sum_{j=1}^n
  \lambda_j f_j(x)} \; d\prior^*(x).
\end{align*}
We denote by $\Lambda$ the set of all $\lambda$ for which the
integrals above converge, and define a set
$\mathcal{A}\subseteq(\SetR\cup\{-\infty\})^n$ by
\begin{align*}
  \mathcal{A} &= \{(\theta_1(\lambda), ..., \theta_n(\lambda)); \lambda \in \Lambda \}.
\end{align*}

Let $\mathbf{M}_1$ be the set of all probabilities on
$(\mathbf{X},\B(\mathbf{X}))$, $\prior^* \in \mathbf{M}_1$
be a fixed reference measure, and $f_i(x),\,i=1,...,n$ be
real functions defined on $\mathbf{X}$.  Assume that
$\mathbf{F} \subset \mathbf{M}_1$ is a set of the form
\begin{align*}
  \mathbf{F} & \;=\; \{\mu^* \in \mathbf{M}_1:  \int_\mathbf{X} f_i(x) \; d\mu^*(x) = a_i, \; i=1,...,n\},
\end{align*}
where $a_i$, $i=1,...,n$ are given constants such
that $(a_1,...,a_n) \in \mathcal{A}$.  Then the
$\KL$-projection $\widehat\mu^*$ on $\mathbf{F}$ is given by
\beq\label{eq:mredexp}
  \frac{d\widehat\mu^*}{d\prior^*}(x) \;=\; \frac{1}{\Phi(\lambda)}\mexp{\sum_{i=1}^n \lambda_i f_i(x)},
\eeq
where $\lambda = (\lambda_1,...,\lambda_n) \in
\Lambda$ is a vector uniquely determined by solving
\beqn
  \frac{1}{\Phi(\lambda)} \int_{\mathbf{X}} f_i(x)
  \mexp{\sum_{j=1}^{n} \lambda_j f_j(x)} \; d\prior^*(x) \;=\; a_i, \qquad i=1,...,n.
\eeqn
The corresponding minimum relative entropy is given by
\begin{align*}
  \KL(\widehat\mu^*||\prior^*) &= \sum_{i=1}^n \lambda_i a_i - \log \Phi(\lambda).
\end{align*}

We will construct, where possible, a function $\mu$ that has
minimum relative entropy with respect to $\prior$ while still
satisfying our constraints as represented by $\mu_0$ and the
$\ph_i$, $i=1,...,n$. First, we construct a
function $\mu^*$ on interpretations that will end up meeting
our constraints while minimising the relative entropy to
$\prior^*$.

Choose the $f_i=\ind{\model(\ph_i)}$ as indicator function on $\mathbf{X}=\I$,
which is 1 on models of $\ph_i$ and zero elsewhere. Set $a_i =
\mu_0(\ph_i)$, $i=1,...,n$. The constraints
\req{eq:contraint} then reduce to
\beqn
  \mu(\ph_i) \;=\; \mu^*(\model(\ph_i))
  \;=\; \int_\I\ind{\model(\ph_i)}d\mu^*
  \;=\; a_i \;=\; \mu_0(\ph_i)
\eeqn
as intended.

Equation~\req{eq:mredexp} then tells us that the scaling
function, $\frac{d\mu^*}{d\prior^*}$, between $\prior^*$ and
$\mu^*$ is piecewise constant. In particular,
$\frac{d\mu^*}{d\prior^*}$ is constant across each of the sets
$\model(\psi_S)$ related to the sentences, $\psi_S$,
constructed in Proposition~\ref{prop:psiphi}.
\begin{align*}
  & \mu^*(\model(\ph)) = \int_{\model(\ph)} \frac{d\mu^*}{d\prior^*} \; d\prior^*
  = \sum_{S\subseteq\{1:n\}} \int_{\model(\ph \wedge \psi_S)} \frac{d\mu^*}{d\prior^*} \; d\prior^*  \\
  &= \sum_{S\subseteq\{1:n\}} \int_{\model(\ph \wedge \psi_S)} \frac{1}{\Phi(\lambda)}\mexp{\sum_{i=1}^n \lambda_i f_i(x)}\, d\prior^*(x)
     & \text{[Equation \req{eq:mredexp}]}\\
  &= \sum_{S\subseteq\{1:n\}} \frac{1}{\Phi(\lambda)}\mexp{\sum_{i=1}^n \lambda_i f_i(\model(\psi_S))}\, \prior^*(\model(\ph \wedge \psi_S))
     & \text{[$f_i$ constant on $\model(\psi_S)$]}\\
  &= \frac{1}{\Phi(\lambda)} \sum_{S\subseteq\{1:n\}} \mexp{\sum_{i \in S} \lambda_i} \prior(\ph \wedge \psi_S)
     & \text{[$f_i=1$ iff $i\in S$]}\\
\end{align*}
This leads to the following definition for $\hat\mu$:

\begin{definition}[minimally more informative probability]\label{def:mmimu}
Let $\prior$ be an arbitrary probability on sentences,
and $\mu_0:\{\ph_1,...,\ph_n\}\to[0,1]$ constrain the probability $\hat\mu$ of
the sentences $\ph_1,...,\ph_n$. Let
\begin{align*}
  \hat\mu(\ph) & \;:=\; \sum_{S\subseteq\{1:n\}} w_S\, \prior(\ph \wedge \psi_S)
  & [\text{Defining equation}] \\
  w_S & \;:=\; \frac{1}{\Phi(\lambda)}\mexp{\sum_{j \in S} \lambda_j}
  & [\text{Weights}] \\
  \Phi(\lambda) & \;:=\; \sum_{S\subseteq\{1:n\}} \mexp{\sum_{j \in S} \lambda_j} \prior(\psi_S)
  & [\text{Normalizing constant}]\\
  \mu_0(\ph_i) & \;=\; \sum_{S\subseteq\{1:n\}} w_S\, \prior(\ph_i \wedge \psi_S)
  \;\equiv\; \sum_{S\ni i} w_S\prior(\psi_S)
  & \smash{\left[{\text{Consistency equations}\atop
     \text{for $\lambda_i\in\SetR\cup\{-\infty\}$}}\right]}
\end{align*}
if the expressions are well-defined and a solution exists.
Otherwise $\hat\mu$ is undefined. We call $\hat\mu$ {\em
minimally more informative} than $\prior$ given $\mu_0$ (if it
exists).
\end{definition}

For $\ph=\psi_{S'}$, only the term $S=S'$ contributes to the
defining equations, which gives the useful relation
$\hat\mu(\psi_{S'})=w_{S'}\,\prior(\psi_{S'})$. So indeed,
$w_S=\hat\mu(\psi_S)/\prior(\psi_S)$ is the local scaling factor.
Inserting this back into the defining equation, gives
\beq\label{eq:mremup}
  \hat\mu(\ph) \;=\; \sum_{S:\prior(\psi_S)>0} \hat\mu(\psi_S) \prior(\ph|\psi_S).
\eeq
This also implies that if $\prior$ is Gaifman, then
$\hat\mu(\ph)$ is Gaifman. Furthermore, $\hat\mu(\ph)>0$
whenever consistently with $\mu_0$ possible and $\xi(\ph)>0$,
i.e.\ for (strongly) Cournot $\xi$, $\hat\mu$ is as
``Cournot'' as possible.

\begin{proposition}[minimally more informative probability]\label{prop:mmimu}
If $\mu_0$ can be extended to a probability on $\S$, %
and prior $\prior(\psi_{n,S})>0$ for all satisfiable $\psi_{n,S}$, %
then $\hat\mu$ in Definition~\ref{def:mmimu} is the unique
minimum of the relative entropy w.r.t.\ $\prior$ under the constraints
$\hat\mu(\ph_i)=\mu_0(\ph_i)$, $i=1,...,n$:
\begin{align*}
  \min_{\mu:\mu(\ph_i)=\mu_0(\ph_i),i=1..n} \{\KL(\mu||\prior)\}
  & \;=\; \KL(\hat\mu||\prior) \\
  \;=\; \sum_{S\subseteq\{1:n\}}\hat\mu(\psi_S)\log{\hat\mu(\psi_S)\over\prior(\psi_S)}
  & \;=\; \sum_{i=1}^n\lambda_i\mu_0(\ph_i)-\log \Phi(\lambda)
\end{align*}
\end{proposition}

\begin{proof}
A measure-theoretic proof can be based on the second definition
in Definition~\ref{def:relentropy} and Equation~\req{eq:mredexp}.
Here we give an elementary proof based on the first definition:
First note that the sum over $S$ is well defined and finite,
since $\prior(\psi_S)=0$ implies $\psi_S$ unsatisfiable implies
$\hat\mu(\psi_S)=0$ by Proposition~\ref{prop:PPS}.3. Therefore,
wherever necessary or convenient, we interpret sums as
being restricted to those $S$ for which $\psi_S$ is
satisfiable. We have
\begin{align*}
  \KL(\mu||\prior) & \;=  \sum_{S\subseteq\{1:n\}} \mu(\psi_{n,S})\log{\mu(\psi_{n,S})\over\prior(\psi_{n,S})} \\
  & \;+ \lim_{m\to\infty} \!\! \sum_{S\subseteq\{1:n\}} \!\!\! \mu(\psi_{n,S})\nq
    \sum_{T\subseteq\{n+1:m\}} \nq\mu(\psi_{m,S\cup T}|\psi_{n,S})
    \log{\mu(\psi_{m,S\cup T}|\psi_{n,S})\over\prior(\psi_{m,S\cup T}|\psi_{n,S})}
\end{align*}
By multiplying the first term with
$1=\sum_{T\subseteq\{n+1:m\}} \mu(\psi_{m,S\cup T}|\psi_{n,S})$
and elementary algebra one can easily verify that this
expression indeed reduces to the first one in
Definition~\ref{def:relentropy}.
Now we need to minimize this w.r.t.\ to $\mu$. The first term
involves a constrained minimization over the $2^n-1$
``parameters'' $\mu(\psi_S):S\subseteq\{1\!:\!n\}$. The second
term (for fixed $m$) involves a free minimization over the
$2^n(2^{m-n}-1)$ parameters $\mu(\psi_{m,S\cup
T}|\psi_{n,S}):T\subseteq\{n\!+\!1\!:\!m\},S\subseteq\{1\!:\!n\}$.
Since the two parameter sets are independent, we can minimize
both terms separately.
Since there are no constraints for the second
minimization, and the second term is monotone increasing in
$m$, the unique solution is obviously $\mu(\psi_{m,S\cup
T}|\psi_{n,S})=\prior(\psi_{m,S\cup T}|\psi_{n,S})$.
The first term, since $\prior(\psi_{n,S})>0$ and the relative
entropy is non-negative and continuous and strictly convex and
the domain is finite-dimensional convex and compact (a $2^n-1$
dimensional probability simplex), it has a unique minimum on
the convex subspace generated by the linear constraints. With
Lagrange multipliers and differentiation one can derive the
consistency equations in Definition~\ref{def:mmimu}, which
uniquely determine the solution (this follows the same line of
reasoning as after Definition~\ref{def:relentropy}, but now in
finite sample spaces this is elementary). \qed
\end{proof}

The next section will develop necessary and sufficient
conditions under which $\mu_0$ can be extended to some $\mu$
and hence a minimally more informative $\mu$.

\section{Extension of Probabilities}\label{sec:Extension}

Maintaining consistency in large knowledge bases is a
non-trivial problem. Its probabilistic cousin studied in this
section is no easier: Given some probabilistic knowledge, does
this correspond to a coherent set of probabilistic beliefs?

More formally, suppose a finite set of sentences are given
pre-determined probabilities. An interesting, and practically
important, question is: what are necessary and sufficient
conditions for the existence of a probability on sentences that
gives precisely these probabilities on the finite set of
sentences?
The next result answers this question.

\begin{proposition}[extension of probabilities]\label{prop:ExtProb}\label{prop:mmiGmu}
Let the alphabet be countable alphabet, $\{\ph_1, ..., \ph_n\}$
be a finite set of sentences, and $\mu_0:\{\ph_1, ...,
\ph_n\}\to[0,1]$ a function. For each $S\subseteq\{1\!:\!n\}$,
let
\beqn
  \psi_S :=(\bigwedge_{i \in S} \ph_i) \wedge (\bigwedge_{j \in \{1:n\} \setminus S} \neg \ph_j).
\eeqn
Then $\mu_0$ can be extended to a (Gaifman) probability $\mu:\S\to\SetR$
{\em iff} the following set of equations for the $2^n$ variables $\a_S$, for $S\subseteq\{1\!:\!n\}$, has a solution:
\begin{gather*}
\sum_{S\subseteq\{1:n\}} \a_S = 1 \\*
\sum_{S\subseteq\{1:n\}: i \in S} \a_S = \mu_0(\ph_i),  \text{ for } i=1, ..., n \\*
\a_S \geq 0, \text{ for } S\subseteq\{1\!:\!n\} \\*
\a_S = 0 \text{ if } \psi_S \text{ has no (separating) model, for } S\subseteq\{1\!:\!n\}.
\end{gather*}
\end{proposition}

If the above conditions on $\a_S$ are met, then
Proposition~\ref{prop:mmimu} and the remark before it imply
that $\mu_0$ can in particular be extended to a probability
$\hat\mu$ that is minimally more informative than some prior
$\prior$, and $\mu$ is Gaifman if $\prior$ is.

\begin{proof}
{\boldmath{$(\Rightarrow)$}}
Suppose first that $\mu_0$ can be extended to a probability
$\mu:\S\to\SetR$.
We show that the set of equations has a solution.

Define $\a_S = \mu(\psi_S)$, for each $S\subseteq\{1\!:\!n\}$.
Since the $\psi_S$'s are pairwise disjoint, by the definition
of a probability, Proposition~\ref{prop:PPS}.6, and
Proposition~\ref{prop:psiphi}.2, $\sum_{S\subseteq\{1:n\}} \a_S
= 1$. Also $\sum_{S\subseteq\{1:n\}: i \in S} \a_S = \mu(\ph_i)
= \mu_0(\ph_i)$, by Propositions~\ref{prop:psiphi}.3
and~\ref{prop:PPS}.6. Since $\mu$ is a probability, $\a_S \geq
0$ for $S\subseteq\{1\!:\!n\}$. Finally, $\a_S=0$ if $\psi_S$
is unsatisfiable for $S\subseteq\{1\!:\!n\}$, by
Proposition~\ref{prop:PPS}.3.
(In case $\mu$ is Gaifman, we use $\widehat\mu^*$ of
Proposition~\ref{prop:mutomustar} to show that $\a_S =
\mu(\psi_S)=\widehat\mu^*(\sepmodel(\psi_S))=\widehat\mu^*(\emptyset)=0$
if $\psi_S$ has no separating model.)

{\boldmath{$(\Leftarrow)$}}
Conversely, suppose that the equations have a solution.
Let $\prior$ be a strongly Cournot probability on $\S$ (whose existence is
given by Proposition~\ref{thm:SC}).
Put
\beqn
  \Sat = \{ S\subseteq\{1\!:\!n\} \; | \; \psi_S \text{ is satisfiable} \}.
\eeqn
Define $\mu:\S\to\SetR$ by
\beq\label{eq:phiexp}
  \mu(\ph) \;:=\; \sum_{S\in\Sat} \a_S \, \prior(\ph | \psi_S)
  \;=\; \sum_{S\in\Sat} w_S \, \prior(\ph\wedge\psi_S)
\eeq
for $\ph \in \S$, where $w_S:=\a_S/\prior(\psi_S)$ for $S\in\Sat$.
The function $\mu$ is well-defined, since $\prior(\psi_S) > 0$, if $\psi_S$ is satisfiable.
We claim that $\mu$ is a probability on sentences.
Clearly, $\mu$ is non-negative.

Suppose that $\ph$ is valid.
Then
\begin{align*}
  & \; \mu(\ph) \\
= & \; \textstyle\sum_{S\in\Sat} \a_S \, \prior(\ph | \psi_S) \\
= & \; \textstyle\sum_{S\in\Sat} \a_S
                  & \text{[$\ph$ is valid and $\xi(\cdot|\psi_S)$ is a probability]} \\
= & \; 1. & \text{[$\a_S=0$ for $S\not\in\Sat$]}
\end{align*}

Suppose that $\neg(\ph \wedge \psi)$ is valid. Then
\begin{align*}
  & \; \mu(\ph \vee \psi)  \\
= & \; \textstyle\sum_{S\in\Sat} w_S \, \prior((\ph \vee \psi) \wedge \psi_S) & \text{[Equation~\req{eq:phiexp}]} \\
= & \; \textstyle\sum_{S\in\Sat} w_S \, \prior((\ph \wedge \psi_S) \vee (\psi \wedge \psi_S)) \\
= & \; \textstyle\sum_{S\in\Sat} w_S \, [\prior(\ph \wedge \psi_S) + \prior(\psi \wedge \psi_S)]
  & \text{[$\neg((\ph \wedge \psi_S) \wedge (\psi \wedge \psi_S))$ valid]} \\
= & \; \mu(\ph) + \mu(\psi).
\end{align*}

Thus $\mu$ is a probability on sentences.

Finally, $\mu$ extends $\mu_0$:
\beqn
  \mu(\ph_i)
\;=\; \sum_{S\in\Sat} \a_S \, \prior(\ph_i | \psi_S)
\;=\; \sum_{S\in\Sat: i \in S} \a_S
\;=\; \sum_{S\subseteq\{1:n\}: i \in S} \a_S
\;=\; \mu_0(\ph_i)
\eeqn
for $i=1, ..., n$, which completes the proof.
(To proof that $\mu$ is Gaifman, simply %
replace `is satisfiable' by `has a separating model' in particular in $\Sat$, and
`$\xi$ strongly Cournot' by `$\xi$ Cournot and Gaifman' %
in the above proof.)
\qed
\end{proof}

Next we study conditions on the set of sentences which guarantee that the
equations of Proposition~\ref{prop:ExtProb} have a solution.
First, a necessary condition is introduced.

\begin{definition}[subadditive $\mu_0$]\label{def:subadditive}
Let $\{ \ph_1, ..., \ph_n \}$ be a finite set of sentences and
$\mu_0:\{ \ph_1, ..., \ph_n \}\to[0,1]$ a function.
Then $\mu_0$ is {\em subadditive}\index{subadditive function on sentences}
if, for each
$i, i_1, ..., i_k \in \{ 1, ..., n \}$ such that the sentences
$\ph_{i_1}, ..., \ph_{i_k}$ are pairwise disjoint
and $\bigvee_{j = 1}^k \ph_{i_j} \rightarrow \ph_i$ is valid,
\begin{align*}
  \sum_{j = 1}^k \mu_0(\ph_{i_j}) &\leq \mu_0(\ph_i) & \text{and}
\\
  \sum_{j = 1}^k \mu_0(\ph_{i_j}) &= \mu_0(\ph_i)
  \text{ \; if additionally \; } \ph_i \rightarrow \bigvee_{j = 1}^k \ph_{i_j} \text{ is valid}.
\end{align*}
\end{definition}

Here is another necessary condition that will be needed.

\begin{definition}[eligible $\mu_0$]\label{def:eligible}
Let $\{ \ph_1, ..., \ph_n \}$ be a finite set of sentences and
$\mu_0:\{ \ph_1, ..., \ph_n \}\to[0,1]$ a function.
Then $\mu_0$ is {\em eligible}\index{eligible function} if, for each $i=1, ..., n$,
$\mu_0(\ph_i) = 0$ if $\ph_i$ is unsatisfiable.
\end{definition}

Now the conditions of subadditivity and eligibility are shown to be necessary.

\begin{proposition}[subadditive and eligible $\mu_0$]\label{prop:SAandE}
Let $\{ \ph_1, ..., \ph_n \}$ be a finite set of sentences and
$\mu_0:\{ \ph_1, ..., \ph_n \}\to[0,1]$ a function.
Suppose that $\mu_0$ can be extended to a probability on $\S$.
Then $\mu_0$ is subadditive and eligible.
\end{proposition}

\begin{proof}
Let $\mu:\S\to\SetR$ be a probability that extends $\mu_0$.

Suppose that, for some $i, i_1, ..., i_k \in \{ 1, ..., n \}$, the sentences
$\ph_{i_1}, ..., \ph_{i_k}$ are pairwise disjoint
and $\bigvee_{j = 1}^k \ph_{i_j} \rightarrow \ph_i$ is valid.
Then
\begin{align*}
     & \; \textstyle\sum_{j=1}^k \mu_0(\ph_{i_j}) \\
   = & \; \textstyle\sum_{j=1}^k \mu(\ph_{i_j}) & \text{[$\mu$ extends $\mu_0$]} \\
   = & \; \textstyle\mu(\bigvee_{j=1}^k \ph_{i_j}) & \text{[Proposition~\ref{prop:PPS}.6]} \\
\leq & \; \textstyle\mu(\ph_i)  & \text{[Proposition~\ref{prop:PPS}.4]} \\
   = & \; \mu_0(\ph_i).
\end{align*}

Also
\begin{align*}
             & \; \textstyle\ph_i \rightarrow \bigvee_{j=1}^k \ph_{i_j} \text{ is valid} \\
\mathit{implies} & \;  \textstyle\mu(\bigvee_{j=1}^k \ph_{i_j}) = \mu(\ph_i)
                          & \text{[$\textstyle\bigvee_{j = 1}^k \ph_{i_j} \rightarrow \ph_i$ is valid]} \\
\mathit{implies} & \; \textstyle\sum_{j=1}^k \mu(\ph_{i_j}) = \mu(\ph_i) \\
\mathit{implies} & \; \textstyle\sum_{j=1}^k \mu_0(\ph_{i_j}) = \mu_0(\ph_i).
\end{align*}
Thus $\mu_0$ is subadditive.

For $i \in \{ 1, ..., n\}$, $\mu(\ph_i) = 0$ if $\ph_i$ is unsatisfiable,
since $\mu$ is a probability; and $\mu_0(\ph_i) = \mu(\ph_i)$.
Hence $\mu_0$ is eligible.
\qed
\end{proof}

Now a further structural condition on the set of sentences is
introduced that, together with subaddivity and eligibility,
will be sufficient to guarantee that there is a solution of the
equations.

\begin{definition}[hierarchical sentences]\label{def:HierarchicalS}
A finite set of sentences $\{ \ph_1, ..., \ph_n \}$
is {\em hierarchical}\index{hierarchical sentences} if, for
each $i \neq j$, exactly one of the following holds: $\neg
(\ph_i \wedge \ph_j)$ is valid or $\ph_i
\rightarrow \ph_j$ is valid or $\ph_j \rightarrow
\ph_i$ is valid.
\end{definition}

Intuitively, Definition~\ref{def:HierarchicalS}
states that, if $\ph_i$ and $\ph_j$ $(i \neq j)$ are
sentences, then either they are disjoint or one of them is
stronger than the other. An hierarchical set of
sentences is illustrated in
Figure~\ref{fig:HierarchicalS}. Each circle or oval
indicates the set of models of a particular sentence.

\begin{figure}[htbp]
\begin{center}
\framebox[4.5in]{
\footnotesize\unitlength=1pt
\begin{picture}(100, 170)(0, 0)
\put(80,140){\oval(110,40)}
\put(90,140){\circle{30}}
\put(50,140){\circle{35}}
\put(55,135){\circle{20}}
\put(90,145){\circle{15}}

\put(0,60){\oval(50,120)}
\put(5,90){\circle{30}}
\put(0,50){\circle{40}}

\put(-70,60){\circle{25}}

\put(60,40){\circle{20}}

\put(120,70){\circle{40}}
\put(115,75){\circle{7}}
\put(125,65){\circle{9}}
\end{picture}
}
\end{center}
\caption{An hierarchical set of sentences}
\label{fig:HierarchicalS}
\end{figure}
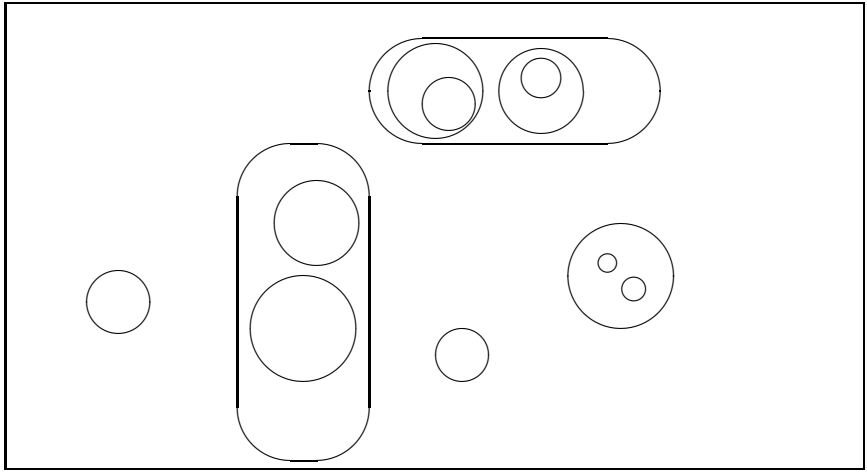

For the next result, the proof is by induction on the depth of an hierarchical set of sentences;
we now define the concept of depth.

\begin{definition}[depth of a sentence]
Let $\mathcal{H}$ be an hierarchical set of sentences. The {\em
depth}\index{depth of a sentence} of $\ph \in \mathcal{H}$
is defined to be the length $p$ of the unique sequence
$\ph_1, ..., \ph_p \equiv \ph$ of sentences in $\mathcal{H}$ such that
(a) $\ph_{i+1} \rightarrow \ph_i$ is valid, for $i=1, \ldots, p-1$;
(b) for each $\psi \in \mathcal{H}$, $\ph_{i+1} \rightarrow \psi$ and $\psi \rightarrow \ph_i$ are valid, for some some $i$,
implies $\psi = \ph_{i+1}$ or $\psi = \ph_i$;
and (c) for each $\psi \in \mathcal{H}$,
$\ph_{1} \rightarrow \psi$ is valid implies $\psi = \ph_1$.

The {\em depth}\index{depth of a set of sentences} of
$\mathcal{H}$ is the maximum depth of its sentences.
\end{definition}

An empty set of sentences has depth 0. The depth of the set of
sentences in Figure~\ref{fig:HierarchicalS} is 3.

\begin{proposition}[extending hierarchical constraints]\label{prop:ERC}
Let the alphabet be countable, $\{ \ph_1, ..., \ph_n \}$ a set
of sentences, and $\mu_0:\{ \ph_1, ..., \ph_n \}\to[0,1]$ a
subadditive eligible function. Suppose that $\{ \ph_1, ...,
\ph_n \}$ is hierarchical. Then $\mu_0$ can be extended to a
minimally more informative probability
$\mu:\S\to\SetR$ than some prior $\xi$ (see Definition~\ref{def:mmimu}),
which is Gaifman if $\xi$ is.
\end{proposition}

\begin{proof}
The proof is by induction on the depth $d$ of the hierarchical set of sentences.

Suppose first that $d = 0$, that is, the set of sentences is
empty. To show that $\mu_0$ can be extended to a probability
$\mu:\S\to\SetR$, it suffices by
Proposition~\ref{prop:ExtProb} to show that the equations of
that proposition for this case have a solution. Since the index
set of the set of sentences is empty, its only subset is $S =
\emptyset$. Furthermore, $\psi_S$ is $\top$. Put $\a_S = 1$.
Then the first equation from Proposition~\ref{prop:ExtProb} is trivially satisfied.
The second set of equations does not appear in this case.
Finally, the third and fourth equations are trivially satisfied.
This completes the base case of the induction argument.

Now suppose the result holds for hierarchical sets of sentences having depth $d$.
Let $\{ \ph_1, ..., \ph_n \}$ be an hierarchical set of sentences with depth $d+1$.
Without loss of generality, we can assume that $\{ \ph_1, ..., \ph_p \}$, for $p < n$,
is an hierarchical set of sentences of depth $d$ and the sentences
$\ph_{p+1}, ..., \ph_n$ all have depth $d+1$.
By the induction hypothesis, $\mu_0$ restricted to $\{ \ph_1, ..., \ph_p \}$
can be extended to a probability $\mu:\S\to\SetR$.
Thus, by Proposition~\ref{prop:ExtProb}, the following set of equations has a solution:
\begin{align*}
 & \textstyle \sum_{S\subseteq\{1:p\}} \a_S = 1 \\*
 & \textstyle \sum_{S\subseteq\{1:p\}: i \in S} \a_S = \mu_0(\ph_i),  \text{ for } i=1, ..., p \\*
 & \textstyle \a_S \geq 0, \text{ for } S\subseteq\{1\!:\!p\} \\*
 & \textstyle \a_S = 0 \text{ if } \psi_S \text{ is unsatisfiable}, \text{ for } S\subseteq\{1\!:\!p\}.
\end{align*}

Consider a typical sentence $\ph_i$ of depth $d$ that `contains' sentences
$\ph_{i_1}, ..., \ph_{i_k}$ of depth $d+1$.
Thus $\ph_{i_1}, ..., \ph_{i_k}$ are pairwise disjoint and
$\bigvee_{j=1}^k \ph_{i_j} \rightarrow \ph_i$ is valid.
(See Figure~\ref{depth_d+1}.)
Since $\mu_0$ is subadditive,
\begin{align*}
  \textstyle\sum_{j = 1}^k \mu_0(\ph_{i_j}) &\leq \mu_0(\ph_i) & \text{and} \\
  \textstyle\sum_{j = 1}^k \mu_0(\ph_{i_j}) &= \textstyle \mu_0(\ph_i)
              \text{ \; if also \; } \ph_i \rightarrow \bigvee_{j = 1}^k \ph_{i_j} \text{ is valid}.
\end{align*}
Since $\mu_0$ is eligible, $\mu_0(\ph_{i_j}) = 0$ if $\ph_{i_j}$ is unsatisfiable,
for $j = 1, ..., k$.
It has to be shown that when the depth $d+1$ sentences are added to
$\ph_1, ..., \ph_p$, the corresponding set of equations has a solution.

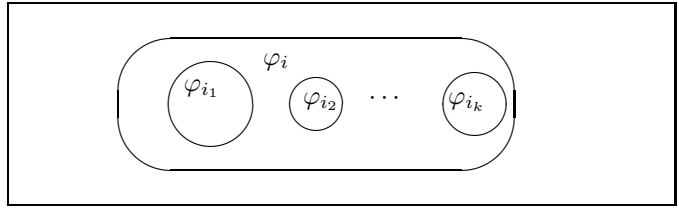
\begin{figure}[htbp]
\begin{center}
\framebox[3.5in]{
\footnotesize\unitlength=1pt
\begin{picture}(100, 70)(-40, 25)
\put(0,60){\oval(150,50)}
\put(-20,75){$\ph_i$}
\put(-40,60){\circle{30}}
\put(-50,65){$\ph_{i_1}$}
\put(0,60){\circle{20}}
\put(-5,60){$\ph_{i_2}$}
\put(20,60){$\cdots$}
\put(60,60){\circle{25}}
\put(50,60){$\ph_{i_k}$}
\end{picture}
}
\end{center}
\caption{Pairwise disjoint sentences $\ph_{i_1}, ..., \ph_{i_k}$ of depth $d+1$}
\label{depth_d+1}
\end{figure}

To simplify the notation, assume for the moment that
$\ph_{i_1}, ..., \ph_{i_k}$ are {\em all} the sentences of
depth $d+1$, so that the index set $\{1, ..., p\}$ for the
set of at-most-depth $d$ sentences is expanded to $\{1, ...,
n\}$ for the whole set of sentences.

Let $S_i \subseteq \{ 1, ..., p \}$ be the set of indices of
sentences in the `path' down to $\ph_i$ in the set of
at-most-depth $d$ sentences, so that $\psi_{S_i} = \ph_i$ is
valid. Now consider the full set of sentences. Then the
following are valid:
\begin{gather*}
\psi_{S_i \cup \{i_1\}} = \ph_{i_1} \\*
\qquad \vdots \\*
\psi_{S_i \cup \{i_k\}} = \ph_{i_k} \\*
\psi_{S_i} = \ph_i \wedge \bigwedge_{i=1}^k \neg \ph_{i_j}.
\end{gather*}

Included in the equations for the full set of sentences are the following:
\begin{gather*}
\a_{S_i \cup \{i_1\}} = \mu_0(\ph_{i_1}) \\*
\qquad \vdots \\*
\a_{S_i \cup \{i_k\}} = \mu_0(\ph_{i_k}) \\*
\a_{S_i} + \a_{S_i \cup \{i_1\}} + \cdots + \a_{S_i \cup \{i_k\}} = \mu_0(\ph_i).
\end{gather*}
(The first $k$ equations are new ones; the last equation replaces
$\a_{S_i}  = \mu_0(\ph_i)$ in the set of equations for the at-most-depth $d$ sentences.)

Furthermore, the term $\a_{S_i}$ in the first equation of
the set of equations for the at-most-depth $d$ sentences is
replaced by $\a_{S_i} + \a_{S_i \cup \{i_1\}} + \cdots
+ \a_{S_i \cup \{i_k\}}$ in the equations for the full set
of sentences. (This is the only change to the first equation
because all the other extra subsets $R$ of $\{ 1, ..., n \}$
that have to be considered lead to $\psi_R$ that are logically
equivalent to $\bot$ and hence have $\a_R = 0$.)

Because $\mu_0$ is subadditive and eligible, it is clear that
\begin{gather*}
\a_S \geq 0, \text{ for } S\subseteq\{1\!:\!n\} \\*
\a_S = 0 \text{ if } \psi_S \text{ is unsatisfiable}, \text{ for } S\subseteq\{1\!:\!n\}
\end{gather*}
are satisfied.

Thus the set of equations for the full set of sentences has a solution.
The case when there are extra sentences of depth $d+1$ `inside'
other $\ph_j$ is handled in a similar way.

Now use Propositions~\ref{prop:mmimu} and \ref{prop:ExtProb} to
conclude that $\mu_0$ can be extended to a minimally more
informative probability $\mu:\S\to\SetR$ than some prior $\xi$.
This completes the induction argument. \qed
\end{proof}

\section{User Manual}\label{sec:Examples}

This section is a brief outlook on how (approximations of) the
theory developed in this paper might be used in autonomous
reasoning agents. We discuss the special case of certain
knowledge and how it can be used to make inferences about
statements that are not logical implications of the knowledge
base. For instance, if our agent has observed a large number of
ravens which are all black without exception, how strongly
should it belief in the hypothesis that ``all ravens are
black''?
We construct an agent that can learn in the limit in the usual
time-series forecasting setting with an observation sequence
indexed by natural numbers.

\paradot{Certain knowledge}
A common case of knowledge is a set of sentences $\ph_i$, each
having degree of belief 1 (that is, $\mu_0(\ph_i) = 1$, for $i=1, \ldots, n$). In other words, there is certainty that each
$\ph_i$ is valid in the intended  interpretation. This
corresponds to non-logical axioms in a theory. Let $\prior$ be a
Cournot probability and suppose that $\mu$ is minimally more
informative than $\prior$ given $\mu_0$. In this case, each
$\mu(\psi_S)$, for $S\subseteq\{1\!:\!n\}$, is uniquely
determined.

To see this, suppose that $S \neq \{1:n\}$, say, $i
\notin S$, Then $\mu(\psi_S) \leq \mu(\neg \ph_{i}) = 1 -
\mu(\ph_{i}) = 1 -\mu_0(\ph_{i}) = 0$, so that
$\mu(\psi_S) = 0$. Hence $\mu(\psi_{\{1:n\}}) = 1$.
Thus, in this situation, by \req{eq:mremup} $\mu$ satisfies
\beq\label{eq:mumuprel}
  \mu(\ph) = \prior(\ph \, | \, \ph_1 \wedge \cdots \wedge\ph_n),
\eeq
for $\ph \in \S$. Consequently,
there is no optimisation to be done: either $\ph_1\wedge \cdots \wedge\ph_n$ is satisfiable (leading directly to
the above definition for $\mu$) or else it is not, in which
case there are no solutions and $\mu$ cannot be defined at all.

A further special case beyond the one just considered is when
$\ph$ is a logical consequence of $\ph_1\wedge \cdots \wedge\ph_n$. In this case,
\beqn
  \mu(\ph)
  \;=\; \prior(\ph \, | \, \ph_1\wedge \cdots \wedge\ph_n)
  \;=\; \frac{\prior(\ph_1\wedge \cdots \wedge\ph_n)}{\prior(\ph_1\wedge \cdots \wedge\ph_n)}
  \;=\; 1,
\eeqn
as one would expect. Similarly when $\neg \ph$ is logical
consequence, then $\mu(\ph) = 0$.

Note that, while it is important that the prior $\prior$ be
Cournot, it is just as important that the posterior $\mu$ be allowed not to
be Cournot. The prior should be Cournot so that the KL
divergence is as widely defined as possible or, more
intuitively, to make sure sentences having a separating model are {\em not
forced} to have $\mu$-probability 0. On the other hand, the
probability $\mu$ should be {\em allowed} to be 0 on
sentences having a separating model since the evidence in the form of the
probabilities on $\ph_1, \dots, \ph_n$ may imply this. This is
apparent, for example, for the case where each $\ph_i$ has
probability 1: according to this evidence, any sentence (even one having a separating model)
that is disjoint from $\ph_1\wedge \cdots \wedge\ph_n$ must have $\mu$-probability 0.

\paradot{Black ravens}
Consider the infamous problem of the black ravens which is one
of the most notorious problems in confirmation theory
\cite{Earman:93,Hutter:11uiphil}. Let the ravens be identified
by positive integers and $B(i)$ denote the fact that raven $i$
is black. The evidence consists of the sentences $B(1), \ldots,
B(n)$. (Thus $\ph_i \equiv B(i)$, for $i=1, \ldots, n$.) Let
$\mu_0:\{ B(1), \ldots, B(n) \}\to[0,1]$ be defined
by $\mu_0(B(i)) = 1$, for $i=1, \ldots, n$. Thus the degree
of belief that the $i$th raven is black is 1, for $i=1,
\ldots, n$. Suppose that $\prior$ is an uninformative prior that
is Cournot and Gaifman. Since a-priori there are no constraints
(on $B$), this implies that $\prior(\forall i.B(i))>0$. Let $\mu$
be a probability that is minimally more informative than
$\prior$ given $\mu_0$. Thus $\mu$ is given by \req{eq:mumuprel}.

Now consider the sentence $\forall i. B(i)$. This is clearly
not a logical consequence of the evidence, but one can use
$\mu$ to ascribe a degree of belief that it is true and,
furthermore, investigate what happens to this probability as
the number of black ravens increases.
Equation~\req{eq:mumuprel} and $\mu_0(B(i)) = 1$, for $i=1,
\ldots,n$, and then Theorem~\ref{thm:CUH} applied to Gaifman
and Cournot $\xi$ show that
\beqn
  \mu(\forall i.B(i))
  \;=\; \prior(\forall i.B(i)\,|\,B(1)\wedge\cdots\wedge B(n))
  \;\stackrel{n\to\infty}\longrightarrow 1
\eeqn
Thus, as the number of observed black ravens increases, the degree
of belief that all ravens are black approaches 1.
Of course this also implies the weaker statement that our belief in the next raven being
black tends to one:
\beqn
  \prior(B(n+1) \, | \, B(1)\wedge \cdots \wedge B(n)) \quad\stackrel{n\to\infty}\longrightarrow\quad 1
\eeqn

\paradot{Naive black ravens}
Continuing the preceding example, suppose given the evidence
$B(1), \ldots, B(n)$, each having probability 1, one wants to
know the degree of belief for $B(n+1)$.
Consider the tree construction in Theorem~\ref{thm:char} for $\xi$ but
with sentences $\ph_1,\ph_2,...$ only ranging over
$\ph_i=B(i)$ and uniform $\a_{n,S}=2^{-n}$. Then
\begin{align*}
  & \; \prior(B(n+1) \, | \, B(1)\wedge \cdots \wedge B(n)) \\
= & \; \frac{\prior(B(1)\wedge \cdots \wedge B(n) \wedge B(n+1))}{\prior(B(1)\wedge \cdots \wedge B(n))} \\
= & \; \frac{\a_{n+1,\{1:n+1\}}}{\a_{n,\{1:n\}}}  \;=\; 1/2. 
\end{align*}
Thus, for this prior, knowing the evidence so far, even for
large $n$, does not give any information about $B(n+1)$. But it
gets worse: Assume $\xi$ is somehow extended to a probability
on all $\S$. Then for any $m\geq n$,
\beqn
  \prior(\forall i.B(i) \, | \, B(1)\wedge\cdots\wedge B(n))
  \;\leq\; \prior(B(1)\wedge \cdots \wedge B(m) \, | \, B(1)\wedge \cdots \wedge B(n))
  \;=\; (\fr12)^{m-n}
\eeqn
hence $\prior(\forall i.B(i) \, | \, B(1)\wedge\cdots\wedge B(n))\equiv 0$ for all $n$,
i.e.\ universal hypotheses can not be confirmed. Even more seriously,
we would be absolutely sure that non-black ravens exist
\beqn
  \prior(\exists i.\neg B(i) \, | \, B(1)\wedge\cdots\wedge B(n))\equiv 1
\eeqn
and no number of observed black ravens $n$ without any counter
examples will ever convince us otherwise.
These conclusions qualitatively hold even when
$\ph_1,\ph_2,...$ ranges over all or any subset of quantifier-free/lambda-free
sentences. There seem to be no simple local rules for choosing $\a_{n,S}$
that allow confirmation of all universal hypotheses.
This shows that it is crucial to include quantified sentences
when constructing a prior and ensure it is Cournot (even when
only making inferences about unquantified sentences like
$B(n+1)$).

\begin{corollary}[learning in the limit]\label{cor:LIL}
Let $\imath \equiv \Nat$, $\ph$ be a closed term of type $\Nat\to\tB$, %
$\mu$ be a Gaifman probability on sentences, %
and $\mu(\forall x.(\ph~x))>0$. Then
\beqn
  \lim_{n\to\infty} \mu(\forall x.(\ph~x) \,|\, (\ph~\tfn{0})\wedge \cdots \wedge(\ph~\tfn{n}))
  \;=\; 1
\eeqn
\end{corollary}

This generalizes the black raven example
and follows from Theorem~\ref{thm:CUH}. In particular,
learning in the limit is possible for the Gaifman and Cournot
probability constructed in the proof of
Theorem~\ref{thm:CandG}, provided $\forall x.(\ph~x)$ has a
separating model.

The proof crucially exploits that $\tfn{0},\tfn{1},\tfn{2},...$
are representatives of all terms of type $\Nat$. As discussed
in Example~\ref{ex:iota}, this would no longer be true had we
introduced a description operator into our logic.
Corollary~\ref{cor:LIL} would break down and universal
hypotheses over the natural numbers could not be inductively
confirmed, not even asymptotically.

\paradot{Approximations}
The construction of Cournot and Gaifman $\mu$
in the proof of Theorem~\ref{thm:CandG} required to determine
particular separating models for $\chi_i$ and to determine whether
they are also models of other sentences $\ph$.
This has been eased by Corollary~\ref{cor:GandC}, which only
requires determining wether sentences $\psi_{n,S}$ have (no)
separating model. Still this is non-decidable.

Assume we had some calculus to determining whether sentences
have (no) separating model. Even an asymptotic or approximate
or incomplete calculus may be of use.
Fix a sequence on-the-fly of all sentences
$\ph_2,\ph_3,...$ satisfying Theorem~\ref{thm:char}.4 (once and for all).
Determine the subsequence of all sentences
$\chi_1=\ph_{j_1},\chi_2=\ph_{j_2},...$ with separating models (on the fly).

In order to determine $\mu$ to accuracy $\eps>0$ for some
finite number of sentences $\{\ph_{i_1},...,\ph_{i_n}\}$ of
interest, we have to perform the tree construction ``only'' for
$\ph_1\in\{\chi_1,...,\chi_m\}$, where
$\sum_{i=m+1}^\infty<\eps$ and up to depth
$d=\max\{i_1,...,i_n\}$, i.e.\ determine finitely many cases.
If a new sentence $\ph_{i_{n+1}}$ of interest ``arrives'' or
higher precision is needed, $d$ respectively $m$ can be
increased appropriately (that's what was meant with
on-the-fly). It is important to expand the already existing
trees with assigned probabilities, rather than restarting the
procedure with a larger $d$, since this can lead to wrong
inductive limits if different choices are made every time.

\paradot{Work flow example for a simple inductive reasoning agent}
Below we present an example of a fictitious inductive reasoning agent.
It is fictitious, since many operations are incomputable.
In practice one needs to employ approximations at various steps.
How to do this is an open problem.

1. Assume the agent has been endowed with some background
knowledge e.g.\ about kinetics, colors, biology, birds, etc.
Its knowledge is represented in the form of a hierarchical
(Definition~\ref{def:HierarchicalS}) set of sentences
$\{\ph_1, \ldots,\ph_n\}$ that hold for sure ($\mu_0(\ph_i)=1$ for
some $i$) or with some probability $0<\mu_0(\ph_i)<1$ for the
other $i$. Our expressive higher-order logic provides a
convenient way of doing so \cite{lloyd-ng-JAAMAS}.

2. Assume $\mu_0$ is subadditive and eligible
(Definitions~\ref{def:subadditive} and~\ref{def:eligible}).
This may not be so easy to achieve, and is akin to the
general problem of maintaining consistent knowledge bases.

3. Next, use an approximation of a Gaifman and Cournot $\prior$
prior, e.g.\ as defined in the proof of
Theorems~\ref{thm:CandG} or Theorem~\ref{thm:charGandC} or
Corollary~\ref{cor:GandC} or and approximation thereof as
outlined above. The agent now constructs via
Definition~\ref{def:mmimu} the minimally more informative
probability $\mu$, which exists by Proposition~\ref{prop:ERC}
and is Gaifman by Proposition~\ref{prop:mmiGmu}
and the remark after Equation~\req{eq:mremup}.

4. Let $o_0,o_1,o_2,...$ be the agent's life-time sequence of past and future
observations of all kinds of objects, ravens and otherwise, all
it has/will ever observe, e.g.\ $o_n$ is what the agent sees $n$
seconds after it has been switched on.

5. Assume current time is $n$, and the agent needs to
hypothesize about the world to decide its next action, e.g.\
whether some observed regularity is ``real''. For instance,
``if observation at time $k$ is a raven, is it also black?''.
We can formalize this with a predicate $\ph$ of type
$\Nat\to\tB$ with the intended interpretation of
$(\ph~\tfn{k})$ as ``if observation at time $k$ is a raven, it is
black''.

6. Of course the answer to $(\ph~\tfn{0}),...,(\ph~\tfn{n})$ is
immediate, since $o_0,...,o_n$ have already been
observed. If they are all true, the agent may start to wonder
whether ``all ravens are black'', or formally, whether $\forall
x.(\ph~x)$ is true. Note that non-raven observations in the
sequence are allowed.

7. If the agent is equipped with our inductive reasoning
system, its degree of belief in this hypothesis is %
$\mu(\forall x.(\ph~x)|(\ph~\tfn{0})\wedge \cdots \wedge(\ph~\tfn{n}))$.

8. This result can be the basis for some decision process
maximizing some utilities resulting in an informed action.

Is the degree of belief derived in Step 7 and used in Step 8
reasonable? At least asymptotically Corollary~\ref{cor:LIL}
ensures that in the limit the agent's belief tends to 1, which
is very reasonable. So our system of inductive reasoning at
least passes this test. Most other inductive reasoning systems
have difficulties in getting this right \cite{Hutter:11uiphil}.

\paradot{The Monty Hall Problem}
The Monty Hall problem is based on a US game show.
A contestant is presented with three doors.  Behind one of the doors is a prize.  The other two doors have nothing behind them.  The contestant is asked to select a door.  After the contestant selects a door, but before that door is opened, the game host selects and opens one of the other two doors.  At this point the contestant is again asked to select their preferred door and will win whatever is behind this final selection.

It is expected that the host will not reveal the prize.  This constraint means that the host will always open a door to reveal nothing behind it.  This limits the contestant's second choice to either persisting with the door selected originally, or switching to the remaining door.  It is a known, if counterintuitive, result that the best strategy for the contestant is to switch doors.

Let $\imath \equiv {\it Door}$. We introduce the constants $D_1, D_2, D_3 : {\it Door}$ and
\begin{gather*}
{\it playerFirstSelection}, {\it hostSelection}, {\it prizeDoor} : {\it Door} \to\tB \\
{\it unique} : ({\it Door} \to\tB) \to\tB.
\end{gather*}
As we shall see, the function ${\it unique}$ is used to capture the constraint on the preceding three predicates that exactly
one door makes each of them true.
With those, we can now define a set of sentences:
\begin{align*}
  \ph_1 := \; & (\textit{unique} = \lambda p.\exists d.((p \; d) \land \forall x.((p\;x) \longrightarrow x = d))) \; \land \\
    & (\textit{unique} \; \textit{playerFirstSelection}) \land (\textit{unique} \; \textit{hostSelection}) \land (\textit{unique} \; \textit{prizeDoor}) \\
  \ph_2 := \; & (\textit{prizeDoor} \; d_1) \\
  \ph_3 := \; & (\textit{prizeDoor} \; d_2) \\
  \ph_4 := \; & (\textit{playerFirstSelection} \; d_1) \\
  \ph_5 := \; & (\textit{playerFirstSelection} \; d_2) \\
  \ph_6 := \; & \forall d. ((hostSelection \; d) \longrightarrow \left( \lnot (playerFirstSelection \; d) \land \lnot (prizeDoor \; d) \right)) \\
  \ph_7 := \; & (\textit{hostSelection} \; d_1) \\
  \ph_8 := \; & (\textit{hostSelection} \; d_2) \\
  \ph_9 := \; & \exists d. \left( (\textit{playerFirstSelection} \; d) \land (\textit{prizeDoor} \; d) \right)
\end{align*}

Selection of the correct prior is very important for this problem.  We require that the prior be symmetric in which door makes the $\textit{prizeDoor}$ predicate true, and which door makes the $\textit{playerFirstSelection}$ predicate true.  We also require that there be no correlation between the doors that make these predicates true in the prior.

We now perform the tree construction from Proposition~\ref{prop:psiphi} using the set of sentences above.
We leave out any branches that will have prior probability 0.
Because of the requirements that the prior be symmetric and uncorrelated, each of the leaf nodes of the following tree will have equal probability.
Note that these requirements mean that $\prior(\ph_2) = 1/3$ rather than $0.5$ as usual.

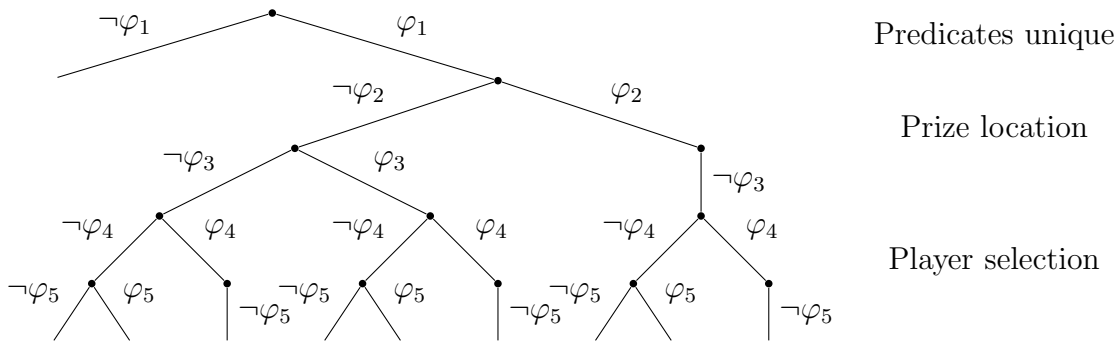
\begin{center}
\begin{tikzpicture}[scale=0.3,st/.style={circle,fill=black,inner sep=0pt,minimum size=1mm}]
  \node[st] (a) at (10,15) {};
  \node     (b) at (0,12) {};
  \node[st] (c) at (20,12) {};
  \node[st] (d) at (29,9) {};
  \node[st] (e) at (11,9) {};
  \node[st] (f) at (29,6) {};
  \node[st] (g) at (17,6) {};
  \node[st] (h) at (5,6) {};
  \node[st] (i) at (32,3) {};
  \node[st] (j) at (26,3) {};
  \node[st] (k) at (20,3) {};
  \node[st] (l) at (14,3) {};
  \node[st] (m) at (8,3) {};
  \node[st] (n) at (2,3) {};
  \node (o) at (32,0) {};
  \node (p) at (28,0) {};
  \node (q) at (24,0) {};
  \node (r) at (20,0) {};
  \node (s) at (16,0) {};
  \node (t) at (12,0) {};
  \node (u) at (8,0) {};
  \node (v) at (4,0) {};
  \node (w) at (0,0) {};

  \node at (42,14) {Predicates unique};
  \node at (42,10) {Prize location};
  \node at (42,4)  {Player selection};

  \draw[auto]
        (a) to node [swap] {$\lnot\ph_1$}      (b)
        (a) to node {$\ph_1$}                  (c)
        (c) to node {$\ph_2$}                  (d)
        (c) to node [swap] {$\lnot\ph_2$}      (e)
        (d) to node {$\lnot\ph_3$}             (f)
        (e) to node {$\ph_3$}                  (g)
        (e) to node [swap] {$\lnot\ph_3$}      (h)
        (f) to node {$\ph_4$}                  (i)
        (f) to node [swap] {$\lnot\ph_4$}      (j)
        (g) to node {$\ph_4$}                  (k)
        (g) to node [swap] {$\lnot\ph_4$}      (l)
        (h) to node {$\ph_4$}                  (m)
        (h) to node [swap] {$\lnot\ph_4$}      (n)
        (i) to node {$\lnot\ph_5$}             (o)
        (j) to node {$\ph_5$}                  (p)
        (j) to node [swap] {$\lnot\ph_5$}      (q)
        (k) to node {$\lnot\ph_5$}             (r)
        (l) to node {$\ph_5$}                  (s)
        (l) to node [swap] {$\lnot\ph_5$}      (t)
        (m) to node {$\lnot\ph_5$}             (u)
        (n) to node {$\ph_5$}                  (v)
        (n) to node [swap] {$\lnot\ph_5$}      (w);
\end{tikzpicture}
\end{center}

Assume, without loss of generality, that the prize is located behind door 1.
This allows us to zoom in onto the right-most sub-tree rooted at $\lnot\ph_3$ and add the Host door selection predicates.
These predicates are the host constraints, $\ph_6$, which we will require to be true with probability 1, and then predicates that elicit the host's selection.
As not all host selections are legal, some branches here have probability 0 (shown dashed).

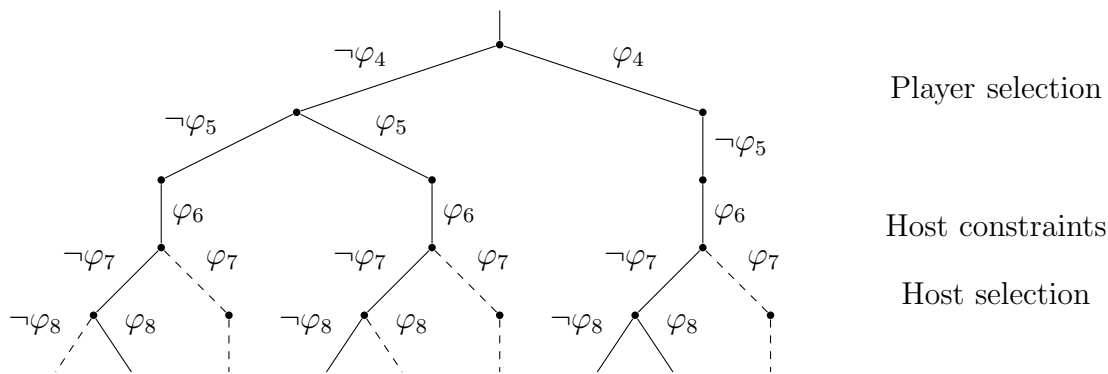
\begin{center}
\begin{tikzpicture}[scale=0.3,st/.style={circle,fill=black,inner sep=0pt,minimum size=1mm}]
  \node     (a) at (20,17) {};
  \node[st] (c) at (20,15) {};
  \node[st] (d) at (29,12) {};
  \node[st] (e) at (11,12) {};
  \node[st] (f) at (29,9) {};
  \node[st] (g) at (17,9) {};
  \node[st] (h) at (5,9) {};
  \node[st] (f') at (29,6) {};
  \node[st] (g') at (17,6) {};
  \node[st] (h') at (5,6) {};
  \node[st] (i) at (32,3) {};
  \node[st] (j) at (26,3) {};
  \node[st] (k) at (20,3) {};
  \node[st] (l) at (14,3) {};
  \node[st] (m) at (8,3) {};
  \node[st] (n) at (2,3) {};
  \node (o) at (32,0) {};
  \node (p) at (28,0) {};
  \node (q) at (24,0) {};
  \node (r) at (20,0) {};
  \node (s) at (16,0) {};
  \node (t) at (12,0) {};
  \node (u) at (8,0) {};
  \node (v) at (4,0) {};
  \node (w) at (0,0) {};

  \node at (42,13) {Player selection};
  \node at (42,7)  {Host constraints};
  \node at (42,4)  {Host selection};

  \draw (a) to (c);

  \draw[auto]
        (c) to node {$\ph_4$}                  (d)
        (c) to node [swap] {$\lnot\ph_4$}      (e)
        (d) to node {$\lnot\ph_5$}             (f)
        (e) to node {$\ph_5$}                  (g)
        (e) to node [swap] {$\lnot\ph_5$}      (h)
        (f) to node {$\ph_6$}                  (f')
        (g) to node {$\ph_6$}                  (g')
        (h) to node {$\ph_6$}                  (h')
        (f')  to node [swap] {$\lnot\ph_7$}      (j)
        (g')  to node [swap] {$\lnot\ph_7$}      (l)
        (h')  to node [swap] {$\lnot\ph_7$}      (n)
        (j) to node {$\ph_8$}                  (p)
        (j) to node [swap] {$\lnot\ph_8$}      (q)
        (l) to node [swap] {$\lnot\ph_8$}      (t)
        (n) to node {$\ph_8$}                  (v);
  \draw[auto,dashed]
        (f')  to node {$\ph_7$}              (i)
        (g')  to node {$\ph_7$}                  (k)
        (h')  to node {$\ph_7$}                  (m)
        (i) to                                 (o)
        (k) to                                 (r)
        (l) to node {$\ph_8$}                  (s)
        (m) to                                 (u)
        (n) to node [swap] {$\lnot\ph_8$}      (w);
\end{tikzpicture}
\end{center}

Each of the three major branches with non-zero probability has
equal prior probability.  Of these the left-hand two
($\lnot\ph_4$) each have the host forced to open one particular
door, and hence the remaining door has the prize -- the player
is better off swapping. Only on
the left hand branch when the player correctly guessed the
prize initially is the player better off not swapping, but this
is a less likely outcome than the other.  Hence the player is
better off swapping.

\section{Discussion}\label{Discussion}

A key goal of this research is that of integrating logic and
probability, a problem that has a history going back around 300
years and for which three main threads can be discerned. The
oldest by far is the philosophical/mathematical thread that can
be traced via Boole \cite{boole-laws,boole-studies} back to
Jacob Bernoulli in 1713. An extensive historical account of
this thread can be found in \cite{hailperin}; the idea of
putting probabilities on sentences goes back to before
\cite{los55} which contains references to even earlier
material; the important Gaifman condition appeared in
\cite{gaifman64} and was further developed in
\cite{gaifman-snir}; in \cite{scott-krauss} the theory is
developed for infinitary logic; overviews of more recent work
from a philosophical perspective can be found in
\cite{hajek,williamson02,williamson-philosophy-of-mathematics}.
The second thread is that of the knowledge representation and
reasoning community in artificial intelligence, of which
\cite{nilsson86,halpern90,fagin-halpern,halpern-uncertainty,shirazi-amir}
are typical works. The third thread is that of the machine
learning community in artificial intelligence, of which
\cite{muggleton95stochastic,DeRaedt-Kersting,milch05blog,richardson06,milch07,braz07thesis,kersting07blp,pfeffer07,Church-language}
are typical works.

An important and useful technical distinction that can be made
between these various approaches is that the combination of
logic and probability can be done externally or internally
\cite{williamson-philosophy-of-mathematics}: in the external
view, probabilities are attached to sentences in some logic; in
the internal view, sentences incorporate statements about
probability. One can even mix the two cases so that
probabilities appear both internally and externally.
We now examine each of these in turn.

\paradot{Probabilities inside sentences}
In the internal view, the uncertainty is modeled inside the
sentences of a theory.
For this to be possible, we must make a careful choice of
logic; in particular, first-order logic (alone) is not
expressive enough for this purpose.
There has been a tradition of extending first-order logic with
probabilistic extensions
\cite{halpern-uncertainty,hajek,williamson02}.
A good alternative approach, studied in
\cite{ng-lloyd-JAL,ng-lloyd-uther-AMAI}, is to simply adopt higher-order logic.
The most crucial property of higher-order
logic that we exploit is that it admits so-called higher-order
functions which take functions as arguments and/or return
functions as results. It is this property that allows the
modelling of, and reasoning about, probabilistic concepts
directly in higher-order theories.

\paradot{Probabilities outside sentences}
In contrast to the internal view, almost all other approaches
to integrating logic and probability model uncertainty by
putting probabilities outside sentences.
This natural idea has been taken up by many
researchers and has a large body of theoretical support.
Here we follow the lead of Gaifman and Snir for first-order
logic in \cite{gaifman-snir} (that builds on earlier work in
\cite{gaifman64}). They showed that, under certain conditions,
there is a probability on sentences that is strictly positive
on consistent sentences (that is, those that have a model).
This is an important property of any probability that is
intended to be used as a prior in Bayesian inference.
An accessible account of this material can be found in
\cite{paris-book}.

Other such systems, and there are now many of these, include Bayesian logic
programs \cite{kersting07blp}, Markov logic networks
\cite{richardson06}, and stochastic logic programs
\cite{muggleton95stochastic}. While the intention is usually
that the probabilities define (or at least constrain) a
distribution on the set of interpretations, some systems take
other approaches. For example, the probabilities can be used to
define a distribution on proofs or a distribution on programs.
For a taxonomy of such systems, see \cite{milch07}.

\paradot{Conclusion}
This paper provides much of the foundation for the design of an
integrated probabilistic reasoning system that can handle
probabilities both inside and outside sentences.
The main challenge for the future lies in the discovery of reasonable
approximation schemes for the different currently incomputable
aspects of the general theory.

\paradot{Acknowledgements}
The research was partly supported by the Australian Research Council Discovery Project DP0877635
``Foundations and Architectures for Agent Systems''.
NICTA is funded by the Australian Government as represented by
the Department of Broadband, Communications and the Digital
Economy and the Australian Research Council through the ICT
Centre of Excellence program.


\addcontentsline{toc}{section}{\refname}
\bibliographystyle{alpha}

\begin{thebibliography}{ABCD}\small

\bibitem[And02]{andrews2}
P.B. Andrews.
\newblock {\em An Introduction to Mathematical Logic and Type Theory: To Truth
  Through Proof}.
\newblock Kluwer Academic Publishers, second edition, 2002.

\bibitem[Boo54]{boole-laws}
G.~Boole.
\newblock {\em An Investigation of the Laws of Thought on which are founded the
  Mathematical Theories of Logic and Probabilities}.
\newblock Walton and Maberly, 1854.

\bibitem[Boo52]{boole-studies}
G.~Boole.
\newblock {\em Studies in Logic and Probability}.
\newblock Watts \& Co, 1952.

\bibitem[Chu40]{church}
A.~Church.
\newblock A formulation of the simple theory of types.
\newblock {\em Journal of Symbolic Logic}, 5:56--68, 1940.

\bibitem[Cou43]{Cournot:1843}
A.~A. Cournot.
\newblock {\em Exposition de la th{\'e}orie des chances et des
  probabilit{\'e}s}.
\newblock L. Hachette, Paris, 1843.

\bibitem[Csi75]{csiszar75}
I.~Csiszar.
\newblock I-divergence geometry of probability distributions and minimization
  problems.
\newblock {\em The Annals of Probability}, 3(1):146--158, 1975.

\bibitem[DK03]{DeRaedt-Kersting}
L.~{De Raedt} and K.~Kersting.
\newblock Probabilistic logic learning.
\newblock {\em SIGKDD Explorations}, 5(1):31--48, 2003.

\bibitem[Doo53]{Doob:53}
J.~L. Doob.
\newblock {\em Stochastic Processes}.
\newblock Wiley, New York, 1953.

\bibitem[dSB07]{braz07thesis}
R.~de~Salvo~Braz.
\newblock {\em Lifted First-Order Probabilistic Inference}.
\newblock PhD thesis, University of Illinois at Urbana-Champaign, 2007.

\bibitem[Dud02]{dudley}
R.M. Dudley.
\newblock {\em Real Analysis and Probability}.
\newblock Cambridge University Press, 2002.

\bibitem[Ear93]{Earman:93}
J.~Earman.
\newblock {\em Bayes or Bust? A Critical Examination of {B}ayesian Confirmation
  Theory}.
\newblock MIT Press, Cambridge, MA, 1993.

\bibitem[Far08]{farmer}
W.M. Farmer.
\newblock The seven virtues of simple type theory.
\newblock {\em Journal of Applied Logic}, 6(3):267--286, 2008.

\bibitem[FH94]{fagin-halpern}
R.~Fagin and J.Y. Halpern.
\newblock Reasoning about knowledge and probability.
\newblock {\em Journal of the ACM}, 41(2):340--367, 1994.

\bibitem[Fin73]{Fine:73}
T.~L. Fine.
\newblock {\em Theories of Probability}.
\newblock Academic Press, New York, 1973.

\bibitem[Gai64]{gaifman64}
H.~Gaifman.
\newblock Concerning measures in first order calculi.
\newblock {\em Israel Journal of Mathematics}, 2(1):1--18, 1964.

\bibitem[GMR{\etalchar{+}}08]{Church-language}
N.~D. Goodman, V.~K. Mansighka, D.~Roy, K.~Bonawitz, and J.~B. Tenenbaum.
\newblock Church: a language for generative models.
\newblock In {\em Uncertainty in Artificial Intelligence}, 2008.

\bibitem[GS82]{gaifman-snir}
H.~Gaifman and M.~Snir.
\newblock Probabilities over rich languages, testing and randomness.
\newblock {\em The Journal of Symbolic Logic}, 47(3):495--548, 1982.

\bibitem[Hai96]{hailperin}
T.~Hailperin.
\newblock {\em Sentential Probability Logic}.
\newblock Lehigh University Press, 1996.

\bibitem[H{\'{a}}j01]{hajek}
A.~H{\'{a}}jek.
\newblock Probability, logic and probability logic.
\newblock In L.~Goble, editor, {\em The {B}lackwell Guide to Philosophical
  Logic}, chapter~16, pages 362--384. Blackwell, 2001.

\bibitem[Hal90]{halpern90}
J.Y. Halpern.
\newblock An analysis of first-order logics of probability.
\newblock {\em Artificial Intelligence}, 46(3):311--350, 1990.

\bibitem[Hal03]{halpern-uncertainty}
J.Y. Halpern.
\newblock {\em Reasoning about Uncertainty}.
\newblock MIT Press, 2003.

\bibitem[Hen50]{henkin}
L.~Henkin.
\newblock Completeness in the theory of types.
\newblock {\em Journal of Symbolic Logic}, 15(2):81--91, 1950.

\bibitem[Iha93]{ihara93}
S.~Ihara.
\newblock {\em Information theory for continuous systems}.
\newblock World scientific publishing, 1993.

\bibitem[KD07]{kersting07blp}
K.~Kersting and L.~{De Raedt}.
\newblock Bayesian logic programming: Theory and tool.
\newblock In L.~Getoor and B.~Taskar, editors, {\em Introduction to Statistical
  Relational Learning}. MIT Press, 2007.

\bibitem[Lei94]{leivant}
D.~Leivant.
\newblock Higher-order logic.
\newblock In D.M. Gabbay, C.J. Hogger, J.A. Robinson, and J.~Siekmann, editors,
  {\em Handbook of Logic in Artificial Intelligence and Logic Programming},
  volume~2, pages 230--321. Oxford University Press, 1994.

\bibitem[Llo03]{LogicforLearning}
J.W. Lloyd.
\newblock {\em Logic for Learning: Learning Comprehensible Theories from
  Structured Data}.
\newblock Cognitive Technologies. Springer, 2003.

\bibitem[LN11]{lloyd-ng-JAAMAS}
J.W. Lloyd and K.S. Ng.
\newblock Declarative programming for agent applications.
\newblock {\em Autonomous Agents and Multi-Agent Systems}, 23(2):224--272,
  2011.
\newblock DOI: 10.1007/s10458-010-9138-1.

\bibitem[{\L}os55]{los55}
J.~{\L}os.
\newblock On the axiomatic treatment of probability.
\newblock {\em Colloquium Mathematicum}, 3:125--137, 1955.

\bibitem[MMR{\etalchar{+}}05]{milch05blog}
B.~Milch, B.~Marthi, S.~Russell, D.~Sontag, D.L. Ong, and A.~Kolobov.
\newblock {\sc Blog}: Probabilistic models with unknown objects.
\newblock In L.P. Kaelbling and A.~Saffiotti, editors, {\em Proceedings of the
  19th International Joint Conference on Artificial Intelligence}, pages
  1352--1359, 2005.

\bibitem[MR07]{milch07}
B.~Milch and S.~Russell.
\newblock First-order probabilistic languages: Into the unknown.
\newblock In S.~Muggleton, R.~Otero, and A.~Tamaddoni-Nezhad, editors, {\em
  Inductive Logic Programming: 16th International Conference, ILP 2006}, pages
  10--24. Springer, LNAI 4455, 2007.

\bibitem[Mug96]{muggleton95stochastic}
S.~Muggleton.
\newblock Stochastic logic programs.
\newblock In L.~De~Raedt, editor, {\em Advances in Inductive Logic
  Programming}, pages 254--264. IOS Press, 1996.

\bibitem[Nil86]{nilsson86}
N.J. Nilsson.
\newblock Probabilistic logic.
\newblock {\em Artificial Intelligence}, 28(1):71--88, 1986.

\bibitem[NL09]{ng-lloyd-JAL}
K.S. Ng and J.~W. Lloyd.
\newblock Probabilistic reasoning in a classical logic.
\newblock {\em Journal of Applied Logic}, 7(2):218--238, 2009.
\newblock DOI:10.1016/j.jal.2007.11.008.

\bibitem[NLU08]{ng-lloyd-uther-AMAI}
K.S. Ng, J.W. Lloyd, and W.T.B. Uther.
\newblock Probabilistic modelling, inference and learning using logical
  theories.
\newblock {\em Annals of Mathematics and Artificial Intelligence}, 54:159--205,
  2008.
\newblock DOI:10.1007/s10472-009-9136-7.

\bibitem[Par94]{paris-book}
J.B. Paris.
\newblock {\em The Uncertain Reasoner's Companion}, volume~39 of {\em Cambridge
  Tracts in Theoretical Computer Science}.
\newblock Cambridge University Press, 1994.

\bibitem[Pfe07]{pfeffer07}
A.~Pfeffer.
\newblock The design and implementation of {IBAL}: A general-purpose
  probabilistic language.
\newblock In Lise Getoor and Ben Taskar, editors, {\em Introduction to
  Statistical Relational Learning}, chapter~14. MIT Press, 2007.

\bibitem[RD06]{richardson06}
M.~Richardson and P.~Domingos.
\newblock Markov logic networks.
\newblock {\em Machine Learning}, 62:107--136, 2006.

\bibitem[RH11]{Hutter:11uiphil}
S.~Rathmanner and M.~Hutter.
\newblock A philosophical treatise of universal induction.
\newblock {\em Entropy}, 13(6):1076--1136, 2011.

\bibitem[SA07]{shirazi-amir}
A.~Shirazi and E.~Amir.
\newblock Probabilistic modal logic.
\newblock In R.C. Holte and A.~Howe, editors, {\em Proceedings of the 22nd AAAI
  Conference on Artificial Intelligence}, pages 489--495, 2007.

\bibitem[Sha01]{shapiroHOL}
S.~Shapiro.
\newblock Classical logic ii -- higher-order logic.
\newblock In L.~Goble, editor, {\em The Blackwell Guide to Philosophical
  Logic}, pages 33--54. Blackwell, 2001.

\bibitem[Sha06]{Shafer:06}
G.~Shafer.
\newblock Why did {C}ournot's principle disappear?, 19 May 2006.
\newblock Presentation. Ecole des Hautes Etudes en Sciences Sociales, Paris.
  Slides, URL: http://www.glennshafer.com/assets/downloads/disappear.pdf.

\bibitem[SK66]{scott-krauss}
D.~Scott and P.~Krauss.
\newblock Assigning probabilities to logical formula.
\newblock In J.~Hintikka and P.~Suppes, editors, {\em Aspects of Inductive
  Logic}, pages 219--264. North-Holland, 1966.

\bibitem[vBD83]{vanbenthem-doets}
J.~van Benthem and K.~Doets.
\newblock Higher-order logic.
\newblock In D.M. Gabbay and F.~Guenther, editors, {\em Handbook of
  Philosophical Logic}, volume~1, pages 275--330. Reidel, 1983.

\bibitem[Wil02]{williamson02}
J.~Williamson.
\newblock Probability logic.
\newblock In D.~Gabbay, R.~Johnson, H.J. Ohlbach, and J.~Woods, editors, {\em
  Handbook of the Logic of Inference and Argument: The Turn Toward the
  Practical}, volume~1 of {\em Studies in Logic and Practical Reasoning}, pages
  397--424. Elsevier, 2002.

\bibitem[Wil08a]{Williamson:08}
J.~Williamson.
\newblock Objective bayesian probabilistic logic.
\newblock {\em Journal of Algorithms}, 63(4):167--183, 2008.

\bibitem[Wil08b]{williamson-philosophy-of-mathematics}
J.~Williamson.
\newblock Philosophies of probability.
\newblock In A.~Irvine, editor, {\em Handbook of the Philosophy of Mathematics,
  Volume 4 of the Handbook of the Philosophy of Science}. Elsevier, 2008.
\newblock In press.
\end{thebibliography}

\newcommand{\etalchar}[1]{$^{#1}$}

\newpage\appendix
\section{List of Notation}\label{secApp}

\begin{tabbing}
  \hspace{2cm} \= \hspace{11cm} \= \kill
  $x,y,z$            \> variables                                                            \\[0.5ex]
  $t,r,s$            \> terms                                                                \\[0.5ex]
  $\alpha,\beta$     \> type of a term                                                       \\[0.5ex]
  $\tB$              \> type of the booleans                                                 \\[0.5ex]
  $\imath$           \> type of individuals                                                  \\[0.5ex]
  $\top$             \> Truth                                                                \\[0.5ex]
  $\bot$             \> Falsity                                                              \\[0.5ex]
  $\ph,\chi,\psi$    \> formula = term of type $\tB$, called sentence if closed              \\[0.5ex]
  $\S$               \> set of all sentences                                                 \\[0.5ex]
  $\I$               \> set of interpretations                                               \\[0.5ex]
  $\widehat\I$       \> set of separating interpretations                                    \\[0.5ex]
  $I$                \> interpretation                                                       \\[0.5ex]
  $\model{(\ph)}$ \> $\{I\in\I|\ph$ is valid in $I\}$ = set of models of $\ph$               \\[0.5ex]
  $\sepmodel{(\ph)}$ \> $\{I\in\widehat\I|\ph$ is valid in $I\}$ = set of separating models of $\ph$     \\[0.5ex]
  $\B$               \> Borel $\sigma$-algebra generated by $\{\model{(\ph)}|\ph\in\S\}$, if alphabet countable     \\[0.5ex]
  $\widehat\B$       \> Borel $\sigma$-algebra generated by $\{\sepmodel{(\ph)}|\ph\in\S\}$, if alphabet countable     \\[0.5ex]
  $\mu,(\hat\mu)$    \> (estimated) probability on sentences                                 \\[0.5ex]
  $\mu^*,(\widehat\mu^*)$ \> probability on sets of (separating) interpretations             \\[0.5ex]
  $i,j,k,n$          \> natural numbers used for indexing                                    \\[0.5ex]
  $\ph_1,\ph_2,...$ \> enumeration of some or all sentences                                  \\[0.5ex]
  $S$               \> $\subseteq\{1\!:\!n\} \equiv \{1,...,n\}$ = index of ``positive'' $\ph$ in ... \\[0.5ex]
  $\psi_S \equiv \psi_{n,S}$ \> $(\bigwedge_{i \in S} \ph_i) \wedge (\bigwedge_{j \in \{1:n\} \setminus S} \neg \ph_j)$ = hierarchical basis \\[0.5ex]
  $\a_{n,S}$        \> $\mu(\psi_{n,S})$ = base probabilities                                \\[0.5ex]
  $\prior$           \> prior probability  (usually Gaifman and Cournot)                     \\[0.5ex]
  $\mu_0(\ph_i)$    \> $:\{\ph_1,...,\ph_n\}\to[0,1]$ = constraints on $\mu$: $~\mu(\ph_i)=\mu_0(\ph_i)$ \\[0.5ex]
\end{tabbing}

\section{List of Definitions, Theorems, Examples, ...}

\theoremlisttype{allname}
{
\listtheorems{theorem,corollary,lemma,definition,example,proposition,principle,table,figure}
}

\end{document}